\documentclass[12pt]{article}%
\usepackage{amssymb}
\usepackage{amsfonts}
\usepackage{amsmath}
\usepackage{sw20elba}
\usepackage{hyperref}
\usepackage{version}
\usepackage[dvips]{graphicx}
\usepackage[section,above,below]{placeins}%
\providecommand{\U}[1]{\protect\rule{.1in}{.1in}}
\newtheorem{theorem}{Theorem}

\newtheorem{corollary}[theorem]{Corollary}

\newtheorem{definition}[theorem]{Definition}

\newtheorem{lemma}[theorem]{Lemma}

\newtheorem{proposition}[theorem]{Proposition}
\newtheorem{remark}[theorem]{Remark}

\newenvironment{proof}[1][Proof]{\textbf{#1.} }{\ \rule{0.5em}{0.5em}}
\begin{document}

\title{Dissipative Properties of Systems Composed of High-Loss and Lossless Components}
\author{Alexander Figotin\\University of California at Irvine
\and Aaron Welters\\Louisiana State University at Baton Rouge}
\maketitle

\begin{abstract}
We study here dissipative properties of systems composed of two components one
of which is highly lossy and the other is lossless. A principal result of our
studies is that all the eigenmodes\ of such a system split into two distinct
classes characterized as high-loss and low-loss. Interestingly, this splitting
is more pronounced the higher the loss of the lossy component. In addition,
the real frequencies of the high-loss eigenmodes can become very small and
even can vanish entirely, which is the case of overdamping.

\end{abstract}

\section{Introduction}

We introduce a general framework to study dissipative properties of two
component systems composed of a high-loss and lossless components. This
framework covers conceptually any dissipative physical system governed by a
linear evolution equation. Such systems include in particular damped
mechanical systems, electric networks or any general Lagrangian system with
losses accounted by the Rayleigh dissipative function, \cite[Sec. 10.11,
10.12]{Pars}, \cite[Sec. 8, 9, 46]{Gantmacher}.

An important motivation and guiding examples for our studies come from
two-component dielectric media composed of a high-loss and lossless
components. Any dielectric medium always absorbs a certain amount of
electromagnetic energy, a phenomenon which is often referred to as loss. When
it comes to the design of devices utilizing dielectric properties very often a
component which carries a useful property, for instance, magnetism is also
lossy. So the question stands: is it possible to design a composite
material/system which would have a desired property comparable with a
naturally occurring bulk substance but with significantly reduced losses. In a
search of such a low-loss composite it is appropriate to assume that the lossy
component, for instance magnetic, constitutes the significant fraction which
carries the desired property. But then it is far from clear whether a
significant loss reduction is achievable at all. It is quite remarkable that
the answer to the above question is affirmative, and an example of a simple
layered structure having magnetic properties comparable with a natural bulk
material but with 100 times lesser losses in wide frequency range is
constructed in \cite{FigVit8}. The primary goal of this paper is to find out
and explain when and how a two component system involving a high-loss
component can have low loss for a wide frequency range. The next question of
how this low loss performance for wide frequency range is combined with a
useful property associated with the lossy component is left for another
forthcoming paper.

A principal result of our studies here is that a two component system
involving a high-loss component can be significantly low loss in a wide
frequency range provided, to some surprise, that the lossy component is
sufficiently lossy. An explanation for this phenomenon is that if the lossy
part of the system has losses exceeding a critical value it goes into
essentially an overdamping regime, that is a regime with no oscillatory
motion. In fact, we think that for Lagrangian systems there will be exactly an
overdamping regime but these studies will also be conducted in already
mentioned forthcoming paper.

The rest of the paper is organized as follows. The model setup and discussion
of main results and their physical significance is provided in Section
\ref{smodel}. In the following Section \ref{secex} we apply the developed
general approach to an electric circuit showing all key features of the
method. \ Sections \ref{spsyo} and \ref{splhf} is devoted to a precise
formulation of all significant results in the form of Theorems, Propositions
and so on. Finally, in the last Section \ref{sprfr} we provide the proofs of
these results.

\section{Model setup and discussion of main results\label{smodel}}

A framework to our studies of energy dissipation is as in \cite{FigSch1} and
\cite{FigSch2}. This framework covers many physical systems including
dielectric, elastic and acoustic media. Our primary subject is a linear system
(medium) whose state is described by a time dependent \emph{generalized
velocity} $v\left(  t\right)  $ taking values in a Hilbert space $H$ with
scalar product $\left(  \cdot,\cdot\right)  $. The evolution of $v$ is
governed by a linear equation incorporating \emph{retarded friction}%
\begin{equation}
m\partial_{t}v\left(  t\right)  =-\mathrm{i}Av\left(  t\right)  -\int
_{0}^{\infty}a\left(  \tau\right)  v\left(  t-\tau\right)  \,d\tau+f\left(
t\right)  ,\ \label{mvt1}%
\end{equation}
where $m>0$ is a positive mass operator in $H$, $A$ is a self-adjoint operator
in $H$, $f\left(  t\right)  $ is a time dependent external \emph{generalized
force}, and $a\left(  t\right)  $, $t\geq0$, is an operator valued function
which we call the \emph{operator valued friction retardation function}
\cite[Section 1.6]{KubTod2}, or just the \emph{friction function}. The names
\textquotedblleft generalized velocity\textquotedblright\ and
\textquotedblleft generalized force\textquotedblright\ are justified when we
interpret the real part of the scalar product $\operatorname{Re}\left\{
\left(  v\left(  t\right)  ,f\left(  t\right)  \right)  \right\}  $ as the
work $\mathcal{W}$ done by $f\left(  t\right)  $ per unit of time at instant
$t$, that is%
\begin{equation}
\mathcal{W}=\int_{-\infty}^{\infty}\operatorname{Re}\left\{  \left(  v\left(
t\right)  ,f\left(  t\right)  \right)  \right\}  \,dt. \label{mvt2}%
\end{equation}
The \emph{internal energy of the system} is naturally defined as%
\begin{equation}
\text{internal energy }U=\frac{1}{2}\left(  v(t),mv(t)\right)  , \label{mvt2a}%
\end{equation}
and it readily follows from (\ref{mvt1}) that it satisfies the following
energy balance equation%
\begin{equation}
\frac{dU}{dt}=\operatorname{Re}\left\{  \left(  v\left(  t\right)  ,f\left(
t\right)  \right)  \right\}  -\operatorname{Re}\left\{  \left(  v\left(
t\right)  ,\int_{0}^{\infty}a\left(  \tau\right)  v\left(  t-\tau\right)
\,d\tau\right)  \right\}  . \label{mvt6}%
\end{equation}
The second term of the right hand side of (\ref{mvt6}) is interpreted as the
instantaneous rate of \textquotedblleft work done by the system on
itself,\textquotedblright\ or more properly the negative rate of energy
dissipation due to friction.

If we rescale the variables according to the formulas%
\begin{equation}
\tilde{v}=\sqrt{m}v,\ \Omega=\frac{1}{\sqrt{m}}A\frac{1}{\sqrt{m}},\ \tilde
{a}=\frac{1}{\sqrt{m}}a\frac{1}{\sqrt{m}},\ \tilde{f}=\frac{1}{\sqrt{m}}f,
\label{mvt3}%
\end{equation}
then the equation (\ref{mvt1}) reduces to the special form in which $m$ in the
new variables is the identity operator, i.e.%
\begin{equation}
\partial_{t}\tilde{v}\left(  t\right)  =-\mathrm{i}\Omega\tilde{v}\left(
t\right)  -\int_{0}^{\infty}\tilde{a}\left(  \tau\right)  \tilde{v}\left(
t-\tau\right)  \,d\tau+\tilde{f}\left(  t\right)  , \label{mvt5}%
\end{equation}
and in view of (\ref{mvt2a}) the internal energy $U$ associated with the state
$\tilde{v}$ turns into the scalar product, that is%
\begin{equation}
\text{internal energy }U=\frac{1}{2}\left(  \tilde{v}(t),\tilde{v}(t)\right)
. \label{mvt5a}%
\end{equation}
The system evolution equation of the special form (\ref{mvt5}) has two
important characteristic properties: (i) the operator $\Omega=\frac{1}%
{\sqrt{m}}A\frac{1}{\sqrt{m}}$ can be interpreted as the \emph{system
frequency operator} in $H$; (ii) the system internal energy is simply the
scalar product (\ref{mvt5a}). We refer to this important form (\ref{mvt5}) as
the \emph{system canonical evolution equation}.

For the sake of simplicity we assume the friction function to be
\textquotedblleft instantaneous,\textquotedblright\ that is $a\left(
t\right)  =\beta B\delta\left(  t\right)  $ with $B$ self-adjoint and
$\beta\geq0$ is a dimensionless \emph{loss parameter} which scales the
intensity of dissipation. \ Of course, such an idealized friction function is
a simplification that we take to avoid significant difficulties associated
with more realistic friction functions as in \cite{FigSch1}, \cite{FigSch2}.

Now, assuming that the rescaling (\ref{mvt3}) was applied, the canonical
evolution equation (\ref{mvt5}) takes the form%
\begin{equation}
\partial_{t}v\left(  t\right)  =-\mathrm{i}A\left(  \beta\right)  v\left(
t\right)  +f\left(  t\right)  ,\quad\text{where }A\left(  \beta\right)
=\Omega-\mathrm{i}\beta B,\ \beta\geq0. \label{mvt7}%
\end{equation}
Importantly, for the \emph{system operator}\textit{\ }$\Omega-\mathrm{i}\beta
B$ it is assumed the operator $B$ satisfies the \emph{power dissipation
condition}%
\begin{equation}
B\geq0. \label{mvt8}%
\end{equation}
The general energy balance equation (\ref{mvt6}) takes now a simpler form
\begin{equation}
\frac{dU}{dt}=\operatorname{Re}\left\{  \left(  v\left(  t\right)  ,f\left(
t\right)  \right)  \right\}  -W_{\text{\textrm{dis}}}. \label{mvt9}%
\end{equation}
where%
\begin{gather}
U=\text{system energy}=\frac{1}{2}\left(  v(t),v(t)\right)  ,\label{mvt9a}\\
W_{\text{\textrm{dis}}}=\text{system dissipated power}=\beta\left(  v\left(
t\right)  ,Bv\left(  t\right)  \,\right)  .\nonumber
\end{gather}
Most of the time it is assumed that the system governed by (\ref{mvt7})\ is at
rest for all negative times, i.e.%
\begin{equation}
v\left(  t\right)  =0,\ f\left(  t\right)  =0,\ t\leq0. \label{mvt7a}%
\end{equation}

To simplify further technical aspects of our studies and to avoid the
nontrivial subtleties involved in considering unbounded operators in
infinite-dimensional Hilbert spaces, we assume the Hilbert space $H$ to be of
a finite dimension $N$. Keeping in mind our motivations we associate the
operator $B$ with the lossy component of a composite and express the
significance of the lossy component in terms of the rank $N_{B}$ of the
operator $B$. Continuing this line of thought we introduce a space $H_{B}$,
the range of the operator $B$, and the corresponding orthogonal projection
$P_{B}$ on it, that is
\begin{equation}
H_{B}=\left\{  Bu:u\in H\right\}  ,\quad N_{B}=\dim H_{B}. \label{hbu1}%
\end{equation}
In what follows we refer to $H_{B}$ as a \emph{subspace of degrees of freedom
susceptible to losses or just loss subspace}. We also refer to the orthogonal
complement of $H_{B}^{\bot}=H\ominus H_{B}$ as the \emph{subspace of lossless
degrees of freedom or no-loss subspace}. The number $N_{B}$ plays an important
role in the Livsic theory of open systems where $N_{B}$ is called
\textquotedblleft the index of non-Hermiticity\textquotedblright\ of the
operator $\Omega-\mathrm{i}\beta B$, \cite[pp. 24, 27-28]{Liv}. The definition
of $P_{B}$ readily implies the following identity%
\begin{equation}
B=P_{B}BP_{B}. \label{hbu2}%
\end{equation}
We suppose the dimension $N_{B}$ to satisfy the following \emph{loss fraction
condition}%
\begin{equation}
0<\delta_{B}=\frac{N_{B}}{N}<1, \label{hbu3}%
\end{equation}
which signifies in a rough form that \emph{only a fraction }$\delta_{B}%
<1$\emph{\ of the degrees of freedom are susceptible to lossy behavior}. As we
will see later when the loss parameter $\beta\gg1$, only a fraction of the
system eigenmodes are associated with high losses and this fraction is exactly
$\delta_{B}$. For this reason we may refer to $\delta_{B}$ as the
\emph{fraction of high-loss eigenmodes}.

It turns out that the system dissipative behavior is qualitatively different
when the loss parameter $\beta$ is small or large. It seems that common
intuition about losses is associated with the small values of $\beta$. The
spectral analysis of the system operator $\Omega-\mathrm{i}\beta B$ for small
$\beta$ can be handled by the standard perturbation theory, \cite{Bau85}.
\ The results of our analysis are contained in Theorem \ref{Thm2} and may be
summarized as follows: \ Let $\omega_{j}$, $1\leq j\leq N$ denote the all
eigenvalues of the operator $\Omega$ repeated according to their
multiplicities then there exists a corresponding orthonormal basis of
eigenvectors $u_{j}$, $1\leq j\leq N$ such that if $0\leq\beta\ll1$ then the
operator $\Omega-\mathrm{i}\beta B$ is diagonalizable with a complete set of
$\zeta_{j}\left(  \beta\right)  $ eigenvalues and eigenvectors $v_{j}\left(
\beta\right)  $ having the expansions
\begin{equation}
\zeta_{j}\left(  \beta\right)  =\omega_{j}-\mathrm{i}\beta\left(  u_{j}%
,Bu_{j}\right)  +O\left(  \beta\right)  ,\text{ }\omega_{j}=\left(
u_{j},\Omega u_{j}\right)  ,\text{ }v_{j}\left(  \beta\right)  =u_{j}+O\left(
\beta\right)  ,\ 1\leq j\leq N,\text{ }\beta\ll1. \label{zej1}%
\end{equation}
The effect of small losses described by the above formula is well known, of
course, see, for instance, \cite[Sec. 46]{Gantmacher}.

The perturbation analysis of the system operator $\Omega-\mathrm{i}\beta B$
for large values of the loss parameter $\beta\gg1$ requires more efforts and
its results are quite surprising. It shows, in particular, that all the
eigenmodes split into two distinct classes according to their dissipative
behavior : \emph{high-loss and low-loss modes}. We refer to such a splitting
as \emph{modal dichotomy}.

In view of the above discussion we decompose the Hilbert space $H$ into the
direct sum of invariant subspaces of the operator $B\geq0$, that is,
\begin{equation}
H=H_{B}\oplus H_{B}^{\bot}, \label{hbcom1}%
\end{equation}
where $H_{B}=\operatorname{ran}B$ is the loss subspace of dimension $N_{B}$
with orthogonal projection $P_{B}$ and its orthogonal complement, $H_{B}%
^{\bot}=\ker B$, is the no-loss subspace of dimension $N-N_{B}$ with
orthogonal projection $P_{B}^{\bot}$. Then the operators $\Omega$ and $B$,
with respect to this direct sum, are $2\times2$ block operator matrices%
\begin{equation}
\Omega=\left[
\begin{array}
[c]{cc}%
\Omega_{2} & \Theta\\
\Theta^{\ast} & \Omega_{1}%
\end{array}
\right]  ,\quad B=\left[
\begin{array}
[c]{cc}%
B_{2} & 0\\
0 & 0
\end{array}
\right]  , \label{hbcom2}%
\end{equation}
where $\Omega_{2}:=\left.  P_{B}\Omega P_{B}\right\vert _{H_{B}}%
:H_{B}\rightarrow H_{B}$ and $B_{2}:=\left.  P_{B}BP_{B}\right\vert _{H_{B}%
}:H_{B}\rightarrow H_{B}$ are restrictions of the operators $\Omega$ and $B$
respectively to loss subspace $H_{B}$ whereas $\Omega_{1}:=\left.  P_{B}%
^{\bot}\Omega P_{B}^{\bot}\right\vert _{H_{B}^{\bot}}:H_{B}^{\bot}\rightarrow
H_{B}^{\bot}$ is the restriction of $\Omega$ to complementary subspace
$H_{B}^{\bot}$. \ Also, $\Theta:H_{B}^{\bot}\rightarrow H_{B}$ is the operator
$\Theta:=\left.  P_{B}\Omega P_{B}^{\bot}\right\vert _{H_{B}^{\bot}}$ whose
adjoint is given by $\Theta^{\ast}=\left.  P_{B}^{\bot}\Omega P_{B}\right\vert
_{H_{B}}:H_{B}\rightarrow H_{B}^{\bot}$. The block representation
(\ref{hbcom2}) plays an important role in our analysis involving the
perturbation theory as well as the \emph{Schur complement} concept described
in Appendix \ref{apxsc}.

\subsection{Modal dichotomy for the high-loss regime\label{smdhl}}

Notice\ first that in view of (\ref{mvt7}) the operator $\left(
-\mathrm{i}\beta\right)  ^{-1}A\left(  \beta\right)  =B+\mathrm{i}\beta
^{-1}\Omega$ is analytic in $\beta^{-1}$ in a vicinity of $\beta=\infty$. Let
then $\zeta\left(  \beta\right)  $ be an analytic in $\beta^{-1}$ eigenvalue
of $A\left(  \beta\right)  $ in the same vicinity, with the possible exception
of a pole at $\beta=\infty$. Notice that if use the substitution
$\varepsilon=\left(  -\mathrm{i}\beta\right)  ^{-1}$ the operator $\varepsilon
A\left(  \mathrm{i}\varepsilon^{-1}\right)  =B+\varepsilon\Omega$ is a
self-adjoint for real $\varepsilon$ and consequently the eigenvalue
$\lambda\left(  \varepsilon\right)  =\varepsilon\zeta\left(  \mathrm{i}%
\varepsilon^{-1}\right)  $ of the operator $B+\varepsilon\Omega$ must be an
analytic function of $\varepsilon$ and real-valued for real $\varepsilon$.
Hence it satisfies the identity $\overline{\lambda\left(  \overline
{\varepsilon}\right)  }=\lambda\left(  \varepsilon\right)  $ where
$\overline{\varepsilon}$ is the complex conjugate to $\varepsilon$. The later
in view of the identity $\zeta\left(  \beta\right)  =\left(  -\mathrm{i}%
\beta\right)  \lambda\left(  \left(  -\mathrm{i}\beta\right)  ^{-1}\right)  $
readily implies the following identities for the eigenvalue $\zeta\left(
\beta\right)  $ for real $\beta$ in a vicinity of $\beta=\infty$
\begin{equation}
\overline{\zeta\left(  \beta\right)  }=\zeta\left(  -\beta\right)  ,\text{ or
}\operatorname{Re}\zeta\left(  -\beta\right)  =\operatorname{Re}\zeta\left(
\beta\right)  ,\quad\operatorname{Im}\zeta\left(  -\beta\right)
=-\operatorname{Im}\zeta\left(  \beta\right)  \text{.} \label{abw9}%
\end{equation}
Consequently, $\operatorname{Re}\zeta\left(  \beta\right)  $ and
$\operatorname{Im}\zeta\left(  \beta\right)  $ are respectively an even and an
odd function for real $\beta$ in a vicinity of $\beta=\infty$ implying that
their Laurent series in $\beta^{-1}$ have respectively only even and odd powers.

The perturbation analysis for $\beta\gg1$ of the operator $A\left(
\beta\right)  =\Omega-\mathrm{i}\beta B$ described in Section \ref{shlre},
introduces an orthonormal basis $\{\mathring{w}_{j}\}_{j=1}^{N}$ diagonalizing
the operators $\Omega_{1}$ and $B_{2}$ from the block form (\ref{hbcom2}),
that is%
\begin{equation}
B_{2}\mathring{w}_{j}=\mathring{\zeta}_{j}\mathring{w}_{j}\ \text{for }1\leq
j\leq N_{B};\quad\Omega_{1}\mathring{w}_{j}=\rho_{j}\mathring{w}%
_{j}\ \text{for }N_{B}+1\leq j\leq N, \label{hbu3d}%
\end{equation}
where%
\begin{gather}
\mathring{\zeta}_{j}=\left(  \mathring{w}_{j},B_{2}\mathring{w}_{j}\right)
=\left(  \mathring{w}_{j},B\mathring{w}_{j}\right)  \ \text{for }1\leq j\leq
N_{B};\label{hbu3da}\\
\rho_{j}=\left(  \mathring{w}_{j},\Omega_{1}\mathring{w}_{j}\right)  =\left(
\mathring{w}_{j},\Omega\mathring{w}_{j}\right)  \ \text{for }N_{B}+1\leq j\leq
N.\nonumber
\end{gather}
The summary of the perturbation analysis for the high-loss regime $\beta\gg1$,
as described in Theorem \ref{Thm1}, is as follows. \ The system operator
$A(\beta)$ is diagonalizable and there exists a complete set of eigenvalues
$\zeta_{j}\left(  \beta\right)  $ and eigenvectors $w_{j}\left(  \beta\right)
$ satisfying%
\begin{equation}
A\left(  \beta\right)  w_{j}\left(  \beta\right)  =\zeta_{j}\left(
\beta\right)  w_{j}\left(  \beta\right)  ,\text{$\quad$}1\leq j\leq N,\text{
\ \ }\beta\gg1 \label{abw1}%
\end{equation}
which split into two distinct classes
\begin{gather}
\text{high-loss}\text{:$\quad$}\zeta_{j}\left(  \beta\right)  ,\text{ }%
w_{j}\left(  \beta\right)  ,\text{$\quad$}1\leq j\leq N_{B};\label{abw2}\\
\text{low-loss}\text{:$\quad$}\zeta_{j}\left(  \beta\right)  ,\text{ }%
w_{j}\left(  \beta\right)  ,\text{$\quad$}N_{B}+1\leq j\leq N,\nonumber
\end{gather}
with the following properties.

\textbf{In the high-loss case} the eigenvalues have poles at $\beta=\infty$
whereas their eigenvectors are analytic at $\beta=\infty$, having the
asymptotic expansions%
\begin{equation}
\zeta_{j}\left(  \beta\right)  =-\mathrm{i}\mathring{\zeta}_{j}\beta+\rho
_{j}+O\left(  \beta^{-1}\right)  ,\text{$\quad$}\mathring{\zeta}%
_{j}>0,\text{$\quad$}\rho_{j}\in%
%TCIMACRO{\U{211d} }%
%BeginExpansion
\mathbb{R}
%EndExpansion
,\text{ \ \ }w_{j}\left(  \beta\right)  =\mathring{w}_{j}+O\left(  \beta
^{-1}\right)  ,\text{ \ \ }1\leq j\leq N_{B}. \label{abw3}%
\end{equation}
The vectors $\mathring{w}_{j}$, $1\leq j\leq N_{B}$ form an orthonormal basis
of the loss subspace $H_{B}$ and%
\begin{equation}
B\mathring{w}_{j}=\mathring{\zeta}_{j}\mathring{w}_{j},\text{$\quad$}\rho
_{j}=\left(  \mathring{w}_{j},\Omega\mathring{w}_{j}\right)  ,\text{ for
}1\leq j\leq N_{B}. \label{abw4}%
\end{equation}

\textbf{In the low-loss case} the eigenvalues and eigenvectors are analytic at
$\beta=\infty$, having the asymptotic expansions%
\begin{align}
\zeta_{j}\left(  \beta\right)   &  =\rho_{j}-\mathrm{i}d_{j}\beta
^{-1}+O\left(  \beta^{-2}\right)  ,\text{ \ \ }\rho_{j}\in%
%TCIMACRO{\U{211d} }%
%BeginExpansion
\mathbb{R}
%EndExpansion
,\text{ $\ \ d_{j}$}\geq0,\label{abw5}\\
w_{j}\left(  \beta\right)   &  =\mathring{w}_{j}+w_{j}^{(-1)}\beta
^{-1}+O\left(  \beta^{-2}\right)  ,\text{ \ \ }N_{B}+1\leq j\leq N.\nonumber
\end{align}
The vectors $\mathring{w}_{j}$, $N_{B}+1\leq j\leq N$ form an orthonormal
basis of the no-loss subspace $H_{B}^{\perp}$ and%
\begin{equation}
B\mathring{w}_{j}=0,\text{$\quad$}\rho_{j}=\left(  \mathring{w}_{j}%
,\Omega\mathring{w}_{j}\right)  ,\text{$\quad$}d_{j}=\left(  w_{j}%
^{(-1)},Bw_{j}^{(-1)}\right)  \text{ for }N_{B}+1\leq j\leq N. \label{abw6}%
\end{equation}

The expansions (\ref{abw3}) and (\ref{abw5}) together with (\ref{abw9})
readily imply the following asymptotic formulas for the real and imaginary
parts of the complex eigenvalues $\zeta_{j}(\beta)$ for $\beta\gg1$
\begin{align}
\text{high-loss}  &  \text{: \ }\operatorname{Re}\zeta_{j}(\beta)=\rho
_{j}+O\left(  \beta^{-2}\right)  ,\text{ \ }\operatorname{Im}\zeta_{j}\left(
\beta\right)  =-\mathring{\zeta}_{j}\beta+O\left(  \beta^{-1}\right)  ,\text{
\ }1\leq j\leq N_{B};\label{abw7a}\\
\text{low-loss}  &  \text{: \ }\operatorname{Re}\zeta_{j}(\beta)=\rho
_{j}+O\left(  \beta^{-2}\right)  ,\text{ \ }\operatorname{Im}\zeta_{j}\left(
\beta\right)  =-d_{j}\beta^{-1}+O\left(  \beta^{-3}\right)  ,\text{ }%
N_{B}+1\leq j\leq N. \label{abw7b}%
\end{align}
Observe that the expansions (\ref{abw7a}) and (\ref{abw7b}) readily yield
\begin{equation}
\lim_{\beta\rightarrow\infty}\operatorname{Im}\zeta_{j}\left(  \beta\right)
=-\infty\text{ for }1\leq j\leq N_{B};\quad\lim_{\beta\rightarrow\infty
}\operatorname{Im}\zeta_{j}\left(  \beta\right)  =0\text{ for }N_{B}+1\leq
j\leq N, \label{abw8}%
\end{equation}
justifying the names high-loss and low-loss. Notice also that the relations
(\ref{abw3})--(\ref{abw6}) imply that the high-loss eigenmodes projection on
the no-loss subspace $H_{B}^{\bot}$ is of order $\beta^{-1}$ in contrast to
the low-loss eigenmodes for which the projection on the loss subspace $H_{B}$
is of order $\beta^{-1}$. \emph{In other words, for }$\beta\gg1$\emph{\ the
high-loss eigenmodes are essentially confined to the loss subspace }$H_{B}%
$\emph{\ whereas the low-loss modes are essentially expelled from it.}

\subsection{Losses and the quality factor associated with the
eigenmodes\label{SecQualFac}}

Here we consider the energy dissipation associated with high-loss and low-loss
eigenmodes. The power dissipation is commonly quantified by the so called
\emph{quality factor} $Q$ that can naturally be introduced in a few not
entirely equivalent ways, \cite[pp. 47, 70, 71]{Pain}. The most common way to
define the quality factor is based on relative rate of the energy dissipation
per cycle when the system is in a state of damped harmonic oscillations
$v\left(  t\right)  $ with a given frequency $\omega$, namely,%
\begin{equation}
Q=2\pi\frac{\text{energy stored in system}}{\text{energy lost per cycle}%
}=\left\vert \omega\right\vert \frac{U}{W_{\mathrm{dis}}}=\left\vert
\omega\right\vert \frac{\left(  v(t),v(t)\,\right)  }{2\beta\left(
v(t),Bv\,(t)\right)  }, \label{hbu7}%
\end{equation}
where we used for the system energy $U$ and the dissipated power
$W_{\mathrm{dis}}$ their expressions (\ref{mvt9a}). Notice also that in the
above formula we use the absolute value $\left\vert \omega\right\vert $ of the
frequency $\omega$ since in our settings the frequency $\omega$ can be
negative. The state of damped harmonic oscillations $v(t)$ is defined by an
eigenvector $w$ of the system operator $A\left(  \beta\right)  =\Omega
-\mathrm{i}\beta B$ with eigenvalue $\zeta$, and it evolves as
$v(t)=w\mathrm{e}^{-\mathrm{i}\zeta t}$ with the frequency $\omega
=\operatorname{Re}\zeta$ and the damping factor $-\operatorname{Im}\zeta$. Its
system energy $U$, dissipated power $W_{\mathrm{dis}}$, and quality factor $Q$
satisfy
\begin{equation}
U=U\left[  w\right]  e^{2\operatorname{Im}\zeta t},\text{ \ \ }%
W_{\text{\textrm{dis}}}=W_{\text{\textrm{dis}}}\left[  w\right]
e^{2\operatorname{Im}\zeta t},\text{ \ \ }Q=Q\left[  w\right]  ,
\end{equation}
where%
\begin{equation}
U\left[  w\right]  =\frac{1}{2}\left(  w,w\right)  ,\quad
W_{\text{\textrm{dis}}}\left[  w\right]  =-2\operatorname{Im}\zeta U\left[
w\right]  , \label{quwz1}%
\end{equation}%
\begin{equation}
\operatorname{Re}\zeta=\frac{\left(  w,\Omega w\right)  }{\left(  w,w\right)
},\quad\operatorname{Im}\zeta=-\frac{\left(  w,\beta Bw\right)  }{\left(
w,w\right)  }, \label{quwz2}%
\end{equation}%
\begin{equation}
Q\left[  w\right]  =-\frac{1}{2}\frac{\left\vert \operatorname{Re}%
\zeta\right\vert }{\operatorname{Im}\zeta}. \label{quwz3}%
\end{equation}
For an eigenvector $w$ we refer to the terms $U\left[  w\right]  $,
$W_{dis}\left[  w\right]  $, and $Q\left[  w\right]  $ as its energy, power of
energy dissipation, and quality factor, respectively. \ Observe, that\emph{
eigenvectors with the same eigenvalue have equal quality factors}. Notice also
that an eigenvector $w$ with eigenvalue $\zeta$ has power of energy
dissipation $W_{dis}\left[  w\right]  $ equal to the product $-\left(
w,w\right)  \operatorname{Im}\zeta$ and quality factor $Q[w]$ equal the ratio
$\frac{\left\vert \operatorname{Re}\zeta\right\vert }{-2\operatorname{Im}%
\zeta}$.

Consider now the high-loss regime $\beta\gg1$. Let $\zeta_{j}\left(
\beta\right)  $, $1\leq j\leq N$ denote the high-loss and low-loss eigenvalues
of the system operator $A\left(  \beta\right)  $ which have the expansions
(\ref{abw7a}), (\ref{abw7b}). Then for any eigenvectors $w_{j}\left(
\beta\right)  $, $1\leq j\leq N$ with these eigenvalues, respectively, which
are normalized in the sense
\begin{equation}
\left(  w_{j}\left(  \beta\right)  ,w_{j}\left(  \beta\right)  \right)
=1+O\left(  \beta^{-1}\right)  ,\text{$\quad$}1\leq j\leq N
\end{equation}
as $\beta\rightarrow\infty$, the following asymptotic formulas holds as
$\beta\rightarrow\infty$ for the energy and the power of energy dissipation of
these modes
\begin{align}
&  U\left[  w_{j}\left(  \beta\right)  \right]  =\frac{1}{2}+O\left(
\beta^{-1}\right)  ,\text{$\quad$}1\leq j\leq N;\label{quwz4}\\
\text{high-loss}\text{: \ \ }  &  W_{\text{\textrm{dis}}}\left[  w_{j}\left(
\beta\right)  \right]  =\mathring{\zeta}_{j}\beta+O\left(  1\right)  ,\text{
\ \ }1\leq j\leq N_{B};\label{quwz5}\\
\text{low-loss}\text{: \ \ }  &  W_{\text{\textrm{dis}}}\left[  w_{j}\left(
\beta\right)  \right]  =d_{j}\beta^{-1}+O\left(  \beta^{-2}\right)  ,\text{
\ \ }N_{B}+1\leq j\leq N. \label{quwz6}%
\end{align}
We see clearly now the modal dichotomy, i.e. eigenmode splitting according to
their dissipative properties: high-loss modes $w_{j}\left(  \beta\right)  $,
$1\leq j\leq N_{B}$ and low-loss modes $w_{j}\left(  \beta\right)  $,
$N_{B}+1\leq j\leq N$. \ Indeed, these asymptotic formulas (\ref{quwz5}),
(\ref{quwz6}) imply%
\begin{equation}
\text{high-loss modes: }\lim_{\beta\rightarrow\infty}W_{\text{\textrm{dis}}%
}\left[  w_{j}\left(  \beta\right)  \right]  =\infty;\text{ low-loss modes:
}\lim_{\beta\rightarrow\infty}W_{\text{\textrm{dis}}}\left[  w_{j}\left(
\beta\right)  \right]  =0. \label{quwz6_5}%
\end{equation}
\qquad The quality factor $Q\left[  w_{j}\left(  \beta\right)  \right]  $ for
each high-loss eigenmode has a series expansion containing only odd powers of
$\beta^{-1}$ with the asymptotic formula as $\beta\rightarrow\infty$%
\begin{equation}
Q\left[  w_{j}\left(  \beta\right)  \right]  =\frac{1}{2}\frac{\left\vert
\rho_{j}\right\vert }{\mathring{\zeta}_{j}}\beta^{-1}+O\left(  \beta
^{-3}\right)  ,\text{$\quad$}1\leq j\leq N_{B}. \label{quwz7}%
\end{equation}
The quality factor $Q\left[  w_{j}\left(  \beta\right)  \right]  $ for each
low-loss eigenvectors has a series expansion containing only odd powers of
$\beta^{-1}$ as well provided $\operatorname{Im}\zeta_{j}\left(  \beta\right)
\not \equiv 0$ for $\beta\gg1$. \ Moreover, it satisfies the following
asymptotic formula as $\beta\rightarrow\infty$
\begin{equation}
Q\left[  w_{j}\left(  \beta\right)  \right]  =\frac{1}{2}\frac{\left\vert
\rho_{j}\right\vert }{d_{j}}\beta+O\left(  \beta^{-1}\right)  ,\text{$\quad$%
}N_{B}+1\leq j\leq N, \label{quwz8}%
\end{equation}
provided $d_{j}\not =0$. \ In fact, it is true under rather general conditions
that $d_{j}>0$, for $j=N_{B}+1,\ldots,N$ (see (\ref{abw5}) and Remark
\ref{Rem1}). These asymptotic formulas (\ref{quwz7}), (\ref{quwz8}) readily
imply that%
\begin{equation}
\text{high-loss modes: }\lim_{\beta\rightarrow\infty}Q\left[  w_{j}\left(
\beta\right)  \right]  =0;\text{ low-loss modes: }\lim_{\beta\rightarrow
\infty}Q\left[  w_{j}\left(  \beta\right)  \right]  =\left\{
\begin{array}
[c]{cc}%
\infty & \text{if }\rho_{j}\not =0,\\
0 & \text{if }\rho_{j}=0.
\end{array}
\right.  \label{hbu11}%
\end{equation}
\emph{Observe, that the relations (\ref{quwz6}) and (\ref{quwz8}) clearly
indicate that for the low-loss modes the larger values of }$\beta
$\emph{\ imply lesser losses and the possibility of a higher quality factor!}
\ \emph{In particular, the more lossy is the lossy component the less lossy
are the low-loss modes}. This somewhat paradoxical conclusion can be explained
by the fact that the\ low-loss eigenmodes are being expelled from the loss
subspace $H_{B}$ in the sense that their projection onto this subspace
satisfies asymptotically $P_{B}w_{j}\left(  \beta\right)  =O(\beta^{-1})$ as
$\beta\rightarrow\infty$.

\subsection{Losses for external harmonic forces\label{SecQualFacHarm}}

Let us subject now our system to a harmonic external force $\hat{f}\left(
\omega\right)  e^{-\mathrm{i}\omega t}$ of a frequency $\omega$ that will set
the system into a stationary oscillatory motion of the form $\hat{v}\left(
\omega\right)  e^{-\mathrm{i}\omega t}$ of the same frequency $\omega$ and
amplitude $\hat{v}\left(  \omega\right)  $ depending on the energy
dissipation. Or more generally we can subject the system to an external force
$f\left(  t\right)  $ and observe its response $v\left(  t\right)  $ governed
by the evolution equation (\ref{mvt7}). The solution to this problem in view
of the rest condition (\ref{mvt7a}) can be obtained with the help of the
Fourier-Laplace transform%
\begin{equation}
\hat{v}\left(  \xi\right)  =\int_{0}^{\infty}e^{\mathrm{i}\xi t}%
v(t)dt,\ \operatorname{Im}\xi>0, \label{vzet1}%
\end{equation}
applied to the evolution equation (\ref{mvt7}) resulting in the following
equation
\begin{equation}
\xi\hat{v}\left(  \xi\right)  =\left[  \Omega-\mathrm{i}\beta B\right]
\hat{v}\left(  \xi\right)  +\mathrm{i}\hat{f}\left(  \xi\right)  ,\ \xi
=\omega+\mathrm{i}\eta,\ \eta=\operatorname{Im}\xi>0. \label{vzet2}%
\end{equation}
For $\operatorname{Im}\xi>0$ in view of $B\geq0$ the operator $\xi
I-(\Omega-\mathrm{i}\beta B)$ is invertible, and hence
\begin{gather}
\hat{v}\left(  \xi\right)  =\mathfrak{A}\left(  \xi\right)  \hat{f}\left(
\xi\right)  ,\ \xi=\omega+\mathrm{i}\eta,\ \eta=\operatorname{Im}%
\xi>0,\label{vzet3}\\
\mathfrak{A}\left(  \xi\right)  =\mathrm{i}\left[  \xi I-(\Omega
-\mathrm{i}\beta B)\right]  ^{-1}. \label{vzet4}%
\end{gather}
For a harmonic force $f\left(  t\right)  =fe^{-\mathrm{i}\omega t}$ the
corresponding harmonic solution is $v\left(  t\right)  =ve^{-\mathrm{i}\omega
t}$ where%
\begin{equation}
v=\mathfrak{A}\left(  \omega\right)  f,\text{ \ \ }\mathfrak{A}\left(
\omega\right)  =\mathrm{i}\left[  \omega I-A(\beta)\right]  ^{-1},\text{
\ \ }A(\beta)=\Omega-\mathrm{i}\beta B,\text{ }\beta\geq0. \label{vtez5}%
\end{equation}
The operator $\mathfrak{A}\left(  \omega\right)  $ is called the
\emph{admittance operator}.

For the stationary regime associated with a harmonic external force $f\left(
t\right)  =fe^{-\mathrm{i}\omega t}$ the quality factor can be naturally
defined by a formula analogous to (\ref{hbu7}), namely%
\begin{equation}
Q=Q_{f,\omega}=2\pi\frac{\text{energy stored in system}}{\text{energy lost per
cycle}}, \label{vtez6b}%
\end{equation}
where the energy lost refers specifically to the energy loss due to friction
in the system. \ By the expressions (\ref{mvt9a}), the quality factor turns
into%
\begin{equation}
Q=\left\vert \omega\right\vert \frac{U}{W_{\mathrm{dis}}}=\left\vert
\omega\right\vert \frac{\frac{1}{2}\left(  v,v\right)  }{\beta\left(
v,Bv\right)  }, \label{vtez6c}%
\end{equation}
where according to (\ref{mvt9})%
\begin{gather}
U=\frac{1}{2}\left(  v,v\right)  \text{ is the stored energy,}\label{vtez6d}\\
W_{\mathrm{dis}}=\beta\left(  v,Bv\right)  \text{ is the power of dissipated
energy.}\nonumber
\end{gather}

In many cases of interest the external force $f$ is outside the loss subspace
corresponding to a situation when the driving forces/sources are located
outside the lossy component of the system. This important factor is described
by the projection on the no-loss space $H_{B}^{\bot}=H\ominus H_{B}$, that is
by $P_{B}^{\bot}f$. We may expect the effect of losses to depend significantly
on whether $P_{B}^{\bot}f=0$ or $P_{B}^{\bot}f\not =0$. But even if
$P_{B}^{\bot}f=f$, that is $f$ is outside the loss subspace $H_{B}$, there may
still be losses since all system degrees of freedom can be coupled. The
analysis of the stored and dissipated energies, in view of the relations
(\ref{vzet3})-(\ref{vtez6d}), depends on the admittance operator
$\mathfrak{A}\left(  \omega\right)  $. To study the properties of the
admittance operator $\mathfrak{A}\left(  \omega\right)  $ defined by
(\ref{vtez5}) we consider the block form (\ref{hbcom1}) and (\ref{hbcom2}) and
represent $\omega I-A\left(  \beta\right)  $ as $2\times2$ block operator
matrix%
\begin{gather}
\omega I-A\left(  \beta\right)  =\left[
\begin{array}
[c]{cc}%
\Xi_{2}\left(  \omega,\beta\right)  & -\Theta\\
-\Theta^{\ast} & \Xi_{1}\left(  \omega\right)
\end{array}
\right]  ,\label{ksia1}\\
\Xi_{2}\left(  \omega,\beta\right)  :=\omega I_{2}-\left(  \Omega
_{2}-\mathrm{i}\beta B_{2}\right)  ,\quad\Xi_{1}\left(  \omega\right)
:=\omega I_{1}-\Omega_{1}\nonumber
\end{gather}
where $I_{2}$ and $I_{1}$ denote the identity operators on the spaces $H_{B}$
and $H_{B}^{\bot}$, respectively. With respect to this block representation,
the \emph{Schur complement} of $\Xi_{2}\left(  \omega,\beta\right)  $ in
$\omega I-A\left(  \beta\right)  $ is defined as the operator
\begin{equation}
S_{2}\left(  \omega,\beta\right)  =\Xi_{1}\left(  \omega\right)  -\Theta
^{\ast}\Xi_{2}\left(  \omega,\beta\right)  ^{-1}\Theta, \label{ksia2}%
\end{equation}
whenever $\Xi_{2}\left(  \omega,\beta\right)  $ is invertible.

In what follows we assume the frequency $\omega\not =0$ is not one of the
resonance frequencies, that is $\omega\neq\rho_{j},$ $N_{B}+1\leq j\leq N$.
\ Then we know by Proposition \ref{Prop7} that the operators $\Xi_{1}\left(
\omega\right)  $, $\Xi_{2}\left(  \omega,\beta\right)  $, $S_{2}\left(
\omega,\beta\right)  $, and $\omega I-A\left(  \beta\right)  $ are invertible
for $\beta\gg1$. To simplify lengthy expressions we will suppress the symbols
$\omega,\beta$ appearing as arguments in operators $\Xi_{1}\left(
\omega\right)  $, $\Xi_{2}\left(  \omega,\beta\right)  $, $S_{2}\left(
\omega,\beta\right)  $. \ Furthermore, the explicit formula based on the Schur
complement is derived for the admittance operator%
\begin{equation}
\mathfrak{A}\left(  \omega\right)  =\mathrm{i}\left[  \omega I-A(\beta
)\right]  ^{-1}=\mathrm{i}\left[
\begin{array}
[c]{cc}%
\Xi_{2}^{-1}+\Xi_{2}^{-1}\Theta S_{2}^{-1}\Theta^{\ast}\Xi_{2}^{-1} & \Xi
_{2}^{-1}\Theta S_{2}^{-1}\\
S_{2}^{-1}\Theta^{\ast}\Xi_{2}^{-1} & S_{2}^{-1}%
\end{array}
\right]  . \label{ksia3}%
\end{equation}
A perturbation analysis at $\beta=\infty$ of the admittance operator
$\mathfrak{A}\left(  \omega\right)  $, the results of which are summarized in
Proposition \ref{Prop7}, yields the following asymptotic expansion for
$\beta\rightarrow\infty$
\begin{equation}
\mathfrak{A}\left(  \omega\right)  =\left[
\begin{array}
[c]{cc}%
0 & 0\\
0 & \mathrm{i}\Xi_{1}^{-1}%
\end{array}
\right]  +W^{\left(  -1\right)  }\beta^{-1}+O\left(  \beta^{-2}\right)
,\text{ \ \ }W^{\left(  -1\right)  }\geqslant0, \label{ksia4}%
\end{equation}
where%
\begin{equation}
W^{\left(  -1\right)  }=\left[
\begin{array}
[c]{cc}%
B_{2}^{-1} & B_{2}^{-1}\Theta\Xi_{1}^{-1}\\
\left(  \Xi_{1}^{-1}\right)  ^{\ast}\Theta^{\ast}B_{2}^{-1} & \left(  \Xi
_{1}^{-1}\right)  ^{\ast}\Theta^{\ast}B_{2}^{-1}\Theta\Xi_{1}^{-1}%
\end{array}
\right]  ,\text{ \ \ }B_{2}^{-1}>0\text{.}%
\end{equation}
These asymptotics for the admittance operator lead to asymptotic formulas as
$\beta\rightarrow\infty$ for the energy $U$, the dissipation power
$W_{\text{\textrm{dis}}}$, and the quality $Q$ factor which depend on whether
$P_{B}^{\bot}f=0$ or $P_{B}^{\bot}f\not =0$. Namely, if $P_{B}^{\bot}f=0$,
that is if $f$ is inside the loss subspace $H_{B}$, then Theorem \ref{Thm3}
tells us that%
\begin{align}
&  U=\frac{1}{2}\left(  f,\left[  B_{2}^{-2}+B_{2}^{-1}\Theta\left(  \Xi
_{1}^{-1}\right)  ^{\ast}\Xi_{1}^{-1}\Theta^{\ast}B_{2}^{-1}\right]  f\right)
\beta^{-2}+O\left(  \beta^{-3}\right)  ,\label{ksia5}\\
&  W_{\text{\textrm{dis}}}=\left(  f,B_{2}^{-1}f\right)  \beta^{-1}+O\left(
\beta^{-2}\right)  ,\nonumber\\
&  Q=\left\vert \omega\right\vert \frac{\frac{1}{2}\left(  f,\left[
B_{2}^{-2}+B_{2}^{-1}\Theta\left(  \Xi_{1}^{-1}\right)  ^{\ast}\Xi_{1}%
^{-1}\Theta^{\ast}B_{2}^{-1}\right]  f\right)  }{\left(  f,B_{2}^{-1}f\right)
}\beta^{-1}+O\left(  \beta^{-2}\right)  ,\nonumber
\end{align}
and the leading order terms of $U$, $W_{\text{\textrm{dis}}}$ and $Q$ are
positive numbers. \ In particular, the quality factor $Q\rightarrow0$ as
$\beta\rightarrow\infty$.

If $P_{B}^{\bot}f\not =0$ then Theorem \ref{Thm4} tells us that%
\begin{align}
&  U=\frac{1}{2}\left(  \Xi_{1}^{-1}P_{B}^{\bot}f,\Xi_{1}^{-1}P_{B}^{\bot
}f\right)  +O\left(  \beta^{-1}\right)  ,\label{ksia6}\\
&  W_{\text{\textrm{dis}}}=\left(  f,W^{\left(  -1\right)  }f\right)
\beta^{-1}+O\left(  \beta^{-2}\right)  ,\nonumber
\end{align}
and the leading order term of $U$ and $W_{\text{\textrm{dis}}}$ is a positive
and a nonnegative number, respectively. \ Furthermore, the quality factor is
either infinite for $\beta\gg1$ (the case $W_{\text{\textrm{dis}}}\equiv0$) or
$Q\rightarrow\infty$ as $\beta\rightarrow\infty$. \ Moreover,
\begin{equation}
Q=\left\vert \omega\right\vert \frac{\frac{1}{2}\left(  \Xi_{1}^{-1}%
P_{B}^{\bot}f,\Xi_{1}^{-1}P_{B}^{\bot}f\right)  }{\left(  f,W^{\left(
-1\right)  }f\right)  }\beta+O\left(  1\right)
\end{equation}
provided $\left(  f,W^{\left(  -1\right)  }f\right)  \not =0$, in which case
the leading order term for $Q$ is a positive number.

Therefore the quality factor $Q$ satisfies $\lim_{\beta\rightarrow\infty}Q$
$=\infty$ provided $f$ has a non-zero projection on the no-loss subspace
$H_{B}^{\bot}=H\ominus H_{B}$, whereas otherwise $\lim_{\beta\rightarrow
\infty}Q$ $=0$.

\section{An electric circuit example\label{secex}}

One of the important applications of our methods described above is electric
circuits and networks involving resistors representing losses. A general study
of electric networks with losses can be carried out with the help of the
Lagrangian approach, and that systematic study is left for another
publication. For Lagrangian treatment of electric networks and circuits we
refer to \cite[Sec. 9]{Gantmacher}, \cite[Sec. 2.5]{Goldstein}, \cite{Pars}.

We illustrate the idea and give a flavor of the efficiency of our methods by
considering below a rather simple example of an electric circuit as in Fig.
\ref{Figc1}. \ This example\ will show the essential features of two component
systems incorporating high-loss and lossless components.%

%TCIMACRO{\FRAME{ftbpFU}{4.8698in}{2.6749in}{0pt}{\Qcb{An electric circuit
%involving three capacitances $C_{1}$, $C_{2}$, $C_{12}$, two inductances
%$L_{1}$, $L_{2}$, a resistor $R_{2}$, and two sources $E_{1}$, $E_{2}$.}%
%}{\Qlb{Figc1}}{losscircuit1.eps}{\special{ language "Scientific Word";
%type "GRAPHIC";  maintain-aspect-ratio TRUE;  display "USEDEF";
%valid_file "F";  width 4.8698in;  height 2.6749in;  depth 0pt;
%original-width 6.9176in;  original-height 3.7818in;  cropleft "0";
%filename 'LossCircuit1.eps';file-properties "XNPEU";}} }%
%BeginExpansion
\begin{figure}
[ptb]
\begin{center}
\includegraphics[
height=2.6749in,
width=4.8698in
]%
{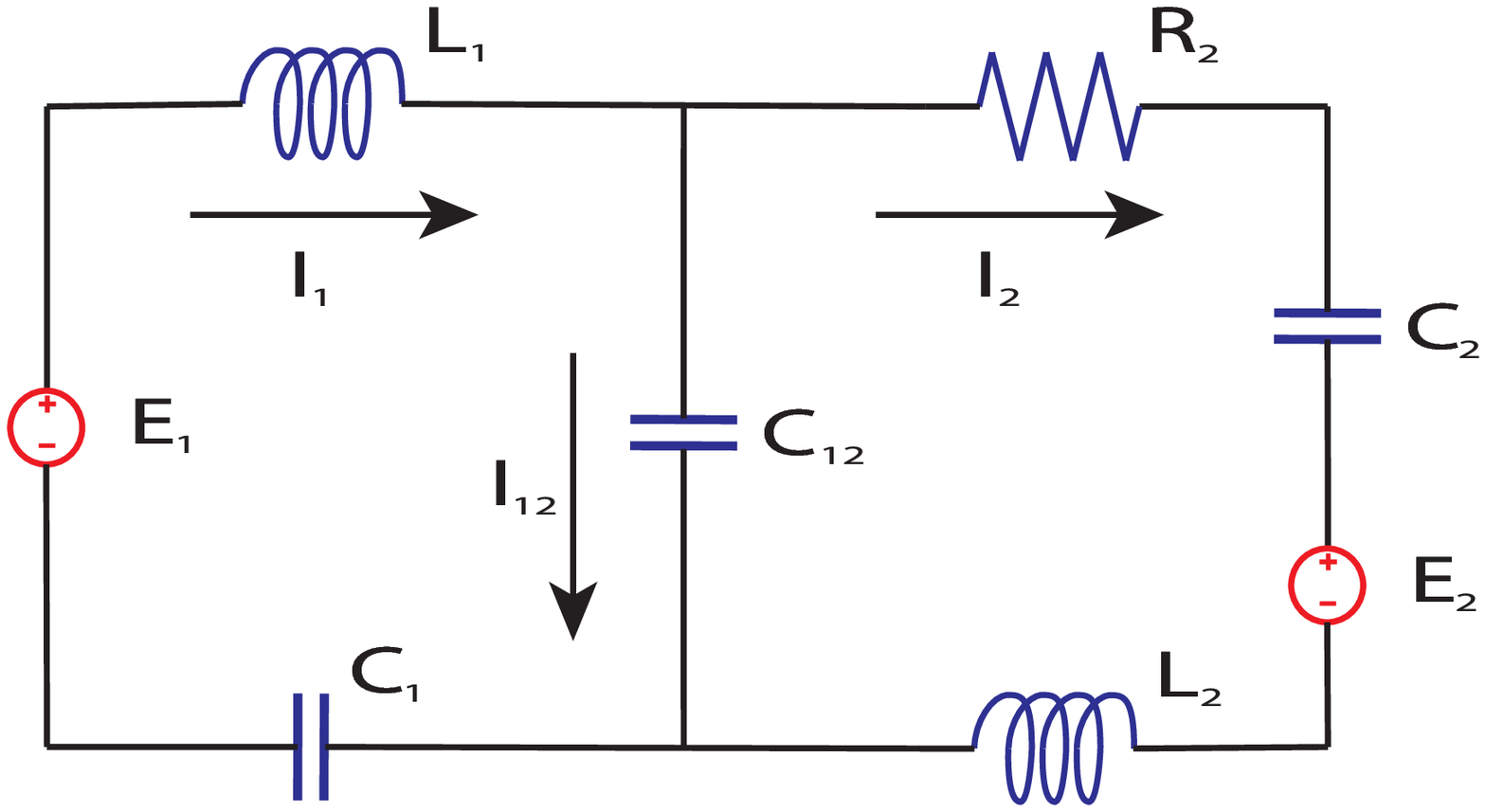}%
\caption{An electric circuit involving three capacitances $C_{1}$, $C_{2}$,
$C_{12}$, two inductances $L_{1}$, $L_{2}$, a resistor $R_{2}$, and two
sources $E_{1}$, $E_{2}$.}%
\label{Figc1}%
\end{center}
\end{figure}
%EndExpansion
To derive evolution equations for the electric circuit in Fig. \ref{Figc1} we
use a general method for constructing Lagrangians for circuits, \cite[Sec.
9]{Gantmacher}, that yields%
\begin{equation}
T=\frac{L_{1}}{2}\dot{q}_{1}^{2}+\frac{L_{2}}{2}\dot{q}_{2}^{2},\quad
U=\frac{1}{2C_{1}}q_{1}^{2}+\frac{1}{2C_{12}}\left(  q_{1}-q_{2}\right)
^{2}+\frac{1}{2C_{2}}q_{2}^{2},\quad R=\frac{R_{2}}{2}\dot{q}_{2}^{2},
\label{circ1}%
\end{equation}
where $T$ and $U$ are respectively the kinetic and the potential energies,
$T-U$ is the Lagrangian, and $R$ is the \textit{Rayleigh dissipative
function}. Notice that $I_{1}=\dot{q}_{1}$ and $I_{2}=\dot{q}_{2}$ are the
currents. The general Euler-Lagrange equations are, \cite[Sec. 8]{Gantmacher},%
\begin{equation}
\frac{\partial}{\partial t}\frac{\partial T}{\partial\dot{q}_{j}}%
-\frac{\partial T}{\partial q_{j}}=-\frac{\partial U}{\partial q_{j}}%
-\frac{\partial R}{\partial\dot{q}_{j}}, \label{circ2}%
\end{equation}
and for (\ref{circ1}) the general Euler-Lagrange equations therefore have the
following form%
\begin{gather}
\frac{\partial}{\partial t}L_{1}\dot{q}_{1}=-\frac{1}{C_{1}}q_{1}-\frac
{1}{C_{12}}\left(  q_{1}-q_{2}\right)  +f_{1},\label{circ3}\\
\frac{\partial}{\partial t}L_{2}\dot{q}_{2}=-\frac{1}{C_{2}}q_{2}+\frac
{1}{C_{12}}\left(  q_{1}-q_{2}\right)  -R_{2}\dot{q}_{2}+f_{2}.\nonumber
\end{gather}
If we introduce%
\begin{equation}
Q=\left[
\begin{array}
[c]{c}%
q_{1}\\
q_{2}%
\end{array}
\right]  ,\quad F=\left[
\begin{array}
[c]{c}%
f_{1}\\
f_{2}%
\end{array}
\right]  , \label{circ4}%
\end{equation}%
\begin{equation}
\mathsf{L}=\left[
\begin{array}
[c]{cc}%
L_{1} & 0\\
0 & L_{2}%
\end{array}
\right]  ,\quad\mathsf{G}=\left[
\begin{array}
[c]{cc}%
C_{1}^{-1}+C_{12}^{-1} & -C_{12}^{-1}\\
-C_{12}^{-1} & C_{2}^{-1}+C_{12}^{-1}%
\end{array}
\right]  ,\quad\mathsf{R}=\left[
\begin{array}
[c]{cc}%
0 & 0\\
0 & R_{2}%
\end{array}
\right]  , \label{circ5}%
\end{equation}
the system (\ref{circ3}) can be recast into
\begin{equation}
\mathsf{L}\partial_{t}^{2}Q+\mathsf{R}\partial_{t}Q+\mathsf{G}Q=F.
\label{circ6}%
\end{equation}
To provide for an efficient spectral study of the vector equation
(\ref{circ6}) we notice that%
\begin{equation}
\mathsf{L}>0,\text{ }\mathsf{G}>0, \label{circ7}%
\end{equation}
and introduce%
\begin{equation}
\tilde{Q}=\mathsf{L}^{\frac{1}{2}}Q,\quad\tilde{F}=\mathsf{L}^{-\frac{1}{2}%
}F,\quad\mathsf{\tilde{R}}=\mathsf{L}^{-\frac{1}{2}}\mathsf{RL}^{-\frac{1}{2}%
},\quad\mathsf{\tilde{G}}=\mathsf{L}^{-\frac{1}{2}}\mathsf{GL}^{-\frac{1}{2}%
}=\Phi^{2}. \label{circ8}%
\end{equation}
The vector equation (\ref{circ6}) is transformed then into
\begin{equation}
\partial_{t}^{2}\tilde{Q}+\mathsf{\tilde{R}}\partial_{t}\tilde{Q}%
+\mathsf{\tilde{G}}\tilde{Q}=\tilde{F}. \label{circ9}%
\end{equation}
To rewrite the above second-order ODE as the first-order ODE we set%
\begin{equation}
X=\left[
\begin{array}
[c]{c}%
y\\
x
\end{array}
\right]  =\left[
\begin{array}
[c]{c}%
\partial_{t}\tilde{Q}\\
\Phi\tilde{Q}%
\end{array}
\right]  ,\text{ that is }y=\partial_{t}\tilde{Q},\quad x=\Phi\tilde{Q},
\label{circ10}%
\end{equation}
allowing one to recast the vector equation (\ref{circ9}) as
\begin{equation}
\partial_{t}X=\left[
\begin{array}
[c]{cc}%
-\mathsf{\tilde{R}} & -\Phi\\
\Phi & 0
\end{array}
\right]  X+\left[
\begin{array}
[c]{c}%
\tilde{F}\\
0
\end{array}
\right]  , \label{circ11}%
\end{equation}
where%
\begin{equation}
\left[
\begin{array}
[c]{cc}%
-\mathsf{\tilde{R}} & -\Phi\\
\Phi & 0
\end{array}
\right]  =\left[
\begin{array}
[c]{cccc}%
0 & 0 & -\Phi_{11} & -\Phi_{12}\\
0 & -R_{2}L_{2}^{-1} & -\Phi_{12} & -\Phi_{22}\\
\Phi_{11} & \Phi_{12} & 0 & 0\\
\Phi_{12} & \Phi_{22} & 0 & 0
\end{array}
\right]  . \label{circ12}%
\end{equation}
Consequently, the general Euler-Lagrange equations (\ref{circ3}) for the
electric circuit in Fig. \ref{Figc1} are transformed into the canonical form
(\ref{mvt7}), namely%
\[
\partial_{t}X=-\mathrm{i}\left(  \Omega-\mathrm{i}\beta B\right)  X+\left[
\begin{array}
[c]{c}%
\tilde{F}\\
0
\end{array}
\right]  ,
\]
with the system operator
\begin{equation}
A\left(  \beta\right)  =\Omega-\mathrm{i}\beta B,\text{\quad}\beta\geq0,
\label{circ13}%
\end{equation}
\qquad%
\begin{equation}
\Omega=\left[
\begin{array}
[c]{cc}%
0 & -\mathrm{i}\Phi\\
\mathrm{i}\Phi & 0
\end{array}
\right]  ,\quad\Phi=\left[
\begin{array}
[c]{cc}%
\Phi_{11} & \Phi_{12}\\
\Phi_{12} & \Phi_{22}%
\end{array}
\right]  >0, \label{circ14}%
\end{equation}
\qquad%
\begin{equation}
B=\left[
\begin{array}
[c]{cccc}%
0 & 0 & 0 & 0\\
0 & \tau^{-1} & 0 & 0\\
0 & 0 & 0 & 0\\
0 & 0 & 0 & 0
\end{array}
\right]  ,\text{ \ \ and }\beta=\frac{R_{2}\tau}{L_{2}}\text{, where }%
\tau>0\text{ is a unit of time}. \label{circ15}%
\end{equation}
As we can see this electric circuit example fits within the framework of our
model.\ Indeed, since resistors represent losses, this two component system
consists of a lossy component and a lossless component -- the right and left
circuits in Fig. \ref{Figc1}, respectively.

The next two sections are devoted to the analysis of this electric circuit
both theoretically and numerically in the high-loss regime $\beta\gg1$ using
the methods developed in this paper. \ For this purpose the following
properties of the matrix $\Phi$ in (\ref{circ14}) are useful:%
\begin{gather}
\Phi=\left(  \sqrt{\operatorname{Tr}\left(  \Phi^{2}\right)  +2\sqrt
{\det\left(  \Phi^{2}\right)  }}\right)  ^{-1}\left(  \sqrt{\det\left(
\Phi^{2}\right)  }I_{2}+\Phi^{2}\right)  ,\label{circ14_1}\\
\Phi^{2}=\left[
\begin{array}
[c]{cc}%
\frac{1}{L_{1}}\left(  \frac{1}{C_{1}}+\frac{1}{C_{12}}\right)  & -\frac
{1}{\sqrt{L_{1}}\sqrt{L_{2}}C_{12}}\\
-\frac{1}{\sqrt{L_{1}}\sqrt{L_{2}}C_{12}} & \frac{1}{L_{2}}\left(  \frac
{1}{C_{2}}+\frac{1}{C_{12}}\right)
\end{array}
\right]  >0,\nonumber\\
\Phi_{11},\text{ }\Phi_{22}>0,\quad\Phi_{12}\in%
%TCIMACRO{\U{211d} }%
%BeginExpansion
\mathbb{R}
%EndExpansion
\backslash\{0\},\nonumber
\end{gather}
where $\sqrt{\cdot}$ denotes the positive square root.

\subsection{Spectral analysis in the high-loss regime}

In this section, a spectral analysis of the electric circuit example in Fig.
\ref{Figc1} in the high-loss regime is given using the main results of this paper.

\subsubsection{Perturbation analysis}

The finite dimensional Hilbert space is $H=%
%TCIMACRO{\U{2102} }%
%BeginExpansion
\mathbb{C}
%EndExpansion
^{4}$ under the standard inner product $\left(  \cdot,\cdot\right)  $. \ It is
decomposed into the direct sum of invariant subspace of the operator $B\geq0$
in (\ref{circ15}),%
\[
H=H_{B}\oplus H_{B}^{\bot},
\]
where $H_{B}=\operatorname{ran}B$, $H_{B}^{\bot}=\ker B$ are the loss subspace
and no-loss subspace with dimensions $N_{B}=1$, $N-N_{B}=3$ and orthogonal
projections
\[
P_{B}=\left[
\begin{array}
[c]{cccc}%
0 & 0 & 0 & 0\\
0 & 1 & 0 & 0\\
0 & 0 & 0 & 0\\
0 & 0 & 0 & 0
\end{array}
\right]  ,\text{\quad}P_{B}^{\perp}=\left[
\begin{array}
[c]{cccc}%
1 & 0 & 0 & 0\\
0 & 0 & 0 & 0\\
0 & 0 & 1 & 0\\
0 & 0 & 0 & 1
\end{array}
\right]  ,
\]
respectively. The operators $\Omega$ and $B$, with respect to this direct sum,
are $2\times2$ block operator matrices%
\[
\Omega=\left[
\begin{array}
[c]{cc}%
\Omega_{2} & \Theta\\
\Theta^{\ast} & \Omega_{1}%
\end{array}
\right]  ,\quad B=\left[
\begin{array}
[c]{cc}%
B_{2} & 0\\
0 & 0
\end{array}
\right]  ,
\]
where $\Omega_{2}:=\left.  P_{B}\Omega P_{B}\right\vert _{H_{B}}%
:H_{B}\rightarrow H_{B}$ and $B_{2}:=\left.  P_{B}BP_{B}\right\vert _{H_{B}%
}:H_{B}\rightarrow H_{B}$ are restrictions of the operators $\Omega$ and $B$,
respectively, to loss subspace $H_{B}$ whereas $\Omega_{1}:=\left.
P_{B}^{\bot}\Omega P_{B}^{\bot}\right\vert _{H_{B}^{\bot}}:H_{B}^{\bot
}\rightarrow H_{B}^{\bot}$ is the restriction of $\Omega$ to complementary
subspace $H_{B}^{\bot}$. \ Also, $\Theta:H_{B}^{\bot}\rightarrow H_{B}$ is the
operator $\Theta:=\left.  P_{B}\Omega P_{B}^{\bot}\right\vert _{H_{B}^{\bot}}$
whose adjoint is given by $\Theta^{\ast}=\left.  P_{B}^{\bot}\Omega
P_{B}\right\vert _{H_{B}}:H_{B}\rightarrow H_{B}^{\bot}$. \ Moreover,
according to our perturbation theory the operator $\Theta^{\ast}B_{2}%
^{-1}\Theta:H_{B}^{\bot}\rightarrow H_{B}^{\bot}$ plays a key role in the
analysis. These operators act on the $4\times1$ column vectors in their
respective domains as matrix multiplication by the $4$ x $4$ matrices%
\begin{gather*}
\Omega_{2}=0,\text{\quad}B_{2}=B,\text{\quad}\Omega_{1}=\left[
\begin{array}
[c]{cccc}%
0 & 0 & -\mathrm{i}\Phi_{11} & -\mathrm{i}\Phi_{12}\\
0 & 0 & 0 & 0\\
\mathrm{i}\Phi_{11} & 0 & 0 & 0\\
\mathrm{i}\Phi_{12} & 0 & 0 & 0
\end{array}
\right]  ,\\
\Theta=\left[
\begin{array}
[c]{cccc}%
0 & 0 & 0 & 0\\
0 & 0 & -\mathrm{i}\Phi_{12} & -\mathrm{i}\Phi_{22}\\
0 & 0 & 0 & 0\\
0 & 0 & 0 & 0
\end{array}
\right]  ,\text{\quad}\Theta^{\ast}=\left[
\begin{array}
[c]{cccc}%
0 & 0 & 0 & 0\\
0 & 0 & 0 & 0\\
0 & \mathrm{i}\Phi_{12} & 0 & 0\\
0 & \mathrm{i}\Phi_{22} & 0 & 0
\end{array}
\right]  ,\\
B_{2}^{-1}=\left[
\begin{array}
[c]{cccc}%
0 & 0 & 0 & 0\\
0 & \tau & 0 & 0\\
0 & 0 & 0 & 0\\
0 & 0 & 0 & 0
\end{array}
\right]  ,\text{\quad}\Theta^{\ast}B_{2}^{-1}\Theta=\left[
\begin{array}
[c]{cccc}%
0 & 0 & 0 & 0\\
0 & 0 & 0 & 0\\
0 & 0 & \tau\Phi_{12}^{2} & \tau\Phi_{12}\Phi_{22}\\
0 & 0 & \tau\Phi_{12}\Phi_{22} & \tau\Phi_{22}^{2}%
\end{array}
\right]  ,
\end{gather*}
where by "$=$" we mean equality as functions from the domain of the operator
on the LHS of the equal sign.

The operators $\Omega_{1}$ and $B_{2}$ in this example have only simple
eigenvalues and so we will use Corollary \ref{Cor2}. \ We introduce below a
fixed orthonormal basis $\{\mathring{w}_{j}\}_{j=1}^{4}$ diagonalizing the
operators $\Omega_{1}$ and $B_{2}$ and then determine the values
$\mathring{\zeta}_{j}$, $\rho_{j}$, $d_{j}$ from the relations
\begin{align*}
B_{2}\mathring{w}_{j}  &  =\mathring{\zeta}_{j}\mathring{w}_{j},\text{ }%
\rho_{j}=\left(  \mathring{w}_{j},\Omega\mathring{w}_{j}\right)  \text{ for
}j=1;\text{ \ }\\
\Omega_{1}\mathring{w}_{j}  &  =\rho_{j}\mathring{w}_{j},\text{ }d_{j}=\left(
\mathring{w}_{j},\Theta^{\ast}B_{2}^{-1}\Theta\mathring{w}_{j}\right)
\ \text{for }2\leq j\leq4.
\end{align*}
In particular,%
\begin{equation}
\mathring{w}_{1}=\left[
\begin{array}
[c]{c}%
0\\
1\\
0\\
0
\end{array}
\right]  ,\text{ \ \ }\mathring{\zeta}_{1}=\tau^{-1},\text{ \ \ }\rho_{1}=0;
\label{circ17}%
\end{equation}%
\begin{gather}
\mathring{w}_{2}=\frac{1}{\sqrt{\Phi_{11}^{2}+\Phi_{12}^{2}}}\left[
\begin{array}
[c]{c}%
0\\
0\\
-\Phi_{12}\\
\Phi_{11}%
\end{array}
\right]  ,\text{$\quad$}\rho_{2}=0,\text{ \ \ }d_{2}=\frac{\tau\left(
\Phi_{12}^{2}-\Phi_{11}\Phi_{22}\right)  ^{2}}{\Phi_{11}^{2}+\Phi_{12}^{2}%
}>0,\label{circ19}\\[0.02in]
\mathring{w}_{3}=\frac{1}{\sqrt{2}}\frac{1}{\sqrt{\Phi_{11}^{2}+\Phi_{12}^{2}%
}}\left[
\begin{array}
[c]{c}%
-\mathrm{i}\sqrt{\Phi_{11}^{2}+\Phi_{12}^{2}}\\
0\\
\Phi_{11}\\
\Phi_{12}%
\end{array}
\right]  ,\text{ }\rho_{3}=\sqrt{\Phi_{11}^{2}+\Phi_{12}^{2}},\text{ }%
d_{3}=\frac{\frac{1}{2}\tau\Phi_{12}^{2}\left(  \Phi_{11}+\Phi_{22}\right)
^{2}}{\Phi_{11}^{2}+\Phi_{12}^{2}},\nonumber\\[0.02in]
\mathring{w}_{4}=\overline{\mathring{w}_{3}},\text{$\quad$}\rho_{4}=-\rho
_{3}<0,\text{ \ \ }d_{4}=d_{3}>0.\nonumber
\end{gather}

By Theorem \ref{Thm1} and Corollary \ref{Cor2} of this paper it follows that
in the high-loss regime $\beta\gg1$, the system operator $A\left(
\beta\right)  =\Omega-\mathrm{i}\beta B$ is diagonalizable and there exists a
complete set of eigenvalues and eigenvectors satisfying
\[
A\left(  \beta\right)  w_{j}\left(  \beta\right)  =\zeta_{j}\left(
\beta\right)  w_{j}\left(  \beta\right)  ,\text{$\quad$}1\leq j\leq4,
\]
which splits into two classes%
\begin{gather}
\text{high-loss}\text{:$\quad$}\zeta_{j}\left(  \beta\right)  ,\text{ }%
w_{j}\left(  \beta\right)  ,\text{$\quad$}j=1;\label{circ16_0}\\
\text{low-loss}\text{:$\quad$}\zeta_{j}\left(  \beta\right)  ,\text{ }%
w_{j}\left(  \beta\right)  ,\text{$\quad$}2\leq j\leq4,\nonumber
\end{gather}
having the following properties.

\textbf{High-loss modes}. The high-loss eigenvalue has a pole at $\beta
=\infty$ whereas its eigenvector is analytic at $\beta=\infty$, having the
asymptotic expansion%
\begin{gather}
\zeta_{1}\left(  \beta\right)  =-\mathrm{i}\mathring{\zeta}_{1}\beta+\rho
_{1}+O\left(  \beta^{-1}\right)  ,\text{$\quad$}\mathring{\zeta}%
_{1}>0,\text{$\quad$}\rho_{1}\in%
%TCIMACRO{\U{211d} }%
%BeginExpansion
\mathbb{R}
%EndExpansion
,\label{circ16}\\
w_{1}\left(  \beta\right)  =\mathring{w}_{1}+O\left(  \beta^{-1}\right)
,\nonumber
\end{gather}
as $\beta\rightarrow\infty$. \ The vector $\mathring{w}_{1}$ is an orthonormal
basis of the loss subspace $H_{B}$ and%
\[
B\mathring{w}_{1}=\mathring{\zeta}_{1}\mathring{w}_{1},\text{$\quad$}\rho
_{1}=\left(  \mathring{w}_{1},\Omega\mathring{w}_{1}\right)  .
\]

\textbf{Low-loss modes}. The low-loss eigenvalues and eigenvectors are
analytic at $\beta=\infty$, having the asymptotic expansions%
\begin{gather}
\text{\ }\zeta_{j}\left(  \beta\right)  =\rho_{j}-\mathrm{i}d_{j}\beta
^{-1}+O\left(  \beta^{-2}\right)  ,\text{$\quad$}\rho_{j}\in%
%TCIMACRO{\U{211d} }%
%BeginExpansion
\mathbb{R}
%EndExpansion
,\text{ \ \ }d_{j}>0,\label{circ18}\\
w_{j}\left(  \beta\right)  =\mathring{w}_{j}+w_{j}^{(-1)}\beta^{-1}+O\left(
\beta^{-2}\right)  ,\text{$\quad$}2\leq j\leq4,\nonumber
\end{gather}
as $\beta\rightarrow\infty$. The vectors $\mathring{w}_{j}$, $2\leq j\leq4$
form an orthonormal basis of the no-loss subspace $H_{B}^{\perp}$ and%
\[
B\mathring{w}_{j}=0,\text{$\quad$}\rho_{j}=\left(  \mathring{w}_{j}%
,\Omega\mathring{w}_{j}\right)  ,\text{$\quad$}d_{j}=\left(  w_{j}^{\left(
-1\right)  },Bw_{j}^{\left(  -1\right)  }\right)  ,\text{$\quad$}2\leq
j\leq4.
\]

\subsubsection{Overdamping and symmetries of the spectrum}

The phenomenon of \emph{overdamping} (also called \emph{heavy damping}) is
best known for a simple damped oscillator. Namely, when the damping exceeds
certain critical value all oscillations cease entirely, see, for instance,
\cite[Sec. 2]{Pain}. In other words, if the damped oscillations are described
by the exponential function $\mathrm{e}^{-\mathrm{i}\zeta t}$ with a complex
constant $\zeta$ then in the case of overdamping (heavy damping)
$\operatorname{Re}\zeta=0$. Our interest in overdamping is motivated by the
fact that if an eigenmode becomes overdamped then it will not resonate at any
frequency. Consequently, the contribution of such a mode to losses becomes
minimal, and that provides a mechanism for the absorption suppression for
systems composed of lossy and lossless components.

The treatment of overdamping for systems with many degrees of freedom involves
a number of subtleties particularly in our case when the both lossy and
lossless degrees of freedom are present. We have reasons to believe though
that any Lagrangian system with losses accounted by the Rayleigh dissipation
function can have all high-loss eigenmodes overdamped for a sufficiently large
value of the loss parameter $\beta$. \ In order to give valuable insights into
far more general systems, we focus on the electric circuit example in Fig. 1
giving statements and providing arguments on the spectral symmetry and
overdamping for the circuit. This analysis is used in the next section to
interpret the behavior of the eigenvalues of the circuit operator $A\left(
\beta\right)  $.

Our first principal statement is on a symmetry of the spectrum of the system
operator $A\left(  \beta\right)  $ with respect to the imaginary axis.

\begin{proposition}
[spectral symmetry]\label{Prop0}Let $A\left(  \beta\right)  $ denote the
system operator (\ref{circ13}) for the electric circuit given in Fig.
\ref{Figc1}. Then for each $\beta\geq0$, its spectrum $\sigma\left(  A\left(
\beta\right)  \right)  $ lies in the lower half of the complex plane and is
symmetric with respect to the imaginary axis, that is%
\begin{equation}
\sigma\left(  A\left(  \beta\right)  \right)  =-\overline{\sigma\left(
A\left(  \beta\right)  \right)  }.\text{$\quad$} \label{sigab1}%
\end{equation}
Moreover, except for a finite set of values of $\beta$, the system operator
$A\left(  \beta\right)  $ is diagonalizable with four nondegenerate eigenvalues.
\end{proposition}

\begin{proof}
From the asymptotic analysis in (\ref{circ16})--(\ref{circ19}) it follows that
all the eigenvalues of $A\left(  \beta\right)  $ must be distinct for
$\beta\gg1$. Now the operator $\Omega-\mathrm{i}\beta B$, $\beta\in%
%TCIMACRO{\U{2102} }%
%BeginExpansion
\mathbb{C}
%EndExpansion
$ is analytic on $%
%TCIMACRO{\U{2102} }%
%BeginExpansion
\mathbb{C}
%EndExpansion
$ and so, by a well-known fact from perturbation theory \cite[p. 25, Theorem
3; p. 225, Theorem 1]{Bau85}, its Jordan structure is invariant except on a
set $S\subseteq%
%TCIMACRO{\U{2102} }%
%BeginExpansion
\mathbb{C}
%EndExpansion
$ which is closed and isolated. \ These facts imply the system operator
$A\left(  \beta\right)  $ is diagonalizable with distinct eigenvalues except
on the closed and isolated set $S\cap\lbrack0,\infty)$ which must be bounded
since the eigenvalues of $A\left(  \beta\right)  $ are distinct for $\beta
\gg1$. In particular, this implies $S\cap\lbrack0,\infty)$ is a finite set.
\ This proves that except for a finite set of values of $\beta$, the system
operator $A\left(  \beta\right)  $ is diagonalizable with four nondegenerate eigenvalues.

Next, since $A\left(  \beta\right)  =\Omega-\mathrm{i}\beta B$ in
(\ref{circ13}) is a system operator satisfying the power dissipation condition
$B\geq0$ then it follows from Lemma \ref{apxlm} in Appendix \ref{apxqf} that
if $\beta\geq0$ then $\operatorname{Im}\zeta\leq0$ if $\zeta$ is an eigenvalue
of $A\left(  \beta\right)  $, i.e., the spectrum $\sigma\left(  A\left(
\beta\right)  \right)  $ lies in the lower half of the complex plane. Finally,
one can show that $\det\left(  \zeta I-A\left(  \beta\right)  \right)
=\det\mathsf{L}^{-1}\det\left(  \zeta^{2}\mathsf{L}+\zeta\mathrm{i}%
R-\mathsf{G}\right)  $ for all $\zeta\in%
%TCIMACRO{\U{2102} }%
%BeginExpansion
\mathbb{C}
%EndExpansion
$, $\beta\geq0$. Moreover, from our assumptions $\beta\geq0$, $\tau>0$ and
$L$, $G>0$ it follows that the $2\times2$ matrices $L$, $R$, and $G$ must have
real entries. Using these two facts we conclude
\begin{gather}
\det\left(  \zeta I-A\left(  \beta\right)  \right)  =\det\mathsf{L}^{-1}%
\det\left(  \zeta^{2}\mathsf{L}+\zeta\mathrm{i}R-\mathsf{G}\right)
=\label{sigab2}\\
=\overline{\det\mathsf{L}^{-1}\det\left(  \left(  -\overline{\zeta}\right)
^{2}\mathsf{L}+(-\overline{\zeta})\mathrm{i}R-\mathsf{G}\right)  }%
=\overline{\det\left(  \left(  -\overline{\zeta}\right)  I-A\left(
\beta\right)  \right)  },\nonumber
\end{gather}
and, hence, (\ref{sigab1}) holds.
\end{proof}

\begin{corollary}
[eigenvalue symmetry]\label{Cor0a}Let $\mathcal{I}$ be any open interval in
$\left(  0,\infty\right)  $ with the property that the eigenvalues of system
operator $A\left(  \beta\right)  $ are nondegenerate for every $\beta
\in\mathcal{I}$. Then there exists a unique set of functions $\zeta
_{j}:\mathcal{I}\rightarrow%
%TCIMACRO{\U{2102} }%
%BeginExpansion
\mathbb{C}
%EndExpansion
$, $j=1,2,3,4$ which are analytic at each $\beta\in\mathcal{I}$ and whose
values $\zeta_{1}\left(  \beta\right)  $, $\zeta_{2}\left(  \beta\right)  $,
$\zeta_{3}\left(  \beta\right)  $, $\zeta_{4}\left(  \beta\right)  $ are the
eigenvalues of the system operator $A\left(  \beta\right)  $. Moreover, there
exists a unique permutation $\kappa:\left\{  1,2,3,4\right\}  \longmapsto
\left\{  1,2,3,4\right\}  $ depending only on the interval $\mathcal{I}$ such
that for each $j=1,2,3,4$,%
\begin{equation}
\zeta_{j}\left(  \beta\right)  =-\overline{\zeta_{\kappa\left(  j\right)
}\left(  \beta\right)  }\text{ \ for every }\beta\in\mathcal{I}.
\label{sigab3}%
\end{equation}

\end{corollary}

\begin{proof}
It is a well-known fact from perturbation theory for matrices depending
analytically on a parameter \cite{Bau85}, that simple eigenvalues can be
chosen to be analytic locally in the perturbation parameter and analytically
continued as eigenvalues along any path in the domain of analyticity of the
matrix function which does not intersect a closed and isolated set of
singularities. These singularities are necessarily contained in the set of
parameters in the domain where the value of matrix function has repeated
eigenvalues. The proof of the first part of this corollary now follows
immediately from this fact and the fact the high-loss and low-loss eigenvalues
of $A\left(  \beta\right)  $ are meromorphic and analytic at $\beta=\infty\,$,
respectively, with distinct values for $\beta\gg1$. \ The existence and
uniqueness of the permutation is now obvious from this and symmetry of the
spectrum described in the previous proposition. \ This completes the proof.
\end{proof}

\begin{corollary}
[overdamping]\label{Cor0b}Let $\zeta_{j}\left(  \beta\right)  $, $j=1,2,3,4$
be the high-loss and low-loss eigenvalues of the system operator $A\left(
\beta\right)  $ given by (\ref{circ16_0})-(\ref{circ19}). Then, in the
high-loss regime $\beta\gg1$, these eigenvalues lie in the lower open
half-plane and, moreover, the eigenvalues $\zeta_{j}\left(  \beta\right)  $,
$j=1,2$ are on the imaginary axis whereas the eigenvalues $\zeta_{j}\left(
\beta\right)  $, $j=3,4$ lie off this axis and symmetric to it, i.e.,
$\zeta_{4}\left(  \beta\right)  =-\overline{\zeta_{3}\left(  \beta\right)  }$.
\end{corollary}

\begin{proof}
First, it follows the asymptotic analysis in (\ref{circ16})-(\ref{circ19})
that there exists a $\beta_{0}>0$ such that $\zeta_{j}\left(  \beta\right)  $,
$j=1,2,3,4$ are all the eigenvalues of system operator $A\left(  \beta\right)
$ and are distinct for every $\beta\in\mathcal{(}\beta_{0},\infty)$. \ By the
previous corollary there exists a unique permutation $\kappa:\left\{
1,2,3,4\right\}  \longmapsto\left\{  1,2,3,4\right\}  $ depending only on the
interval $\mathcal{(}\beta_{0},\infty)$ such that for each $j=1,2,3,4$, the
identity (\ref{sigab3}) for every $\beta\in\mathcal{(}\beta_{0},\infty)$.

Next, we will now show that for this permutation we have $\kappa(1)=1$,
$\kappa(2)=2$, $\kappa(3)=4$, $\kappa\left(  4\right)  =3$. Well, consider the
the asymptotic expansions of the imaginary and real parts of high-loss and
low-loss eigenvalues. First, $\lim_{\beta\rightarrow\infty}\operatorname{Im}%
\zeta_{1}\left(  \beta\right)  =-\infty$, $\lim_{\beta\rightarrow\infty
}\operatorname{Im}\zeta_{j}\left(  \beta\right)  =0$, $j=2,3,4$ and since
$\zeta_{1}\left(  \beta\right)  =-\overline{\zeta_{\kappa\left(  1\right)
}\left(  \beta\right)  }$ these properties imply $\kappa(1)=1$. \ Second,
$\lim_{\beta\rightarrow\infty}\operatorname{Re}\zeta_{2}\left(  \beta\right)
=0$, $\lim_{\beta\rightarrow\infty}\operatorname{Re}\zeta_{4}\left(
\beta\right)  =\rho_{4}=-\rho_{3}=-\lim_{\beta\rightarrow\infty}%
\operatorname{Re}\zeta_{3}\left(  \beta\right)  $ with $\rho_{3}>0$ and since
$\zeta_{j}\left(  \beta\right)  =-\overline{\zeta_{\kappa\left(  j\right)
}\left(  \beta\right)  }$ with $\kappa\left(  j\right)  \not =1$ for $j=2,3,4$
these properties imply $\kappa(2)=2$, $\kappa(3)=4$, $\kappa(4)=3$.

To complete the proof we notice that since $\kappa(1)=1$, $\kappa(2)=2$,
$\kappa(3)=4$, $\kappa\left(  4\right)  =3$ then for $\beta\gg1$ we have
$\zeta_{j}\left(  \beta\right)  =-\overline{\zeta_{j}\left(  \beta\right)  }$,
$j=1,2$ and $\zeta_{4}\left(  \beta\right)  =-\overline{\zeta_{3}\left(
\beta\right)  }$. The proof now follows immediately from this and the facts
$-\operatorname{Re}\zeta_{4}\left(  \beta\right)  =\operatorname{Re}\zeta
_{3}\left(  \beta\right)  =\rho_{3}+O\left(  \beta^{-2}\right)  $,
$\operatorname{Im}\zeta_{4}\left(  \beta\right)  =\operatorname{Im}\zeta
_{3}\left(  \beta\right)  =-d_{3}\beta^{-1}+O\left(  \beta^{-3}\right)  $ as
$\beta\rightarrow\infty$ where $\rho_{3}$, $d_{3}>0$.
\end{proof}

\begin{remark}
\label{rbdod}The boundary of the overdamping regime known as critical damping,
corresponds to a value of the loss parameter $\beta=\beta_{0}>0$ at which the
system operator $A\left(  \beta\right)  $ develops a purely imaginary but
degenerate eigenvalue $\zeta_{0}$. \ The\ spectral perturbation analysis of
$A\left(  \beta\right)  $ in a neighborhood of the point $\beta=\beta_{0}$ is
theoretically and computationally a difficult problem since it is a
perturbation of the non-self-adjoint operator $A\left(  \beta_{0}\right)  $
with a degenerate eigenvalue. \ This type of perturbation problem was
considered in \cite{Welters} where asymptotic expansions of the perturbed
eigenvalues and eigenvectors were given and explicit recursive formulas to
compute their series expansions were found \cite[Theorem 3.1]{Welters}, under
a generic condition \cite[p. 2, (1.1)]{Welters}. \ In particular, this
condition is satisfied for the system operator $A\left(  \beta\right)  $ at
the point $\beta=\beta_{0}$ for the degenerate eigenvalue $\zeta_{0}$ since
\[
\frac{\partial}{\partial\beta}\det\left(  \zeta I-A\left(  \beta\right)
\right)  |_{\left(  \zeta,\beta\right)  =\left(  \zeta_{0},\beta_{0}\right)
}=\mathrm{i}\tau^{-1}\zeta_{0}^{3}-\mathrm{i}\tau^{-1}\left(  \Phi_{11}%
^{2}+\Phi_{12}^{2}\right)  \zeta_{0}\not =0.
\]

\end{remark}

\subsection{Numerical analysis}

In order to illustrate the behavior of the eigenvalues of the system operator
for the circuit in Fig. \ref{Figc1} we fix positive values for capacitance
$C_{1}$, $C_{2}$, $C_{12}$, inductances $L_{1}$, $L_{2}$, and the unit of time
$\tau$. Once these are fixed, the system operator $A(\beta)$ is computed using
(\ref{circ5}), (\ref{circ8}), and (\ref{circ13})--(\ref{circ15}). \ These
values constrain the magnitude of the resistance $R_{2}$ of the corresponding
circuit in Fig. \ref{Figc1} to be proportional to the dimensionless loss
parameter $\beta$ since it follows from (\ref{circ15}) that%
\begin{equation}
R_{2}=\frac{L_{2}}{\tau}\beta\text{.} \label{circ22_5}%
\end{equation}
The high-loss regime $\beta\gg1$ is associated with the right circuit in Fig.
\ref{Figc1} experiencing huge losses due to the resistance $R_{2}\gg1$ while
the left circuit remains lossless. \ In particular, each choice of these
values provides a numerical example of a physical model with a two component
system composed of a high-loss and lossless components.

For the numerical analysis in this section we chose
\begin{equation}
C_{1}:=2,\quad C_{2}:=3,\quad C_{12}:=4,\quad L_{1}:=5,\quad L_{2}%
:=6,\ \ \tau:=1. \label{circ23}%
\end{equation}
All graphs were plotted in Maple$^{\text{\textregistered}}$ using these fixed
values and with the loss parameter in the domain $0\leq\beta\leq10$.%

\begin{figure}
[ptb]
\begin{center}
\includegraphics[
height=6.0424in,
width=6.0277in
]%
{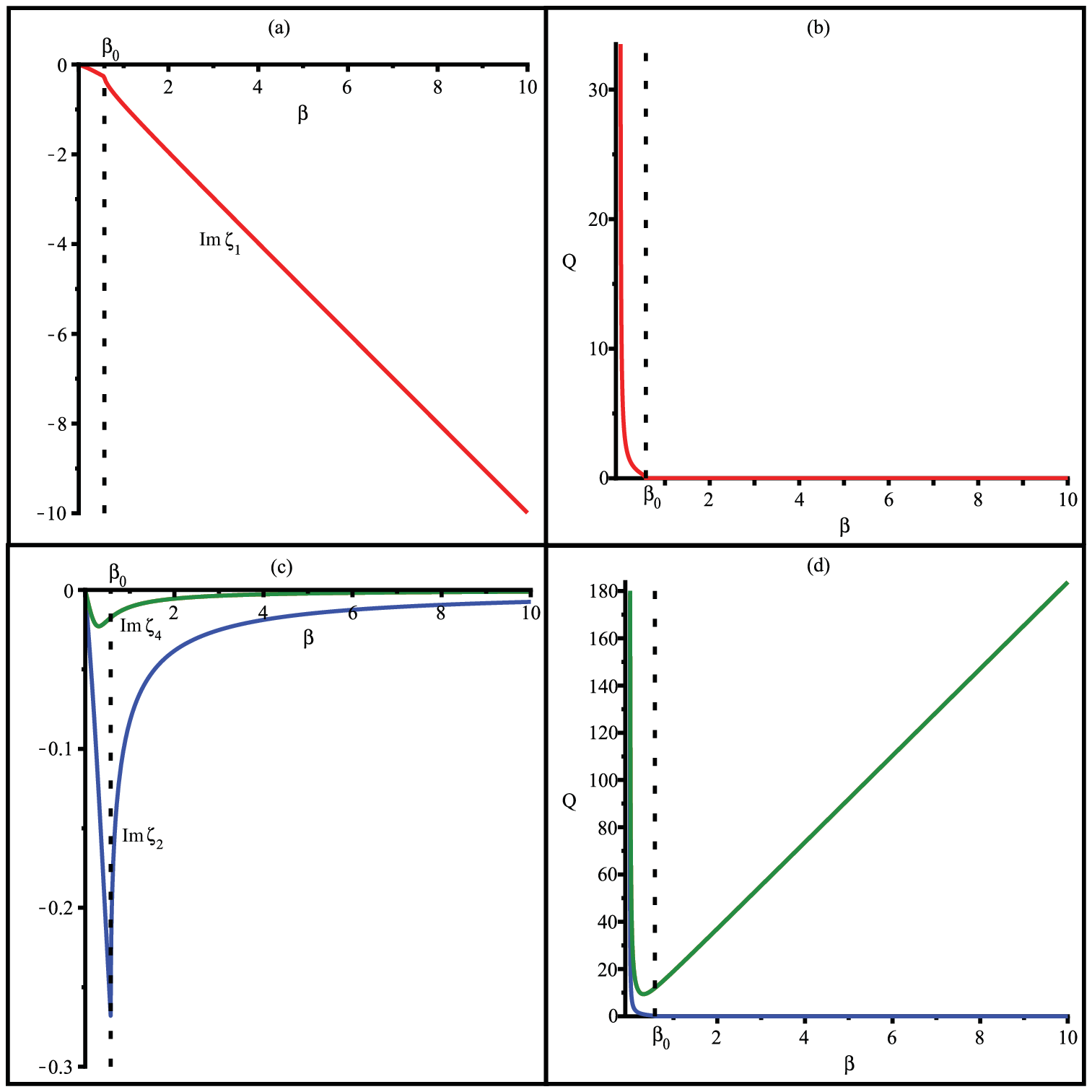}%
\caption{(a)-(d) For each high-loss eigenvalue $\zeta_{j}\left(  \beta\right)
$, $j=1$ and low-loss eigenvalue $\zeta_{j}\left(  \beta\right)  $, $2\leq
j\leq4$ of the system operator $A(\beta)$ for the electric circuit in Fig.
\ref{Figc1} with the values (\ref{circ22_5}) and (\ref{circ23}), comparing its
imaginary part $\operatorname{Im}\zeta_{j}\left(  \beta\right)  $ to the
quality factor $Q_{j}=-\frac{1}{2}\frac{\left\vert \operatorname{Re}\zeta
_{j}\left(  \beta\right)  \right\vert }{\operatorname{Im}\zeta_{j}\left(
\beta\right)  }$ of any one of its eigenmodes (not shown is $Q_{j}%
\rightarrow\infty$ as $\beta\rightarrow0$). For $\beta\geqslant\beta
_{0}\approx0.57282$ (critical damping) and $j=1,2$, the eigenmodes with
eigenvalue $\zeta_{j}\left(  \beta\right)  $ are overdamped since
$\operatorname{Re}\zeta_{j}\left(  \beta\right)  =0$. The overdamping
phenomenon is manifested graphically in (a), (c) with the cusps in the curves
at the intersection of vertical dotted line $\beta=\beta_{0}$ and (b), (d)
with the curves on the line $Q=0$ for $\beta\geqslant\beta_{0}$. (a)
$\operatorname{Im}\zeta_{1}\left(  \beta\right)  $ vs. $\beta$. (b) $Q_{1}$
vs. $\beta$. Due to overdamping, $Q_{1}=0$ for $\beta\geq\beta_{0}$. (c)
$\operatorname{Im}\zeta_{j}\left(  \beta\right)  $ vs. $\beta$, $j=2,3,4$. The
curves $\operatorname{Im}\zeta_{2}\left(  \beta\right)  $ and
$\operatorname{Im}\zeta_{4}\left(  \beta\right)  $ are shown in blue and
green, respectively. Due to eigenvalue symmetry, the curves $\operatorname{Im}%
\zeta_{3}\left(  \beta\right)  $ and $\operatorname{Im}\zeta_{4}\left(
\beta\right)  $ cannot be distinguished in this plot and $\operatorname{Im}%
\zeta_{2}\left(  \beta\right)  =\operatorname{Im}\zeta_{1}\left(
\beta\right)  $ for $\beta\leq\beta_{0}$. (d) $Q_{j}$ vs. $\beta$, $j=2,3,4$.
The curves $Q=Q_{2}$ and $Q=Q_{4}$ are shown in blue and green, respectively,
and because of overdamping $Q_{2}=0$ for $\beta\geqslant\beta_{0}$. Due to
eigenvalue symmetry, the curves $Q=Q_{3}$ and $Q=Q_{4}$ cannot be
distinguished in this plot.}%
\label{Figc2}%
\end{center}
\end{figure}
%EndExpansion
%

%TCIMACRO{\FRAME{ftbpFU}{6.0424in}{6.0251in}{0pt}{\Qcb{(a)-(d) A comparison of
%the real and imaginary parts of the low-loss eigenvalue $\zeta_{3}\left(
%\beta\right)  $ to its truncated asymptotic expansion $\widetilde{\zeta}%
%_{3}\left(  \beta\right)  =\rho_{3}-\QTR{rm}{i}d_{3}\beta^{-1}$ for the values
%$\rho_{3}$, $d_{3}$ predicted by our theory. As evident from these plots,
%$\zeta_{3}\left(  \beta\right)  \approx\widetilde{\zeta}_{3}\left(
%\beta\right)  $ for $\beta$ large.\ (a) $\operatorname{Re}\zeta_{3}\left(
%\beta\right)  $ vs. $\beta$. (b) $\operatorname{Re}\zeta_{3}\left(
%\beta\right)  $ vs. $\beta$ (the solid line) and $\operatorname{Re}%
%\widetilde{\zeta}_{3}\left(  \beta\right)  $ vs. $\beta$\ (the diagonal
%crosses). (c) $\operatorname{Im}\zeta_{3}\left(  \beta\right)  $ vs. $\beta$.
%(d) $\operatorname{Im}\zeta_{3}\left(  \beta\right)  $ vs. $\beta$ (the solid
%line) and $\operatorname{Im}\widetilde{\zeta}_{3}\left(  \beta\right)  $ vs.
%$\beta$ (the diagonal crosses).}}{\Qlb{Figc3}}{eigzeta3vsasympplot.eps}%
%width 6.0424in;  height 6.0251in;  depth 0pt;  original-width 6.0277in;
%original-height 6.0277in;  cropleft "0";  croptop "1";  cropright "1";
%cropbottom "0";  filename 'EigZeta3VsAsympPlot.eps';file-properties "XNPEU";}}
%}%
%BeginExpansion
\begin{figure}
[ptb]
\begin{center}
\includegraphics[
height=6.0251in,
width=6.0424in
]%
{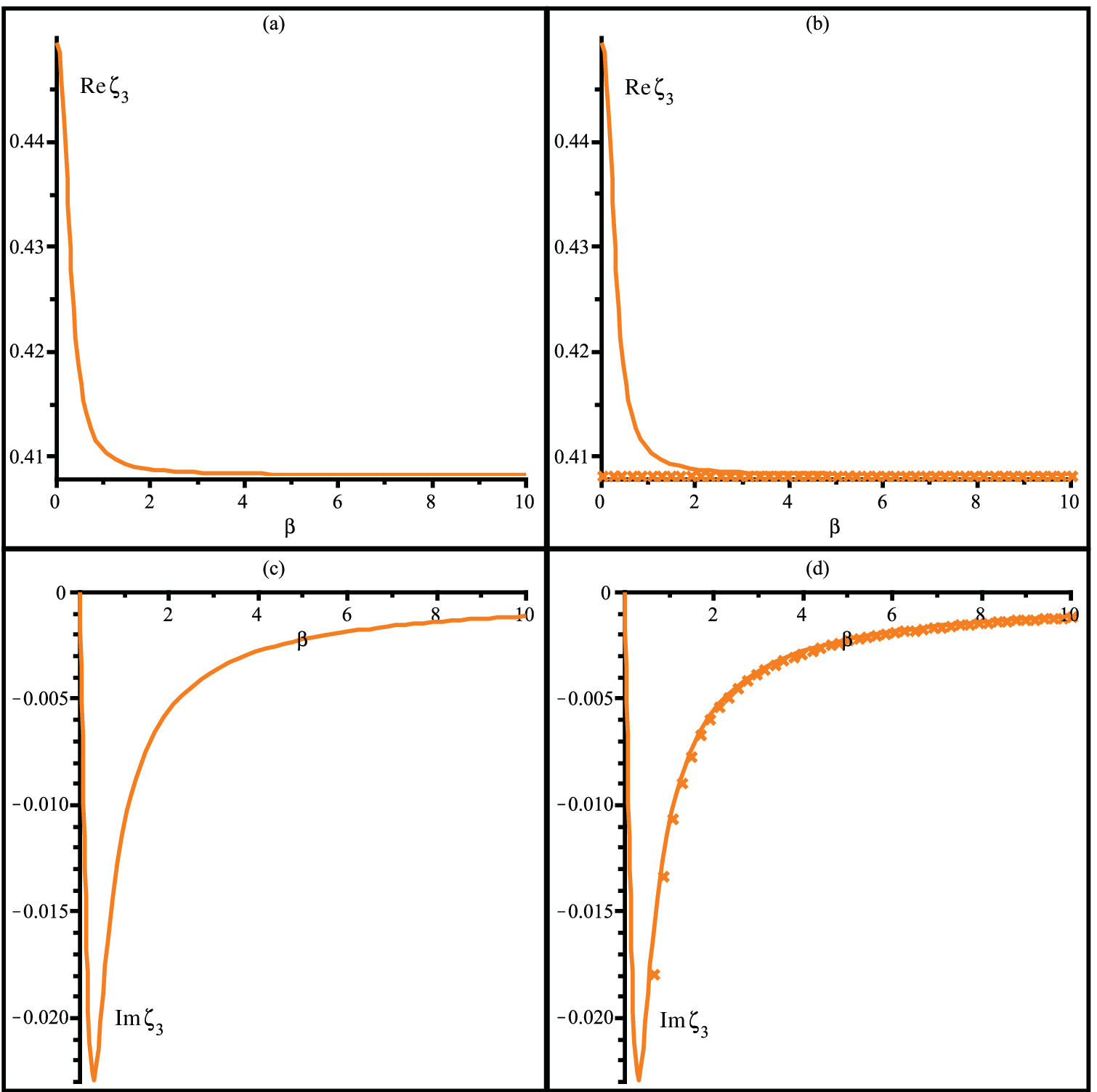}%
\caption{(a)-(d) A comparison of the real and imaginary parts of the low-loss
eigenvalue $\zeta_{3}\left(  \beta\right)  $ to its truncated asymptotic
expansion $\widetilde{\zeta}_{3}\left(  \beta\right)  =\rho_{3}-\mathrm{i}%
d_{3}\beta^{-1}$ for the values $\rho_{3}$, $d_{3}$ predicted by our theory.
As evident from these plots, $\zeta_{3}\left(  \beta\right)  \approx
\widetilde{\zeta}_{3}\left(  \beta\right)  $ for $\beta$ large.\ (a)
$\operatorname{Re}\zeta_{3}\left(  \beta\right)  $ vs. $\beta$. (b)
$\operatorname{Re}\zeta_{3}\left(  \beta\right)  $ vs. $\beta$ (the solid
line) and $\operatorname{Re}\widetilde{\zeta}_{3}\left(  \beta\right)  $ vs.
$\beta$\ (the diagonal crosses). (c) $\operatorname{Im}\zeta_{3}\left(
\beta\right)  $ vs. $\beta$. (d) $\operatorname{Im}\zeta_{3}\left(
\beta\right)  $ vs. $\beta$ (the solid line) and $\operatorname{Im}%
\widetilde{\zeta}_{3}\left(  \beta\right)  $ vs. $\beta$ (the diagonal
crosses).}%
\label{Figc3}%
\end{center}
\end{figure}
%EndExpansion

\paragraph{On Figure \ref{Figc2}.}

In Fig. \ref{Figc2} is a series of plots which compares the imaginary part of
each of the eigenvalues $\zeta_{j}\left(  \beta\right)  $, $1\leq j\leq4$ of
the system operator $A(\beta)$ for the electric circuit in Fig. \ref{Figc1} to
the quality factor $Q$ of their corresponding eigenmodes. \ To plot the
quality factor as a function of the loss parameter $\beta$ we have used
formula (\ref{quwz3}). \ As is evident by this figure, there is clearly modal
dichotomy caused by dissipation.

Indeed, for the high-loss eigenpair $\zeta_{1}\left(  \beta\right)  $,
$w_{1}\left(  \beta\right)  $ we can see in plots (a)--(b) that as the loss
parameter $\beta$ grows large so too does $\operatorname{Im}$ $\zeta
_{1}\left(  \beta\right)  $ whereas the quality factor $Q[w_{1}\left(
\beta\right)  ]=-\frac{1}{2}\frac{\left\vert \operatorname{Re}\zeta_{1}\left(
\beta\right)  \right\vert }{\operatorname{Im}\zeta_{1}\left(  \beta\right)  }$
of the mode goes to zero as predicted by our general theory (cf. (\ref{hlre8})
of Proposition \ref{Prop1} and (\ref{hlrq2}) of Proposition \ref{Prop5}). \ In
fact, by our results on the overdamping phenomenon described in Corollary
\ref{Cor0b} we know that it is exactly zero for all $\beta\geq\beta_{0}$,
where $\beta_{0}$ denotes the boundary of the overdamped regime as discussed
in Remark \ref{rbdod}. \ For the fixed values in (\ref{circ23}), $\beta
_{0}\approx0.57282$ and we have place a vertical dotted line in each of the
plots in Fig. \ref{Figc2} to indicate this boundary.

The behavior is quite different for the low-loss eigenpairs $\zeta_{j}\left(
\beta\right)  $, $w_{j}\left(  \beta\right)  $, $2\leq j\leq4$. \ The plots
(c)--(d) show that as the loss parameter $\beta$ grows large, the values
$\operatorname{Im}$ $\zeta_{j}\left(  \beta\right)  $, $j=2,3,4$ all become
small with the quality factors $Q[w_{3}\left(  \beta\right)  ]=-\frac{1}%
{2}\frac{\left\vert \operatorname{Re}\zeta_{3}\left(  \beta\right)
\right\vert }{\operatorname{Im}\zeta_{3}\left(  \beta\right)  }$ and
$Q[w_{4}\left(  \beta\right)  ]=-\frac{1}{2}\frac{\left\vert \operatorname{Re}%
\zeta_{4}\left(  \beta\right)  \right\vert }{\operatorname{Im}\zeta_{4}\left(
\beta\right)  }$ of the eigenmodes $w_{3}\left(  \beta\right)  $ and
$w_{4}\left(  \beta\right)  $ becoming large as predicted by our theory (cf.
(\ref{hlre8}) of Proposition \ref{Prop1}, ( \ref{hlrq4}) of Proposition
\ref{Prop5}, and formulas (\ref{circ19})). \ The quality factor $Q[w_{2}%
\left(  \beta\right)  ]=-\frac{1}{2}\frac{\left\vert \operatorname{Re}%
\zeta_{2}\left(  \beta\right)  \right\vert }{\operatorname{Im}\zeta_{2}\left(
\beta\right)  }$ of the low-loss mode $w_{2}\left(  \beta\right)  $ becomes
zero for $\beta\geq\beta_{0}$, again a fact which is predicted for this
electric circuit from the overdamping phenomenon described in Corollary
\ref{Cor0b}.

\paragraph{On Figure \ref{Figc3}.}

Figure \ref{Figc3} compares the low-loss eigenvalue $\zeta_{3}\left(
\beta\right)  $ of the system operator $A\left(  \beta\right)  $ to the
truncation $\widetilde{\zeta}_{3}\left(  \beta\right)  $ of its asymptotic
expansion as predicted by our theory in (\ref{circ18}) and (\ref{circ19}),
namely,%
\[
\zeta_{3}\left(  \beta\right)  \approx\widetilde{\zeta}_{3}\left(
\beta\right)  =\rho_{3}-\mathrm{i}d_{3}\beta^{-1}=\sqrt{\Phi_{11}^{2}%
+\Phi_{12}^{2}}-\mathrm{i}\frac{\frac{1}{2}\tau\Phi_{12}^{2}\left(  \Phi
_{11}+\Phi_{22}\right)  ^{2}}{\Phi_{11}^{2}+\Phi_{12}^{2}}\beta^{-1},\text{
}\beta\gg1\text{.}%
\]
The plots (a) and (c) in the figure are the real and imaginary parts,
respectively, of the eigenvalue $\zeta_{3}\left(  \beta\right)  $ and plots
(b) and (d) are the real and imaginary parts, respectively, of both the
eigenvalue $\zeta_{3}\left(  \beta\right)  $ and the truncation of its
asymptotic expansion $\widetilde{\zeta}_{3}\left(  \beta\right)  $.

\paragraph{On Figure \ref{Figc4}}%

\begin{figure}
[ptb]
\begin{center}
\includegraphics[
height=3.045in,
width=6.0701in
]%
{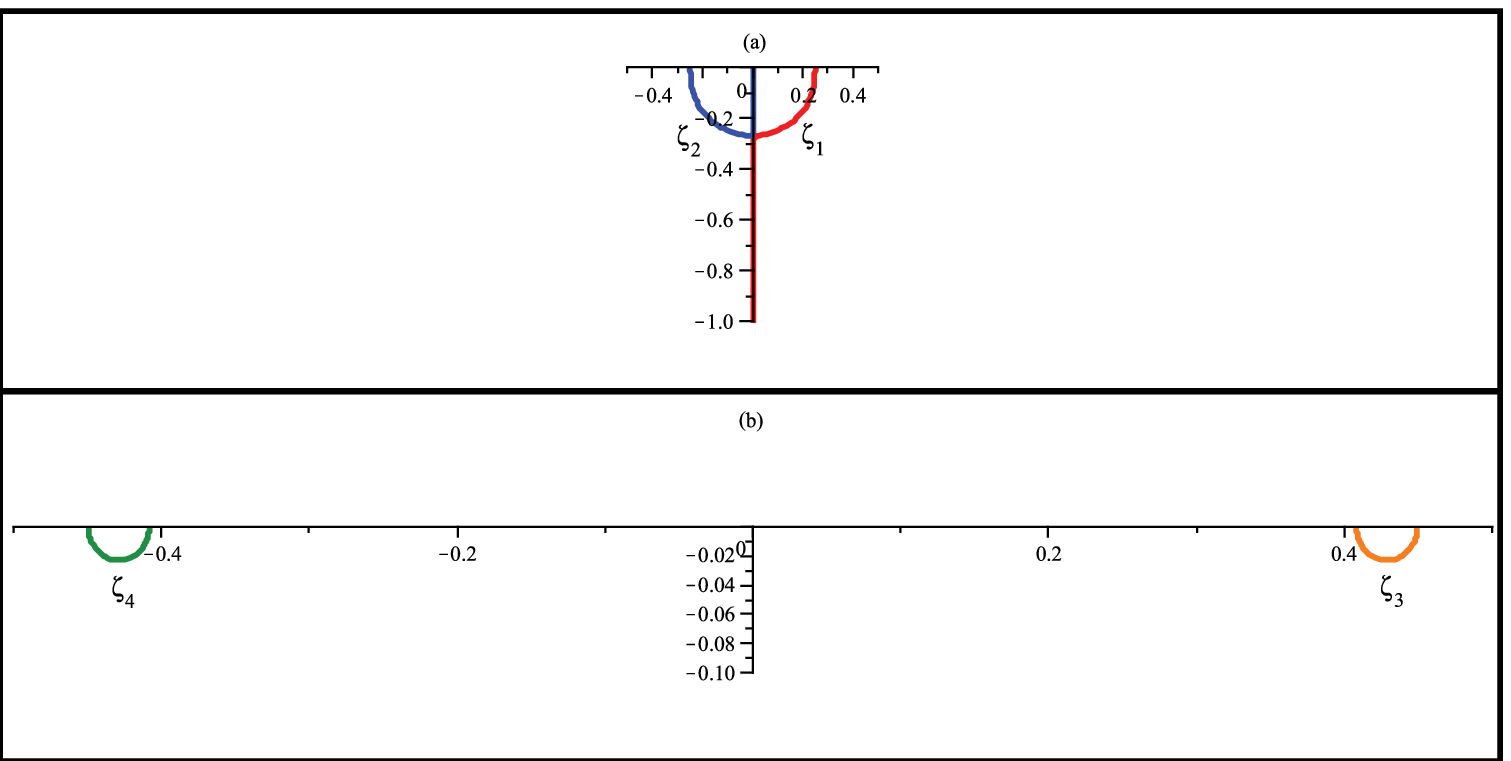}%
\caption{(a)-(b) Image in the complex plane of the eigenvalues $\zeta
_{j}\left(  \beta\right)  $, $1\leq j\leq4$ of the system operator $A(\beta)$
for the electric circuit in Fig. \ref{Figc1} with the values (\ref{circ22_5})
and (\ref{circ23}). The power dissipation condition implies $\operatorname{Im}%
\zeta_{j}\left(  \beta\right)  \leq0$ for $\beta\geq0$, as evident in the
figure. (a) Image of the high-loss eigenvalue $\zeta_{1}\left(  \beta\right)
$ and low-loss eigenvalue $\zeta_{2}\left(  \beta\right)  $ displayed in red
and blue, respectively. For the purpose of comparison, the view is restricted
to a box around the image of $\zeta_{2}\left(  \beta\right)  $. Off the
imaginary axis the blue curve is symmetric about this axis to the red curve
due to the eigenvalue symmetry $\zeta_{2}\left(  \beta\right)  =-\overline
{\zeta_{1}\left(  \beta\right)  }$ for $\beta\leq\beta_{0}\approx0.57282$.
\ The two curves intersect, for $\beta=\beta_{0}$, on the negative imaginary
axis and due to overdamping stay there for all $\beta\geqslant\beta_{0}$.
Moreover, $\operatorname{Im}\zeta_{1}\left(  \beta\right)  \rightarrow-\infty$
and $\operatorname{Im}\zeta_{2}\left(  \beta\right)  \rightarrow0$ as
$\beta\rightarrow\infty$. (b) Image of the low-loss eigenvalues $\zeta
_{3}\left(  \beta\right)  $ and $\zeta_{4}\left(  \beta\right)  $ displayed in
orange and green, respectively. The green curve is symmetric about the
imaginary axis to the orange curve due to the eigenvalue symmetry $\zeta
_{4}\left(  \beta\right)  =-\overline{\zeta_{3}\left(  \beta\right)  }$ for
all $\beta\geqslant0$. \ Moreover, $\operatorname{Im}\zeta_{4}\left(
\beta\right)  =\operatorname{Im}\zeta_{3}\left(  \beta\right)  \rightarrow0$
and $-\operatorname{Re}\zeta_{4}\left(  \beta\right)  =\operatorname{Re}%
\zeta_{3}\left(  \beta\right)  \rightarrow\rho_{3}$ as $\beta\rightarrow
\infty$, where $\rho_{3}\approx0.40825$.}%
\label{Figc4}%
\end{center}
\end{figure}
%EndExpansion

In Fig. \ref{Figc4} we have the images in the complex plane of the eigenvalues
$\zeta_{j}\left(  \beta\right)  $, $1\leq j\leq4$. \ This figure displays the
spectral symmetry and overdamping phenomena as predicted in Proposition
\ref{Prop0} and Corollaries \ref{Cor0a}, \ref{Cor0b}. \ According to Lemma
\ref{apxlm} and as evident in the figure, the eigenvalues lie in the lower
half-plane. \ By Theorem \ref{Thm2}, these eigenvalues converge to the real
axis as $\beta\rightarrow0$ and, in particular, to the eigenvalues of the
frequency operator $\Omega$.

In plot (a) we see the images of the eigenvalues $\zeta_{j}\left(
\beta\right)  $, $j=1,2$ in the complex plane. We observe that overdamping
does occur but only for these two eigenvalues. Indeed, as the loss parameter
$\beta$ increase from zero these two eigenvalues eventually merge on the
negative imaginary axis when $\beta=\beta_{0}\approx0.57282$ and stay on this
axis for all $\beta\geq\beta_{0}$ with $\operatorname{Im}\zeta_{1}\left(
\beta\right)  \rightarrow-\infty$ and $\operatorname{Im}\zeta_{2}\left(
\beta\right)  \rightarrow0$ as $\beta\rightarrow\infty$.

In plot (b) we see the images of the eigenvalues $\zeta_{j}\left(
\beta\right)  $, $j=3,4$ in the complex plane. \ This plot shows the
eigenvalue symmetry for the system operator $A\left(  \beta\right)  $ for the
electric circuit which forces these eigenvalues to satisfy $\zeta_{4}\left(
\beta\right)  =-\overline{\zeta_{3}\left(  \beta\right)  }$, to lie off the
imaginary axis, and to be in the lower open half-plane for all $\beta>0$ with
$\operatorname{Im}\zeta_{j}\left(  \beta\right)  \rightarrow0$ as
$\beta\rightarrow\infty$ for $j=3,4$.

\section{Perturbation analysis of the system operator\label{spsyo}}

This and the following sections are devoted to a rigourous perturbation
analysis of the eigenvalues and eigenvectors of the system operator from
(\ref{mvt7}),%
\begin{equation}
A\left(  \beta\right)  :=\Omega-\mathrm{i}\beta B,\text{$\quad$}\beta\geq0,
\label{syop1}%
\end{equation}
in both the high-loss regime, $\beta\gg1$, and the low-loss regime, $\beta
\ll1$. \ We mainly focus on the high-loss regime and our goal is to develop a
mathematical framework based on perturbation theory for an asymptotic analytic
description of the effects dissipation have on the system (\ref{mvt7})
including modal dichotomy, i.e., splitting of eigenmodes into two distinct
classes according to their dissipative properties: high-loss and low-loss
modes. \ This framework and its rigorous analysis provides for insights into
the mechanism of losses in composite systems and in ways to achieve
significant absorption suppression.

The rest of this section is organized as follows. \ We first recall the basic
assumptions, definitions, and notations from earlier in this paper regarding
the system operator (\ref{syop1}), the quality factor $Q$ and power of energy
dissipation $W_{\text{\textrm{dis}}}$ associated with its modes. \ In the next
section we state our main results on the perturbation analysis of the
eigenvalues and eigenvectors for this operator $A\left(  \beta\right)  $.
\ The result for the high-loss regime, $\beta\gg1$, and the low-loss regime,
$\beta\ll1$, are placed in separate sections. \ Finally, we prove the
statement of our main results in Section \ref{sprfr}

The system operator $A\left(  \beta\right)  $ in (\ref{syop1}), for each value
of the loss parameter $\beta$, is a linear operator on the finite dimensional
Hilbert space $H$ with $N:=\dim H$ and scalar product $\left(  \cdot
,\cdot\right)  $. \ The frequency operator $\Omega$ and the operator
associated with dissipation $B$ are self-adjoint operators on $H$. \ The
operator $B$ satisfies the power dissipation condition (\ref{mvt8}) and the
loss fraction condition (\ref{hbu3}), namely,%
\begin{equation}
B\geq0,\text{$\quad$}0<\delta_{B}<1 \label{syop2}%
\end{equation}
where $N_{B}:=\operatorname{rank}B$ denotes the rank of the operator $B$ and
$\delta_{B}:=\frac{N_{B}}{N}$ is referred to as the fraction of high-loss modes.

The range of the operator $B$, i.e., the loss subspace, is denoted by $H_{B}$
and the orthogonal projection onto this space is denoted by $P_{B}$. \ It
follows immediately from these definitions and the fact $B$ is self-adjoint
that%
\begin{equation}
H=H_{B}\oplus H_{B}^{\perp} \label{syop3}%
\end{equation}
where $H_{B}^{\perp}$, i.e., the no-loss subspace, is the orthogonal
complement of $H_{B}$ in $H$ and is the kernel of $B$ with the orthogonal
projection onto this space given by $P_{B}^{\perp}:=I-P_{B}$. \ In
particular,
\begin{align}
&  H_{B}=\operatorname{ran}B,\text{$\quad$}N_{B}=\dim H_{B},\label{syop4}\\
&  H_{B}^{\perp}=\ker B,\text{$\quad$}N-N_{B}=\dim H_{B}^{\perp}.\nonumber
\end{align}

The energy $U\left[  w\right]  $, power of energy dissipation
$W_{\text{\textrm{dis}}}\left[  w\right]  $, and the quality factor $Q\left[
w\right]  $ of an eigenvector $w$ of the system operator $A\left(
\beta\right)  $ with eigenvalue $\zeta$ is
\begin{equation}
U\left[  w\right]  =\frac{1}{2}\left(  w,w\right)  ,\text{ \ \ }%
W_{\text{\textrm{dis}}}\left[  w\right]  =\left(  w,\beta Bw\right)  ,\text{
\ \ }Q\left[  w\right]  =\left\vert \operatorname{Re}\zeta\right\vert
\frac{\frac{1}{2}\left(  w,w\right)  }{\left(  w,\beta Bw\right)  },
\label{syop5}%
\end{equation}
where $Q\left[  w\right]  $ is said to be finite if $W_{\text{\textrm{dis}}%
}\left[  w\right]  \not =0$. \ In Appendix \ref{apxqf} we show that%
\begin{equation}
\operatorname{Im}\zeta=-\frac{\left(  w,\beta Bw\right)  }{\left(  w,w\right)
},\text{ \ \ }W_{\text{\textrm{dis}}}\left[  w\right]  =-2\operatorname{Im}%
\zeta U\left[  w\right]  ,\text{ \ \ }Q\left[  w\right]  =-\frac{1}{2}%
\frac{\left\vert \operatorname{Re}\zeta\right\vert }{\operatorname{Im}\zeta},
\label{syop6}%
\end{equation}
where $Q\left[  w\right]  $ is finite if and only if $\operatorname{Im}%
\zeta\not =0$.

\subsection{The high-loss regime\label{shlre}}

We begin this section with our results on the perturbation analysis of the
eigenvalues and eigenvectors for this operator $A\left(  \beta\right)  $ for
the high-loss regime in which $\beta\gg1$.

\begin{theorem}
[eigenmodes dichotomy]\label{Thm1}Let $\mathring{\zeta}_{j}$, $1\leq j\leq
N_{B}$ be an indexing of all the nonzero eigenvalues of $B$ (counting
multiplicities) where $N_{B}=\operatorname{rank}B$. \ Then for the high-loss
regime $\beta\gg1$, the system operator $A(\beta)$ is diagonalizable and there
exists a complete set of eigenvalues $\zeta_{j}\left(  \beta\right)  $ and
eigenvectors $w_{j}\left(  \beta\right)  $ satisfying%
\begin{equation}
A\left(  \beta\right)  w_{j}\left(  \beta\right)  =\zeta_{j}\left(
\beta\right)  w_{j}\left(  \beta\right)  ,\text{$\quad$}1\leq j\leq N,
\label{hlre1}%
\end{equation}
which split into two distinct classes of eigenpairs
\begin{align}
&  \text{high-loss}\text{:$\quad$}\zeta_{j}\left(  \beta\right)  ,\text{
}w_{j}\left(  \beta\right)  ,\text{$\quad$}1\leq j\leq N_{B};\label{hlre2}\\
&  \text{low-loss}\text{:$\quad$}\zeta_{j}\left(  \beta\right)  ,\text{ }%
w_{j}\left(  \beta\right)  ,\text{$\quad$}N_{B}+1\leq j\leq N,\nonumber
\end{align}
having the following properties:

\begin{enumerate}
\item[(i)] The high-loss eigenvalues have poles at $\beta=\infty$ whereas
their eigenvectors are analytic at $\beta=\infty$. \ These eigenpairs have the
asymptotic expansions%
\begin{align}
\text{\ }\zeta_{j}\left(  \beta\right)   &  =-\mathrm{i}\mathring{\zeta}%
_{j}\beta+\rho_{j}+O\left(  \beta^{-1}\right)  ,\text{$\quad$}\mathring{\zeta
}_{j}>0,\text{$\quad$}\rho_{j}\in%
%TCIMACRO{\U{211d} }%
%BeginExpansion
\mathbb{R}
%EndExpansion
,\label{hlre3}\\
w_{j}\left(  \beta\right)   &  =\mathring{w}_{j}+O\left(  \beta^{-1}\right)
,\text{$\quad$}1\leq j\leq N_{B}\nonumber
\end{align}
as $\beta\rightarrow\infty$. \ The vectors $\mathring{w}_{j}$, $1\leq j\leq
N_{B}$ form an orthonormal basis of the loss subspace $H_{B}$ and%
\begin{equation}
B\mathring{w}_{j}=\mathring{\zeta}_{j}\mathring{w}_{j},\text{$\quad$}\rho
_{j}=\left(  \mathring{w}_{j},\Omega\mathring{w}_{j}\right)  , \label{hlre4}%
\end{equation}
\ for $1\leq j\leq N_{B}$.

\item[(ii)] The low-loss eigenpairs are analytic at $\beta=\infty$ and have
the asymptotic expansions%
\begin{align}
\text{\ }\zeta_{j}\left(  \beta\right)   &  =\rho_{j}-\mathrm{i}d_{j}%
\beta^{-1}+O\left(  \beta^{-2}\right)  ,\text{$\quad$}\rho_{j}\in%
%TCIMACRO{\U{211d} }%
%BeginExpansion
\mathbb{R}
%EndExpansion
,\text{$\quad$}d_{j}\geq0,\label{hlre5}\\
w_{j}\left(  \beta\right)   &  =\mathring{w}_{j}+w_{j}^{(-1)}\beta
^{-1}+O\left(  \beta^{-2}\right)  ,\text{$\quad$}N_{B}+1\leq j\leq N\nonumber
\end{align}
as $\beta\rightarrow\infty$. \ The vectors $\mathring{w}_{j}$, $N_{B}+1\leq
j\leq N$ form an orthonormal basis of the no-loss subspace $H_{B}^{\perp}$ and%
\begin{equation}
B\mathring{w}_{j}=0,\text{$\quad$}\rho_{j}=\left(  \mathring{w}_{j}%
,\Omega\mathring{w}_{j}\right)  ,\text{$\quad$}d_{j}=\left(  w_{j}%
^{(-1)},Bw_{j}^{(-1)}\right)  \label{hlre6}%
\end{equation}
for $N_{B}+1\leq j\leq N$.
\end{enumerate}
\end{theorem}

\begin{corollary}
[eigenmode expulsion]\label{Cor1}The projections of the eigenvectors
$w_{j}\left(  \beta\right)  $, $1\leq j\leq N$ onto the loss subspace $H_{B}$
and the no-loss subspace $H_{B}^{\perp}$ have the asymptotic expansions%
\begin{align}
P_{B}w_{j}\left(  \beta\right)   &  =\mathring{w}_{j}+O\left(  \beta
^{-1}\right)  ,\text{$\quad$}P_{B}^{\perp}w_{j}\left(  \beta\right)  =O\left(
\beta^{-1}\right)  ,\text{$\quad$}1\leq j\leq N_{B};\label{hlre7}\\
P_{B}^{\perp}w_{j}\left(  \beta\right)   &  =\mathring{w}_{j}+O\left(
\beta^{-1}\right)  ,\text{$\quad$}P_{B}w_{j}\left(  \beta\right)  =O\left(
\beta^{-1}\right)  ,\text{$\quad$}N_{B}+1\leq j\leq N\nonumber
\end{align}
as $\beta\rightarrow\infty$.
\end{corollary}

\begin{proposition}
[eigenfrequency expansions]\label{Prop1}The functions $\operatorname{Re}%
\zeta_{j}\left(  \beta\right)  $, $1\leq j\leq N$ at $\beta=\infty$ are
analytic and their series expansions contain only even powers of $\beta^{-1}$.
\ The functions $\operatorname{Im}\zeta_{j}\left(  \beta\right)  $, $1\leq
j\leq N$ at $\beta=\infty$ have poles for $1\leq j\leq N_{B}$, are analytic
for $N_{B}+1\leq j\leq N$, and their series expansions contain only odd powers
of $\beta^{-1}$. \ Moreover, they have the asymptotic expansions%
\begin{align}
\operatorname{Re}\zeta_{j}(\beta)  &  =\rho_{j}+O\left(  \beta^{-2}\right)
,\text{ \ }\operatorname{Im}\zeta_{j}\left(  \beta\right)  =-\mathring{\zeta
}_{j}\beta+O\left(  \beta^{-1}\right)  ,\text{$\quad$}1\leq j\leq
N_{B};\label{hlre8}\\
\operatorname{Re}\zeta_{j}(\beta)  &  =\rho_{j}+O\left(  \beta^{-2}\right)
,\text{ \ }\operatorname{Im}\zeta_{j}\left(  \beta\right)  =-d_{j}\beta
^{-1}+O\left(  \beta^{-3}\right)  ,\text{ \ }N_{B}+1\leq j\leq N\nonumber
\end{align}
as $\beta\rightarrow\infty$.
\end{proposition}

\begin{proposition}
\label{Prop2}For each $j=1,\ldots,N$ and in the high-loss regime $\beta\gg1$,
the following statements are true:

\begin{enumerate}
\item If $1\leq j\leq N_{B}$ then $\operatorname{Im}\zeta_{j}\left(
\beta\right)  <0 $.

\item If $N_{B}+1\leq$ $j\leq N$ then either $\operatorname{Im}\zeta
_{j}\left(  \beta\right)  \equiv0$ or $\operatorname{Im}\zeta_{j}\left(
\beta\right)  <0$. \ Moreover, $\operatorname{Im}\zeta_{j}\left(
\beta\right)  \equiv0$ if and if $\zeta_{j}\left(  \beta\right)  \equiv
\rho_{j}$.

\item If $N_{B}+1\leq$ $j\leq N$ then $\mathring{w}_{j}\not \in \ker\left(
\rho_{j}I-\Omega\right)  $ if and only if $d_{j}\not =0$.

\item If $N_{B}+1\leq$ $j\leq N$ then $\mathring{w}_{j}\not \in \ker\Omega$ if
and only if $\rho_{j}\not =0$ or $d_{j}\not =0$.
\end{enumerate}
\end{proposition}

\begin{remark}
\label{Rem1}Typically, one can expect that the asymptotic expansion of the
low-loss eigenvalues have $d_{j}\not =0$, for $j=N_{B}+1,\ldots,N$. \ Indeed,
if this were not the case then Theorem \ref{Thm1} and the previous proposition
tell us that the intersection of one of the eigenspaces of the operator
$\Omega$ with the kernel of the operator $B$ would contain a nonzero vector.
\ And this is obviously atypical behavior.
\end{remark}

\begin{remark}
An important subspace which arises in studies of open systems in \cite[pp.
27-28]{Liv} as well as in \cite[Sec. 4.1]{FigShi} is
\begin{equation}
\mathcal{O}_{\Omega}\left(  \operatorname{ran}B\right)  =\operatorname*{Span}%
\left\{  \Omega^{n}Bu:u\in H,\ n=0,1,\ldots\right\}  , \label{zej2}%
\end{equation}
where it is called the orbit and is the smallest subspace of $H$ containing
$\operatorname{ran}B$ that is invariant under $\Omega$. It is shown there that
the orthogonal complement $\mathcal{O}_{\Omega}\left(  \operatorname{ran}%
B\right)  ^{\bot}=H\ominus\mathcal{O}_{\Omega}\left(  \operatorname{ran}%
B\right)  =\cap_{n\geq0}\ker\left(  B\Omega^{n}\right)  $ is the subspace
which is invariant with respect to the operator $\Omega-\mathrm{i}\beta B$ and
the restriction $\left.  \Omega-\mathrm{i}\beta B\right\vert _{H_{\Omega
,B}^{\bot}}$ is self-adjoint. Consequently, the evolution over this subspace
in entirely decoupled from the operator $B$ and there is no energy dissipation
there. \ Moreover, the importance of the orbit to the perturbation analysis is
that in applications it's often not hard to see that
\begin{equation}
\mathcal{O}_{\Omega}\left(  \operatorname{ran}B\right)  =H. \label{zej3}%
\end{equation}
In this case, by Remark \ref{Rem1} we know that the asymptotic expansion of
the low-loss eigenvalues have $d_{j}\not =0$, for $j=N_{B}+1,\ldots,N$.
\end{remark}

For computational purposes the next proposition and its corollary are
important results. \ We first recall some notation. \ The orthogonal
projection onto the loss subspace $\ker B=H_{B}$ is $P_{B}$ and $P_{B}^{\bot}$
is the orthogonal projection onto the no-loss subspace $\operatorname{ran}%
B=H_{B}^{\bot}$. \ Then from (\ref{hbcom2}) we will need the operators
$B_{2}=\left.  P_{B}BP_{B}\right\vert _{H_{B}}:H_{B}\rightarrow H_{B}$,
$\Omega_{1}=\left.  P_{B}^{\bot}\Omega P_{B}^{\bot}\right\vert _{H_{B}^{\bot}%
}:H_{B}^{\bot}\rightarrow H_{B}^{\bot}$, and $\Theta=\left.  P_{B}\Omega
P_{B}^{\bot}\right\vert _{H_{B}^{\bot}}:H_{B}^{\bot}\rightarrow H_{B}$ whose
adjoint is $\Theta^{\ast}=\left.  P_{B}^{\bot}\Omega P_{B}\right\vert _{H_{B}%
}:H_{B}\rightarrow H_{B}^{\bot}$.

\begin{proposition}
[asymptotic spectrum]\label{Prop3}The following statements are true:

\begin{enumerate}
\item In the asymptotic expansions (\ref{hlre3}) for the high-loss eigenpairs,
the coefficients $\mathring{\zeta}_{j}$, $\mathring{w}_{j}$, $1\leq j\leq
N_{B}$ form a complete set of eigenvalues and orthonormal eigenvectors for the
operator $B_{2}$ with%
\begin{equation}
B_{2}\mathring{w}_{j}=\mathring{\zeta}_{j}\mathring{w}_{j},\text{ \ \ }1\leq
j\leq N_{B}.
\end{equation}
\ In particular, $B_{2}$ is a positive definite operator as is its inverse
$B_{2}^{-1}$, i.e.,%
\begin{equation}
B_{2}>0,\text{ \ \ }B_{2}^{-1}>0\text{.} \label{hlre10}%
\end{equation}

\item In the asymptotic expansions (\ref{hlre5}) for the low-loss eigenpairs,
the coefficients $\rho_{j}$, $\mathring{w}_{j}$, $N_{B}+1\leq j\leq N$ form a
complete set of eigenvalues and orthonormal eigenvectors for the self-adjoint
operator $\Omega_{1}$ with%
\begin{equation}
\Omega_{1}\mathring{w}_{j}=\rho_{j}\mathring{w}_{j},\text{ \ \ }N_{B}+1\leq
j\leq N. \label{hlre9}%
\end{equation}

\item The coefficients\ $d_{j}$, $N_{B}+1\leq j\leq N$ in the asymptotic
expansions of the low-loss eigenvalues are given by the formulas%
\begin{equation}
d_{j}=\left(  \mathring{w}_{j},\Theta^{\ast}B_{2}^{-1}\Theta\mathring{w}%
_{j}\right)  ,\ \ N_{B}+1\leq j\leq N. \label{hlre11}%
\end{equation}

\end{enumerate}
\end{proposition}

\begin{corollary}
[computing expansions]\label{Cor2}The following statements give sufficient
conditions that allow computation of $\mathring{w}_{j}$, $\mathring{\zeta}%
_{j}$, $\rho_{j}$, and $d_{j}$ in the asymptotic expansions of the eigenpairs:

\begin{enumerate}
\item If the eigenvalues of $B_{2}$ are distinct and $\varsigma_{j}$, $1\leq
j\leq N_{B}$ is any indexing of these eigenvalues then in Theorem \ref{Thm1},
after a possible reordering of the high-loss eigenpairs in (\ref{hlre2}), the
coefficients in the asymptotic expansions (\ref{hlre3}) are uniquely
determined by the relations
\[
\mathring{\zeta}_{j}=\varsigma_{j},\text{ \ \ }B_{2}\mathring{w}_{j}%
=\mathring{\zeta}_{j}\mathring{w}_{j},\text{$\quad$}\left\vert \left\vert
\mathring{w}_{j}\right\vert \right\vert =1,\text{ \ \ }\rho_{j}=\left(
\mathring{w}_{j},\Omega\mathring{w}_{j}\right)  ,\text{$\quad$}1\leq j\leq
N_{B}\text{.}%
\]

\item If the eigenvalues of $\Omega_{1}$ are distinct and $\varrho_{j}$,
$N_{B}+1\leq j\leq N$ is any indexing of these eigenvalues then in Theorem
\ref{Thm1}, after a possible reordering of the low-loss eigenpairs in
(\ref{hlre2}), the coefficients in the asymptotic expansions (\ref{hlre5}) are
uniquely determined by the relations
\[
\rho_{j}=\varrho_{j},\text{ \ \ }\Omega_{1}\mathring{w}_{j}=\rho_{j}%
\mathring{w}_{j},\text{ \ \ }\left\vert \left\vert \mathring{w}_{j}\right\vert
\right\vert =1,\text{ \ \ }d_{j}=\left(  \mathring{w}_{j},\Theta^{\ast}%
B_{2}^{-1}\Theta\mathring{w}_{j}\right)  ,\text{$\quad$}N_{B}+1\leq j\leq
N\text{.}%
\]

\end{enumerate}
\end{corollary}

The next two propositions give the asymptotic expansions as $\beta
\rightarrow\infty$ of the energy, power of energy dissipation, and quality
factor for the high-loss and low-loss eigenvectors $w_{j}\left(  \beta\right)
$, $1\leq j\leq N$.

\begin{proposition}
[energy and dissipation]\label{Prop4}The energy for each of the high-loss and
low-loss eigenvectors have the asymptotic expansions%
\begin{equation}
U\left[  w_{j}\left(  \beta\right)  \right]  =\frac{1}{2}+O\left(  \beta
^{-1}\right)  ,\text{$\quad$}1\leq j\leq N \label{hlru1}%
\end{equation}
as $\beta\rightarrow\infty$. \ The\ power of energy dissipation for the
high-loss and low-loss eigenvectors have the asymptotic expansions%
\begin{align}
W_{\text{\textrm{dis}}}\left[  w_{j}\left(  \beta\right)  \right]   &
=\mathring{\zeta}_{j}\beta+O\left(  1\right)  ,\text{$\quad$}1\leq j\leq
N_{B};\label{hlrw1}\\
W_{\text{\textrm{dis}}}\left[  w_{j}\left(  \beta\right)  \right]   &
=d_{j}\beta^{-1}+O\left(  \beta^{-2}\right)  ,\text{$\quad$}N_{B}+1\leq j\leq
N\nonumber
\end{align}
as $\beta\rightarrow\infty$. \ In particular,%
\begin{equation}
\lim_{\beta\rightarrow\infty}W_{\text{\textrm{dis}}}\left[  w_{j}\left(
\beta\right)  \right]  =\left\{
\begin{array}
[c]{cc}%
\infty & \text{if }1\leq j\leq N_{B},\\
0 & \text{if }N_{B}+1\leq j\leq N.
\end{array}
\right.  \label{hlrw2}%
\end{equation}

\end{proposition}

\begin{proposition}
[quality factor]\label{Prop5}For each $j=1,\ldots,N$, the following statements
are true regarding the quality factor of the high-loss and low-loss eigenvectors:

\begin{enumerate}
\item For $\beta\gg1$, the quality factor $Q\left[  w_{j}\left(  \beta\right)
\right]  $ is finite if and only if $\operatorname{Im}\zeta_{j}\left(
\beta\right)  \not \equiv 0$.

\item If the quality factor $Q\left[  w_{j}\left(  \beta\right)  \right]  $ is
finite for $\beta\gg1$ then it is either analytic at $\beta=\infty$ or has a
pole, in either case its series expansion contains only odd powers of
$\beta^{-1}$ and, in particular,
\[
\lim_{\beta\rightarrow\infty}Q\left[  w_{j}\left(  \beta\right)  \right]
=0\text{ or }\infty\text{.}%
\]

\item The quality factor of each high-loss eigenvector is finite for $\beta
\gg1$ and has the asymptotic expansion%
\begin{equation}
Q\left[  w_{j}\left(  \beta\right)  \right]  =\frac{1}{2}\frac{\left\vert
\rho_{j}\right\vert }{\mathring{\zeta}_{j}}\beta^{-1}+O\left(  \beta
^{-3}\right)  ,\text{$\quad$}1\leq j\leq N_{B} \label{hlrq1}%
\end{equation}
as $\beta\rightarrow\infty$. \ In particular,
\begin{equation}
\lim_{\beta\rightarrow\infty}Q\left[  w_{j}\left(  \beta\right)  \right]
=0,\text{$\quad$}1\leq j\leq N_{B}. \label{hlrq2}%
\end{equation}

\item If $j\in\left\{  N_{B}+1,\ldots,N\right\}  $ and $d_{j}\not =0\ $then
the quality factor of the low-loss eigenvector $w_{j}\left(  \beta\right)  $
is finite for $\beta\gg1$ and has the asymptotic expansion%
\begin{equation}
Q\left[  w_{j}\left(  \beta\right)  \right]  =\frac{1}{2}\frac{\left\vert
\rho_{j}\right\vert }{d_{j}}\beta+O\left(  \beta^{-1}\right)  \label{hlrq3}%
\end{equation}
as $\beta\rightarrow\infty$. \ In particular,
\begin{equation}
\lim_{\beta\rightarrow\infty}Q\left[  w_{j}\left(  \beta\right)  \right]
=\left\{
\begin{array}
[c]{cc}%
\infty & \text{if }\rho_{j}\not =0,\\
0 & \text{if }\rho_{j}=0.
\end{array}
\right.  \label{hlrq4}%
\end{equation}

\end{enumerate}
\end{proposition}

\subsection{The low-loss regime}

We now give our results on the perturbation analysis of the eigenvalues and
eigenvectors for the system operator $A\left(  \beta\right)  $ in the low-loss
regime $0\leq\beta\ll1$. \ The focus of this paper is on the high-loss regime
and so we do not try to give results as general as those in previous section.
\ Instead, the goal of this section is to show the fundamentally different
asymptotic behavior in the low-loss regime compared to that of the high-loss regime.

\begin{theorem}
[low-loss asymptotics]\label{Thm2}Let $\omega_{j}$, $1\leq j\leq N$ be an
indexing of all the eigenvalues of $\Omega$ (counting multiplicities). \ Then
for $0\leq\beta\ll1$, the system operator $A\left(  \beta\right)
=\Omega-\mathrm{i}\beta B $ is diagonalizable and there exists a complete set
of eigenvalues $\zeta_{j}\left(  \beta\right)  $ and eigenvectors
$v_{j}\left(  \beta\right)  $ of $A\left(  \beta\right)  $ satisfying%
\begin{equation}
A\left(  \beta\right)  v_{j}\left(  \beta\right)  =\zeta_{j}\left(
\beta\right)  v_{j}\left(  \beta\right)  ,\text{$\quad$}1\leq j\leq N
\label{llre1}%
\end{equation}
with the following properties:

\begin{enumerate}
\item[(i)] The eigenvalues and eigenvectors are analytic at $\beta=0$ and have
the asymptotic expansions%
\begin{gather}
\zeta_{j}\left(  \beta\right)  =\omega_{j}-\mathrm{i}\sigma_{j}\beta+O\left(
\beta^{2}\right)  ,\text{$\quad\omega_{j}\in%
%TCIMACRO{\U{211d} }%
%BeginExpansion
\mathbb{R}
%EndExpansion
,$ \ \ }\sigma_{j}\geq0,\label{llre2}\\
v_{j}\left(  \beta\right)  =u_{j}+O\left(  \beta\right)  ,\text{$\quad$}1\leq
j\leq N\nonumber
\end{gather}
as $\beta\rightarrow0$. \ The vectors $u_{j}$, $1\leq j\leq N$ form an
orthonormal basis of eigenvectors of $\Omega$ and%
\begin{equation}
\Omega u_{j}=\omega_{j}u_{j},\text{$\quad$}\sigma_{j}=\left(  u_{j}%
,Bu_{j}\right)  ,\text{$\quad$}1\leq j\leq N. \label{llre3}%
\end{equation}

\end{enumerate}
\end{theorem}

\begin{corollary}
[energy and dissipation]\label{Cor3}The energy and power of energy dissipation
of these eigenvectors have the asymptotic expansions%
\begin{equation}
U\left[  \upsilon_{j}\left(  \beta\right)  \right]  =\frac{1}{2}+O\left(
\beta^{-1}\right)  ,\text{ \ \ }W_{\text{\textrm{dis}}}\left[  \upsilon
_{j}\left(  \beta\right)  \right]  =\text{\ }\sigma_{j}\beta+O\left(
\beta^{2}\right)  ,\text{ \ \ }1\leq j\leq N \label{llre4}%
\end{equation}
as $\beta\rightarrow0$. \ In particular,%
\begin{equation}
\lim_{\beta\rightarrow0}W_{\text{\textrm{dis}}}\left[  \upsilon_{j}\left(
\beta\right)  \right]  =0,\text{ \ \ }1\leq j\leq N. \label{llre5}%
\end{equation}

\end{corollary}

\begin{corollary}
[quality factor]\label{Cor4}The quality factor of each of these eigenvectors
has the asymptotic expansion%
\begin{equation}
Q\left[  \upsilon_{j}\left(  \beta\right)  \right]  =\frac{1}{2}%
\frac{\left\vert \omega_{j}\right\vert }{\sigma_{j}}\beta^{-1}+O\left(
\beta\right)  , \label{llre6}%
\end{equation}
as $\beta\rightarrow0$, provided $\sigma_{j}\not =0$, in which case it has the
limiting behavior%
\begin{equation}
\lim_{\beta\rightarrow0}Q\left[  \upsilon_{j}\left(  \beta\right)  \right]
=\left\{
\begin{array}
[c]{cc}%
\infty & \text{if }\omega_{j}\not =0,\\
0 & \text{if }\omega_{j}=0.
\end{array}
\right.  \label{llre7}%
\end{equation}

\end{corollary}

\begin{remark}
Typically, one can expect that the asymptotic expansion of these eigenvalues
have $\sigma_{j}\not =0$, for $1\leq j\leq N$. \ Indeed, if this were not the
case then it would follow from the assumption $B\geq0$ and (\ref{llre3}) of
Theorem \ref{Thm2} that the intersection of one of the eigenspaces of the
operator $\Omega$ with the kernel of the operator $B$ would contain a nonzero
vector. \ And this is obviously atypical behavior as mentioned previously in
Remark \ref{Rem1}.
\end{remark}

These results show that in the low-loss regime $0\leq\beta\ll1$, all the modes
behave as low-loss modes since the power of energy dissipation is small and
typically the quality factor is very high. \ In contrast, the high-loss regime
$\beta\gg1$ has both a fraction $0<\delta_{B}<1$ of high-loss modes and a
fraction $0<1-\delta_{B}<1$ of low-loss modes. \ The behaviour of the low-loss
modes in either regime is similar whereas the behavior of the high-loss modes
has the opposite behavior with power of energy dissipation large and quality
factor always small.

\section{Perturbation analysis of a system subjected to harmonic forces
\label{splhf}}

In this section we give an asymptotic description of the stored energy, power
of dissipated energy, and quality factor for a harmonic solution
$\upsilon(t)=\upsilon e^{-\mathrm{i}\omega t}$ of the system (\ref{mvt7}) in
the high-loss regime $\beta\gg1$ subjected to a harmonic external force
$f\left(  t\right)  =fe^{-\mathrm{i}\omega t}$ with nonzero amplitude $f\in H$
and frequency $\omega\in%
%TCIMACRO{\U{211d} }%
%BeginExpansion
\mathbb{R}
%EndExpansion
$. \ We will state our main results in this section but hold off on their
proofs until Section \ref{sprfr}.

To begin we recall that according to (\ref{vzet3})--(\ref{vtez5}), assuming
$\omega$ is not in the resolvent set of the system operator $A\left(
\beta\right)  =\Omega-\mathrm{i}\beta B$, there is a unique harmonic solution
$\upsilon(t)=\upsilon e^{-\mathrm{i}\omega t}$ to the system (\ref{mvt7}) with
the harmonic force $f\left(  t\right)  =fe^{-\mathrm{i}\omega t}$ whose
amplitude $\upsilon$ is given by%
\begin{gather}
\upsilon=\mathfrak{A}\left(  \omega\right)  f=\mathrm{i}\left[  \omega
I-(\Omega-\mathrm{i}\beta B)\right]  ^{-1}f,\label{palhf1}\\
\mathfrak{A}\left(  \omega\right)  =\mathrm{i}\left[  \omega I-A\left(
\beta\right)  \right]  ^{-1},\nonumber
\end{gather}
where $\mathfrak{A}\left(  \omega\right)  $ is the admittance operator.\ 

As was introduced in Section \ref{SecQualFacHarm}, the stored energy $U$,
power of dissipated energy $W_{\text{\textrm{dis}}}$, and quality factor
$Q=Q_{f,\omega}$ associated with the harmonic external force $f\left(
t\right)  =fe^{-\mathrm{i}\omega t}$ is given by the quantities%
\begin{equation}
U=\frac{1}{2}\left(  \upsilon,\upsilon\right)  ,\text{ \ \ }%
W_{\text{\textrm{dis}}}=\beta\left(  \upsilon,B\upsilon\right)  ,\text{
\ \ }Q=\left\vert \omega\right\vert \frac{U}{W_{\text{\textrm{dis}}}%
}=\left\vert \omega\right\vert \frac{\frac{1}{2}\left(  \upsilon
,\upsilon\right)  }{\beta\left(  \upsilon,B\upsilon\right)  }, \label{palhf2}%
\end{equation}
where $Q$ is said to be finite if $W_{\text{\textrm{dis}}}\not =0$. \ 

For the results in the rest of this section we assume that $f$, $\omega$ are
independent of the loss parameter $\beta$.

The techniques of analysis in the high-loss regime $\beta\gg1$ differ
significantly depending on whether the frequency $\omega$ is an asymptotic
resonance frequency or not.

\begin{definition}
[nonresonance frequency]A real number $\omega$ is called an asymptotic
nonresonance frequency of the system (\ref{mvt7}) provided $\omega\not =%
\rho_{j}$, $N_{B}+1\leq j\leq N$, otherwise it is an asymptotic resonance frequency.
\end{definition}

The usage of this terminology is justified by the following proposition:

\begin{proposition}
\label{Prop6}Let $\omega\in%
%TCIMACRO{\U{211d} }%
%BeginExpansion
\mathbb{R}
%EndExpansion
$. Then the admittance operator $\mathfrak{A}\left(  \omega\right)  $ is
analytic at $\beta=\infty$ if and only if $\omega$ is an asymptotic
nonresonance frequency.
\end{proposition}

In this paper we will only consider the nonresonance frequencies.

\subsection{Nonresonance frequencies\label{SecAsyNonRes}}

In this section we state the results of our analysis of losses for external
harmonic forces with asymptotic nonresonance frequencies in the high-loss regime.

Recall from (\ref{hbcom1}) the Hilbert space $H$ decomposes into the direct
sum of orthogonal subspaces invariant with respect to the operator $B\geq0$,
namely,%
\[
H=H\oplus H_{B}^{\bot}%
\]
where $H_{B}=\operatorname{ran}B$, $H_{B}^{\bot}=\ker B$ are the loss and
no-loss subspaces with orthogonal projections $P_{B}$, $P_{B}^{\bot}$,
respectively. \ It follows from this and the block representation of $\Omega$
and $B$ in (\ref{hbcom2}) that $\xi I-A\left(  \beta\right)  $, with respect
to this decomposition, is the $2\times2$ block operator matrix in
(\ref{ksia1}), namely,
\begin{gather}
\omega I-A\left(  \beta\right)  =\left[
\begin{array}
[c]{cc}%
\Xi_{2}\left(  \omega,\beta\right)  & -\Theta\\
-\Theta^{\ast} & \Xi_{1}\left(  \omega\right)
\end{array}
\right]  ,\label{anrf1}\\
\Xi_{2}\left(  \omega,\beta\right)  :=\omega I_{2}-\left(  \Omega
_{2}-\mathrm{i}\beta B_{2}\right)  ,\quad\Xi_{1}\left(  \omega\right)
:=\omega I_{1}-\Omega_{1},\nonumber
\end{gather}
where $\Omega_{1}$, $\Omega_{2}$, and $B_{2}$ are self-adjoint operators, the
latter of which has an inverse satisfying $B_{2}^{-1}>0$. With respect to this
block representation, the Schur complement of $\Xi_{2}\left(  \omega
,\beta\right)  $ in $\omega I-A\left(  \beta\right)  $ is the operator in
(\ref{ksia2}), namely,
\begin{equation}
S_{2}\left(  \omega,\beta\right)  =\Xi_{1}\left(  \omega\right)  -\Theta
^{\ast}\Xi_{2}\left(  \omega,\beta\right)  ^{-1}\Theta, \label{anrf2}%
\end{equation}
whenever $\Xi_{2}\left(  \omega,\beta\right)  $ is invertible.

To simplify lengthy expressions we will often suppress the symbols $\omega$,
$\beta$ appearing as arguments in the operators $\Xi_{1}\left(  \omega\right)
$, $\Xi_{2}\left(  \omega,\beta\right)  $, $S_{2}\left(  \omega,\beta\right)
$. \ We now give the main results of this section for an asymptotic
nonresonance frequency $\omega$.

\begin{proposition}
[admittance asymptotics]\label{Prop7}For $\beta\gg1$, each of the operators
$\Xi_{1}\left(  \omega\right)  $, $\Xi_{2}\left(  \omega,\beta\right)  $,
$S_{2}\left(  \omega,\beta\right)  $, and $\omega I-A\left(  \beta\right)  $
are invertible and the admittance operator $\mathfrak{A}\left(  \omega\right)
=\mathrm{i}\left(  \omega I-A\left(  \beta\right)  \right)  ^{-1}$ is given by
the formula%
\begin{gather}
\mathfrak{A}\left(  \omega\right)  =\mathrm{i}\left[
\begin{array}
[c]{cc}%
I_{2} & \Xi_{2}^{-1}\Theta\\
0 & I_{1}%
\end{array}
\right]  \left[
\begin{array}
[c]{cc}%
\Xi_{2}^{-1} & 0\\
0 & S_{2}^{-1}%
\end{array}
\right]  \left[
\begin{array}
[c]{cc}%
I_{2} & 0\\
\Theta^{\ast}\Xi_{2}^{-1} & I_{1}%
\end{array}
\right] \label{adop1}\\
=\mathrm{i}\left[
\begin{array}
[c]{cc}%
\Xi_{2}^{-1}+\Xi_{2}^{-1}\Theta S_{2}^{-1}\Theta^{\ast}\Xi_{2}^{-1} & \Xi
_{2}^{-1}\Theta S_{2}^{-1}\\
S_{2}^{-1}\Theta^{\ast}\Xi_{2}^{-1} & S_{2}^{-1}%
\end{array}
\right]  .\nonumber
\end{gather}
Moreover, $\mathfrak{A}\left(  \omega\right)  $ is analytic at $\beta=\infty$
and has the asymptotic expansion%
\begin{equation}
\mathfrak{A}\left(  \omega\right)  =\left[
\begin{array}
[c]{cc}%
0 & 0\\
0 & \mathrm{i}\Xi_{1}^{-1}%
\end{array}
\right]  +W^{\left(  -1\right)  }\beta^{-1}+O\left(  \beta^{-2}\right)  ,
\label{adop2}%
\end{equation}
as $\beta\rightarrow\infty$, where%
\begin{align}
W^{\left(  -1\right)  }  &  =\left[
\begin{array}
[c]{cc}%
B_{2}^{-1} & B_{2}^{-1}\Theta\Xi_{1}^{-1}\\
\left(  \Xi_{1}^{-1}\right)  ^{\ast}\Theta^{\ast}B_{2}^{-1} & \left(  \Xi
_{1}^{-1}\right)  ^{\ast}\Theta^{\ast}B_{2}^{-1}\Theta\Xi_{1}^{-1}%
\end{array}
\right] \label{adop3}\\
&  =\left[
\begin{array}
[c]{cc}%
I_{2} & 0\\
\left(  \Xi_{1}^{-1}\right)  ^{\ast}\Theta^{\ast} & I_{1}%
\end{array}
\right]  \left[
\begin{array}
[c]{cc}%
B_{2}^{-1} & 0\\
0 & 0
\end{array}
\right]  \left[
\begin{array}
[c]{cc}%
I_{2} & \Theta\Xi_{1}^{-1}\\
0 & I_{1}%
\end{array}
\right]  \text{.}\nonumber
\end{align}
In particular, this is a positive semidefinite operator, i.e.,
\[
W^{\left(  -1\right)  }\geq0\text{.}%
\]

\end{proposition}

\begin{corollary}
\label{Cor5}The operators $\mathfrak{A}\left(  \omega\right)  ^{\ast
}\mathfrak{A}\left(  \omega\right)  $, $P_{B}\mathfrak{A}\left(
\omega\right)  ^{\ast}\mathfrak{A}\left(  \omega\right)  P_{B}$, and
$\mathfrak{A}\left(  \omega\right)  ^{\ast}\beta B\mathfrak{A}\left(
\omega\right)  $ are analytic at $\beta=\infty$. \ Moreover, they have the
asymptotic expansions%
\begin{gather}
\mathfrak{A}^{\ast}\left(  \omega\right)  \mathfrak{A}\left(  \omega\right)
=\left[
\begin{array}
[c]{cc}%
0 & 0\\
0 & \left(  \Xi_{1}^{-1}\right)  ^{\ast}\Xi_{1}^{-1}%
\end{array}
\right] \label{adop4}\\
+\left[
\begin{array}
[c]{cc}%
0 & \mathrm{i}B_{2}^{-1}\Theta\left(  \Xi_{1}^{-1}\right)  ^{\ast}\Xi_{1}%
^{-1}\\
-\mathrm{i}\left(  \Xi_{1}^{-1}\right)  ^{\ast}\Xi_{1}^{-1}\Theta^{\ast}%
B_{2}^{-1} & 2\operatorname{Re}\left\{  \mathrm{i}\left(  \Xi_{1}^{-1}\right)
^{\ast}\Theta^{\ast}B_{2}^{-1}\Theta\left(  \Xi_{1}^{-1}\right)  ^{\ast}%
\Xi_{1}^{-1}\right\}
\end{array}
\right]  \beta^{-1}+O\left(  \beta^{-2}\right)  ,\nonumber
\end{gather}
\qquad%
\begin{equation}
P_{B}\mathfrak{A}\left(  \omega\right)  ^{\ast}\mathfrak{A}\left(
\omega\right)  P_{B}=\left[
\begin{array}
[c]{cc}%
B_{2}^{-2}+B_{2}^{-1}\Theta\left(  \Xi_{1}^{-1}\right)  ^{\ast}\Xi_{1}%
^{-1}\Theta^{\ast}B_{2}^{-1} & 0\\
0 & 0
\end{array}
\right]  \beta^{-2}+O\left(  \beta^{-3}\right)  , \label{adop5}%
\end{equation}
and%
\begin{equation}
\mathfrak{A}^{\ast}\left(  \omega\right)  \beta B\mathfrak{A}\left(
\omega\right)  =W^{\left(  -1\right)  }\beta^{-1}+O\left(  \beta^{-2}\right)
\label{adop6}%
\end{equation}
as $\beta\rightarrow\infty$.
\end{corollary}

The following statements give our main results regarding the stored energy $U
$, power of dissipated energy $W_{\text{\textrm{dis}}}$, and quality factor
$Q=Q_{f,\omega}$ associated a harmonic external force $f\left(  t\right)
=fe^{-\mathrm{i}\omega t}$, where $f$, $\omega$ are independent of $\beta$,
$f\not =0$, and $\omega$ an asymptotic nonresonance frequency of the system
(\ref{mvt7}). \ As we shall see the behaviour in the high-loss regime
$\beta\gg1$ of these quantities is drastically different depending on whether
the amplitude $f$ has a component in the no-loss subspace $H_{B}^{\bot}$ or
not, i.e., $P_{B}^{\bot}f\not =0$ or $P_{B}^{\bot}f=0$.

\begin{theorem}
[quality factor]\label{Thm3}If $P_{B}^{\bot}f=0$ then the stored energy $U$,
power of dissipated energy $W_{\text{\textrm{dis}}}$, and quality factor $Q$
are analytic at $\beta=\infty$ and have the asymptotic expansions
\begin{equation}
U=\frac{1}{2}\left(  f,\left[  B_{2}^{-2}+B_{2}^{-1}\Theta\left(  \Xi
_{1}\left(  \omega\right)  ^{-1}\right)  ^{\ast}\Xi_{1}\left(  \omega\right)
^{-1}\Theta^{\ast}B_{2}^{-1}\right]  f\right)  \beta^{-2}+O\left(  \beta
^{-3}\right)  , \label{adop7}%
\end{equation}%
\begin{equation}
W_{\text{\textrm{dis}}}=\left(  f,B_{2}^{-1}f\right)  \beta^{-1}+O\left(
\beta^{-2}\right)  , \label{adop8}%
\end{equation}%
\begin{equation}
Q=\left\vert \omega\right\vert \frac{\frac{1}{2}\left(  f,\left[  B_{2}%
^{-2}+B_{2}^{-1}\Theta\left(  \Xi_{1}\left(  \omega\right)  ^{-1}\right)
^{\ast}\Xi_{1}\left(  \omega\right)  ^{-1}\Theta^{\ast}B_{2}^{-1}\right]
f\right)  }{\left(  f,B_{2}^{-1}f\right)  }\beta^{-1}+O\left(  \beta
^{-2}\right)  \label{adop9}%
\end{equation}
as $\beta\rightarrow\infty$.$\quad$Moreover, the leading order terms of $U$
and $W_{\text{\textrm{dis}}}$ in these expansions are positive numbers,
similarly for $Q$ provided $\omega\not =0$, and satisfy the inequalities%
\begin{gather*}
\frac{1}{2}\left(  f,\left[  B_{2}^{-2}+B_{2}^{-1}\Theta\left(  \Xi_{1}\left(
\omega\right)  ^{-1}\right)  ^{\ast}\Xi_{1}\left(  \omega\right)  ^{-1}%
\Theta^{\ast}B_{2}^{-1}\right]  f\right) \\
\geq\frac{1}{2}\left(  \sup_{1\leq j\leq N_{B}}\left\{  \mathring{\zeta}%
_{j}\right\}  \right)  ^{-1}\left(  f,B_{2}^{-1}f\right)  \geq\frac{1}%
{2}\left(  \sup_{1\leq j\leq N_{B}}\left\{  \mathring{\zeta}_{j}\right\}
\right)  ^{-2}\left(  f,f\right)  >0\text{.}%
\end{gather*}

\end{theorem}

\begin{theorem}
[quality factor]\label{Thm4}If $P_{B}^{\bot}f\not =0$ then the stored energy
$U$ and power of dissipated energy $W_{\text{\textrm{dis}}}$ are analytic at
$\beta=\infty$ and have the asymptotic expansions%
\begin{equation}
U=\frac{1}{2}\left(  \Xi_{1}\left(  \omega\right)  ^{-1}P_{B}^{\bot}f,\Xi
_{1}\left(  \omega\right)  ^{-1}P_{B}^{\bot}f\right)  +O\left(  \beta
^{-1}\right)  \label{adop10}%
\end{equation}%
\begin{equation}
W_{\text{\textrm{dis}}}=\left(  f,W^{\left(  -1\right)  }f\right)  \beta
^{-1}+O\left(  \beta^{-2}\right)  \label{adop11}%
\end{equation}
as $\beta\rightarrow\infty$. \ In particular, the leading order terms in the
expansions of $U$ and $W_{\text{\textrm{dis}}}$ are positive and nonnegative
numbers, respectively. \ 

Moreover, if $W_{\text{\textrm{dis}}}\not \equiv 0$ for $\beta\gg1$ then the
quality factor $Q$ has a pole at $\beta=\infty$ provided $\omega\not =0$. \ In
particular, if $f\not \in \ker W^{\left(  -1\right)  }$ then it has the
asymptotic expansion
\begin{equation}
Q=\left\vert \omega\right\vert \frac{\frac{1}{2}\left(  \Xi_{1}\left(
\omega\right)  ^{-1}P_{B}^{\bot}f,\Xi_{1}\left(  \omega\right)  ^{-1}%
P_{B}^{\bot}f\right)  }{\left(  f,W^{\left(  -1\right)  }f\right)  }%
\beta+O\left(  1\right)  \label{adop12}%
\end{equation}
as $\beta\rightarrow\infty$, whose leading order term is a positive number
provided $\omega\not =0$.
\end{theorem}

\begin{corollary}
\label{Cor6}The operator $W^{\left(  -1\right)  }$ defined in (\ref{adop3})
has the $N-N_{B}$ dimensional kernel%
\begin{align*}
\ker W^{\left(  -1\right)  }  &  =\ker\left(  P_{B}+\Theta\Xi_{1}\left(
\omega\right)  ^{-1}P_{B}^{\bot}\right) \\
&  =\left\{  f_{1}+f_{2}\in H:f_{1}\in H_{B}^{\bot}\text{ and }f_{2}%
=-\Theta\Xi_{1}\left(  \omega\right)  ^{-1}f_{1}\right\}  \text{.}%
\end{align*}

\end{corollary}

\begin{corollary}
[quality factor]\label{Cor7}The stored energy $U$, power of dissipated energy
$W_{\text{\textrm{dis}}}$, and quality factor $Q$ have the following limits as
$\beta\rightarrow\infty$:%
\begin{equation}
\lim_{\beta\rightarrow\infty}U=\left\{
\begin{array}
[c]{cc}%
\frac{1}{2}\left(  \Xi_{1}\left(  \omega\right)  ^{-1}P_{B}^{\bot}f,\Xi
_{1}\left(  \omega\right)  ^{-1}P_{B}^{\bot}f\right)  >0 & \text{if }%
P_{B}^{\bot}f\not =0,\\
0 & \text{if }P_{B}^{\bot}f=0,
\end{array}
\right.  \label{adop13}%
\end{equation}%
\begin{equation}
\lim_{\beta\rightarrow\infty}W_{\text{\textrm{dis}}}=0, \label{adop14}%
\end{equation}%
\begin{equation}
\lim_{\beta\rightarrow\infty}Q=\left\{
\begin{array}
[c]{ll}%
\infty & \text{if }P_{B}^{\bot}f\not =0\text{ and }\omega\not =0,\\
0 & \text{if }P_{B}^{\bot}f=0\text{ or }\omega=0,
\end{array}
\right.  \label{adop15}%
\end{equation}
where we assume in the statement regarding quality factor for the case
$P_{B}^{\bot}f\not =0$ that $W_{\text{\textrm{dis}}}\not \equiv 0$ for
$\beta\gg1$. \ In particular, if $P_{B}^{\bot}f=0$ or $f\not \in \ker
W^{\left(  -1\right)  }$ then $Q$ is finite for $\beta\gg1$ and we have the
above limits for $U$, $W_{\text{\textrm{dis}}}$, and $Q$ as $\beta
\rightarrow\infty$.
\end{corollary}

\section{Proof of results\label{sprfr}}

This section contains the proofs of the results of this paper. \ We split
these proofs into two subsections. \ In the first subsection we prove the
statements given in Section \ref{spsyo} on the spectral perturbation analysis
in the high-loss and low-loss regime for the system operator. \ In the second
subsection we prove the statements in Section \ref{splhf} on the perturbation
analysis of losses for external harmonic forces. \ All assumptions, notation,
and convention used here will adhere to that previously introduced in those
two sections.

\subsection{Perturbation analysis of the system operator}

This purpose of this section is to prove the results given in Section
\ref{spsyo}. \ We do so by considering in separate subsections the high-loss
regime $\beta\gg1$ and the low-loss regime $0\leq\beta\ll1$.

\subsubsection{The high-loss regime}

\begin{proof}
[Proof of Theorem \ref{Thm1}]Let $\mathring{\zeta}_{j}$, $1\leq j\leq N_{B}$
be an indexing of all the nonzero eigenvalues of $B$ (counting multiplicities)
where $N_{B}=\operatorname{rank}B$. \ These eigenvalues are all positive real
numbers since by assumption $B\geq0$.

We begin by extending the domain of the system operator by $A\left(
\beta\right)  =\Omega-\mathrm{i}\beta B$, $\beta\in%
%TCIMACRO{\U{2102} }%
%BeginExpansion
\mathbb{C}
%EndExpansion
$. \ Recall by our assumption $\Omega$ is self-adjoint. \ Hence the operator
$\left(  -\mathrm{i}\beta\right)  ^{-1}A\left(  \beta\right)  =B+\left(
-\mathrm{i}\beta\right)  ^{-1}\Omega$ is analytic in $\beta^{-1}$ in a complex
neighborhood of $\beta=\infty$ and\ if we use the substitution $\varepsilon=$
$\left(  -\mathrm{i}\beta\right)  ^{-1}$ then the operator $\varepsilon
A\left(  \mathrm{i}\varepsilon^{-1}\right)  =B+\varepsilon\Omega$ is an
analytic operator which is self-adjoint for real $\varepsilon$. \ Thus by a
theorem of Rellich \cite[p. 21, Theorem 1]{Bau85} we know for $\varepsilon\in%
%TCIMACRO{\U{2102} }%
%BeginExpansion
\mathbb{C}
%EndExpansion
$ with $\left\vert \varepsilon\right\vert \ll1$, the operator $B+\varepsilon
\Omega$ is diagonalizable and there exists a complete set of analytic
eigenvalues $\lambda_{j}\left(  \varepsilon\right)  $ and eigenvectors
$x_{j}\left(  \varepsilon\right)  $ satisfying%
\begin{gather}
\left(  B+\varepsilon\Omega\right)  x_{j}\left(  \varepsilon\right)
=\lambda_{j}\left(  \varepsilon\right)  x_{j}\left(  \varepsilon\right)
,\text{ \ \ }\overline{\lambda_{j}\left(  \overline{\varepsilon}\right)
}=\lambda_{j}\left(  \varepsilon\right)  ,\text{ \ \ }1\leq j\leq
N;\label{pfhlr1}\\
\left(  x_{j}\left(  \varepsilon\right)  ,x_{k}\left(  \varepsilon\right)
\right)  =\delta_{jk},\text{ \ \ for }\varepsilon\in%
%TCIMACRO{\U{211d} }%
%BeginExpansion
\mathbb{R}
%EndExpansion
,\text{ \ \ }1\leq j,k\leq N,\nonumber
\end{gather}
where $\delta_{jk}$ denotes the Kronecker delta symbol. \ In particular, the
vectors $\mathring{w}_{j}:=x_{j}\left(  0\right)  $, $j=1,\ldots,N$ form an
orthonormal basis of eigenvectors for the operator $B$. \ Thus, after a
possible reindexing of these analytic eigenpairs, we may assume without loss
of generality that%
\begin{gather}
B\mathring{w}_{j}=\mathring{\zeta}_{j}\mathring{w}_{j},\text{ \ \ }\lambda
_{j}\left(  0\right)  =\mathring{\zeta}_{j}>0,\text{ \ \ }1\leq j\leq
N_{B};\label{pfhlr2}\\
B\mathring{w}_{j}=0,\text{ \ \ }\lambda_{j}\left(  0\right)  =0,\text{
\ \ }N_{B}+1\leq j\leq N.\nonumber
\end{gather}
Denote the derivatives of these eigenvalues at $\varepsilon=0$ by $\rho
_{j}:=\lambda_{j}^{\prime}\left(  0\right)  $, $1\leq j\leq N$. \ Then it
follows that they satisfy
\[
\rho_{j}=\left(  \mathring{w}_{j},\Omega\mathring{w}_{j}\right)  \in%
%TCIMACRO{\U{211d} }%
%BeginExpansion
\mathbb{R}
%EndExpansion
,\ \ 1\leq j\leq N.
\]
Indeed, for $\varepsilon$ real with $\left\vert \varepsilon\right\vert \ll1$
by (\ref{pfhlr1}) we have%
\[
0=\left(  x_{j}\left(  \varepsilon\right)  ,\left(  B+\varepsilon
\Omega-\lambda_{j}\left(  \varepsilon\right)  I\right)  x_{j}\left(
\varepsilon\right)  \right)  ,\text{ \ \ }1\leq j\leq N
\]
and the result follows immediately by taking the derivative on both sides and
evaluating at $\varepsilon=0$.

Now from these facts and recalling the substitution $\varepsilon=$ $\left(
-\mathrm{i}\beta\right)  ^{-1}$ that was made, we conclude that if $\beta\in%
%TCIMACRO{\U{2102} }%
%BeginExpansion
\mathbb{C}
%EndExpansion
$ with $\left\vert \beta\right\vert \gg1$ then the operator $A\left(
\beta\right)  =\left(  -\mathrm{i}\beta\right)  \left(  B+\left(
-\mathrm{i}\beta\right)  ^{-1}\Omega\right)  $ is diagonalizable with a
complete set of eigenvalues $\zeta_{j}\left(  \beta\right)  $ and eigenvectors
$w_{j}\left(  \beta\right)  $ satisfying%
\begin{gather}
\zeta_{j}\left(  \beta\right)  =\left(  -\mathrm{i}\beta\right)  \lambda
_{j}\left(  \left(  -\mathrm{i}\beta\right)  ^{-1}\right)  ,\text{ \ \ }%
w_{j}\left(  \beta\right)  =x_{j}\left(  \left(  -\mathrm{i}\beta\right)
^{-1}\right)  ,\label{pfhlr4}\\
A\left(  \beta\right)  w_{j}\left(  \beta\right)  =\zeta_{j}\left(
\beta\right)  w_{j}\left(  \beta\right)  ,\text{ \ \ }1\leq j\leq
N\text{.}\nonumber
\end{gather}
We will now show that the eigenpairs%
\begin{gather*}
\text{high-loss: \ \ }\zeta_{j}\left(  \beta\right)  ,w_{j}\left(
\beta\right)  ,\text{ \ \ }1\leq j\leq N_{B};\\
\text{low-loss: \ \ }\zeta_{j}\left(  \beta\right)  ,w_{j}\left(
\beta\right)  ,\text{ \ \ }N_{B}+1\leq j\leq N,
\end{gather*}
have the properties described in Theorem \ref{Thm1}.

We start with the high-loss eigenpairs. \ First, it follows from our results
above that the eigenpairs $\lambda_{j}\left(  \varepsilon\right)  $,
$x_{j}\left(  \varepsilon\right)  $, $1\leq j\leq N_{B}$ are analytic
for\ $\varepsilon\in%
%TCIMACRO{\U{2102} }%
%BeginExpansion
\mathbb{C}
%EndExpansion
$ with $\left\vert \varepsilon\right\vert \ll1$ and have the asymptotic
expansions%
\begin{gather*}
\lambda_{j}\left(  \varepsilon\right)  =\mathring{\zeta}_{j}+\rho
_{j}\varepsilon+O\left(  \varepsilon^{2}\right)  ,\text{ \ \ }\mathring{\zeta
}_{j}>0,\text{ \ \ }\rho_{j}\in%
%TCIMACRO{\U{211d} }%
%BeginExpansion
\mathbb{R}
%EndExpansion
,\\
x_{j}\left(  \varepsilon\right)  =\mathring{w}_{j}+O\left(  \varepsilon
\right)  ,\text{ \ \ }1\leq j\leq N_{B}%
\end{gather*}
as $\varepsilon\rightarrow0$. \ This implies by (\ref{pfhlr4}) the high-loss
eigenvalues have poles at $\beta=\infty$ whereas their eigenvectors are
analytic at $\beta=\infty$ and they have the asymptotic expansions
\begin{gather*}
\zeta_{j}\left(  \beta\right)  =\left(  -\mathrm{i}\beta\right)  \lambda
_{j}\left(  \left(  -\mathrm{i}\beta\right)  ^{-1}\right)  =-\mathrm{i}%
\beta\mathring{\zeta}_{j}+\rho_{j}+O\left(  \beta^{-1}\right)  ,\text{
\ \ }\mathring{\zeta}_{j}>0,\text{ \ \ }\rho_{j}\in%
%TCIMACRO{\U{211d} }%
%BeginExpansion
\mathbb{R}
%EndExpansion
,\\
w_{j}\left(  \beta\right)  =x_{j}\left(  \left(  -\mathrm{i}\beta\right)
^{-1}\right)  =\mathring{w}_{j}+O\left(  \beta^{-1}\right)  ,\text{ \ \ }1\leq
j\leq N_{B}%
\end{gather*}
as $\beta\rightarrow\infty$, where the vectors $\mathring{w}_{j}$, $1\leq
j\leq N_{B}$ form an orthonormal basis of the loss subspace $H_{B}%
=\operatorname{ran}B$ and
\[
B\mathring{w}_{j}=\mathring{\zeta}_{j}\mathring{w}_{j},\text{ \ \ }\rho
_{j}=\left(  \mathring{w}_{j},\Omega\mathring{w}_{j}\right)  ,\text{
\ \ }1\leq j\leq N_{B}.
\]
This proves statement (i) of this theorem.

We now consider the low-loss eigenpairs. \ By the results above we know the
eigenpairs $\lambda_{j}\left(  \varepsilon\right)  $, $x_{j}\left(
\varepsilon\right)  $, $N_{B}+1\leq j\leq N$ are analytic for\ $\varepsilon\in%
%TCIMACRO{\U{2102} }%
%BeginExpansion
\mathbb{C}
%EndExpansion
$ with $\left\vert \varepsilon\right\vert \ll1$ and have the asymptotic
expansions%
\begin{gather*}
\lambda_{j}\left(  \varepsilon\right)  =\rho_{j}\varepsilon+\frac{1}{2}%
\lambda_{j}^{\prime\prime}(0)\varepsilon^{2}+O\left(  \varepsilon^{3}\right)
,\\
x_{j}\left(  \varepsilon\right)  =\mathring{w}_{j}+x_{j}^{\prime
}(0)\varepsilon+O\left(  \varepsilon\right)  ,\text{ \ \ }N_{B}+1\leq j\leq N
\end{gather*}
as $\varepsilon\rightarrow0$. \ Moreover, there is an explicit formula for the
second derivative of these eigenvalues%
\[
\lambda_{j}^{\prime\prime}(0)=2\left(  x_{j}^{\prime}(0),Bx_{j}^{\prime
}(0)\right)  \geq0,\ \ N_{B}+1\leq j\leq N.
\]
Indeed, this follows from (\ref{pfhlr1}) and (\ref{pfhlr2}) since%
\begin{gather*}
\rho_{j}\varepsilon+\frac{1}{2}\lambda_{j}^{\prime\prime}(0)\varepsilon
^{2}+O\left(  \varepsilon^{3}\right)  =\lambda_{j}\left(  \varepsilon\right)
=\left(  x_{j}\left(  \varepsilon\right)  ,\Omega x_{j}\left(  \varepsilon
\right)  \right)  \varepsilon+\left(  x_{j}\left(  \varepsilon\right)
,Bx_{j}\left(  \varepsilon\right)  \right) \\
=\left(  \mathring{w}_{j},\Omega\mathring{w}_{j}\right)  \varepsilon+\left(
x_{j}^{\prime}(0),Bx_{j}^{\prime}(0)\right)  \varepsilon^{2}+O\left(
\varepsilon^{-3}\right)
\end{gather*}
for real $\varepsilon$ as $\varepsilon\rightarrow0$. \ Thus we can conclude
from this and (\ref{pfhlr4})\ that the low-loss eigenpairs are analytic at
$\beta=\infty$ and they have the asymptotic expansions
\begin{gather*}
\zeta_{j}\left(  \beta\right)  =\left(  -\mathrm{i}\beta\right)  \lambda
_{j}\left(  \left(  -i\beta\right)  ^{-1}\right)  =\rho_{j}-\mathrm{i}%
d_{j}\beta^{-1}+O\left(  \beta^{-2}\right)  ,\text{ \ \ }\rho_{j}\in%
%TCIMACRO{\U{211d} }%
%BeginExpansion
\mathbb{R}
%EndExpansion
,\text{ \ \ }d_{j}\geq0\\
w_{j}\left(  \beta\right)  =x_{j}\left(  \left(  -\mathrm{i}\beta\right)
^{-1}\right)  =\mathring{w}_{j}+w_{j}^{(-1)}\beta^{-1}+O\left(  \beta
^{-2}\right)  ,\text{ \ \ }N_{B}+1\leq j\leq N
\end{gather*}
as $\beta\rightarrow\infty$, where the vectors $\mathring{w}_{j}$,
$N_{B}+1\leq j\leq N$ form an orthonormal basis of the no-loss subspace
$H_{B}^{\bot}=\ker B$ and
\begin{gather*}
B\mathring{w}_{j}=0,\text{ \ \ }\rho_{j}=\left(  \mathring{w}_{j}%
,\Omega\mathring{w}_{j}\right)  ,\text{ \ \ }w_{j}^{(-1)}=\mathrm{i}%
x_{j}^{\prime}(0),\\
\text{\ }-\mathrm{i}d_{j}=-\mathrm{i}\frac{1}{2}\lambda_{j}^{\prime\prime
}(0)=-\mathrm{i}\left(  x_{j}^{\prime}(0),Bx_{j}^{\prime}(0)\right)
=-\mathrm{i}\left(  w_{j}^{(-1)},Bw_{j}^{(-1)}\right)
\end{gather*}
for $N_{B}+1\leq j\leq N$. \ This completes the proof.
\end{proof}

\begin{proof}
[Proof of Corollary \ref{Cor1}]Let $w_{j}\left(  \beta\right)  $, $1\leq j\leq
N_{B}$ and $w_{j}\left(  \beta\right)  $, $N_{B}+1\leq j\leq N$ denote the
high-loss and low-loss eigenvectors, respectively, given in the previous
theorem. \ Then by our results we know that the zeroth order terms in their
asymptotic expansions must satisfy $P_{B}\mathring{w}_{j}=\mathring{w}_{j}$,
$P_{B}^{\bot}\mathring{w}_{j}=0$ for $1\leq j\leq N_{B}$ and $P_{B}%
\mathring{w}_{j}=0$, $P_{B}^{\bot}\mathring{w}_{j}=\mathring{w}_{j}$ for
$N_{B}+1\leq j\leq N$ since $H=H_{B}\oplus H_{B}^{\bot}$ and $P_{B}$,
$P_{B}^{\bot}$ are the orthogonal projections onto $H_{B}$, $H_{B}^{\bot}$,
respectively. \ The proof of this corollary now follows.
\end{proof}

\begin{proof}
[Proof of Proposition \ref{Prop1}]Let $\zeta_{j}\left(  \beta\right)  $,
$1\leq j\leq N$ be the high-loss and low-loss eigenvalues given in Theorem
\ref{Thm1}. \ Then as described in its proof we can extend the domain of the
system operator $A\left(  \beta\right)  =\Omega-\mathrm{i}\beta B$, $\beta\in%
%TCIMACRO{\U{2102} }%
%BeginExpansion
\mathbb{C}
%EndExpansion
$ and these eigenvalues can be extended uniquely to meromorphic functions in a
neighborhood of $\beta=\infty$ whose values are eigenvalues of $A\left(
\beta\right)  $. \ Moreover, using the same notation to denote their
extensions, it follows from (\ref{pfhlr1}) and (\ref{pfhlr4}) that these
functions satisfy
\[
\overline{\zeta_{j}\left(  \overline{\beta}\right)  }=\zeta_{j}\left(
-\beta\right)  ,\text{ \ \ }1\leq j\leq N
\]
for all $\beta\in%
%TCIMACRO{\U{2102} }%
%BeginExpansion
\mathbb{C}
%EndExpansion
$ with $\left\vert \beta\right\vert \gg1$. \ This implies $\frac{1}{2}\left(
\zeta_{j}\left(  \beta\right)  +\zeta_{j}\left(  -\beta\right)  \right)  $ and
$\frac{1}{2i}\left(  \zeta_{j}\left(  \beta\right)  -\zeta_{j}\left(
-\beta\right)  \right)  $ are even and odd functions, respectively, and
meromorphic at $\beta=\infty$ and equal $\operatorname{Re}\zeta_{j}\left(
\beta\right)  $ and $\operatorname{Im}\zeta_{j}\left(  \beta\right)  $,
respectively, for real $\beta$. \ In particular, this implies their Laurent
series in $\beta^{-1}$ have only even and odd powers, respectively. \ The rest
of the proof of this proposition now follows immediately by considering the
real and imaginary part of the asymptotic expansions of the high-loss and
low-loss eigenvalues given in Theorem \ref{Thm1}.
\end{proof}

\begin{proof}
[Proof of Proposition \ref{Prop2}: 1. \& 2]Here we will prove just the first
two statements of the proposition. \ We come back to the proof of the third
and fourth statements after we have proved Proposition \ref{Prop3}.

Let $j\in\left\{  1,\ldots,N\right\}  $ and $\beta\gg1$. \ Well, since
$\zeta_{j}\left(  \beta\right)  ,w_{j}\left(  \beta\right)  $ is an eigenpair
of $A\left(  \beta\right)  $ then it follows from (\ref{syop6}) and the fact
$B\geq0$ that $\operatorname{Im}\zeta_{j}\left(  \beta\right)  \leq0$. \ From
this and Proposition \ref{Prop1}, since $\operatorname{Im}\zeta_{j}\left(
\beta\right)  $ either as a pole or is analytic at $\beta=\infty$, then either
$\operatorname{Im}\zeta_{j}\left(  \beta\right)  \equiv0$ or
$\operatorname{Im}\zeta_{j}\left(  \beta\right)  \leq0$. \ If
$\operatorname{Im}\zeta_{j}\left(  \beta\right)  \equiv0$ then it follows by
(\ref{syop6}) that $Bw_{j}\left(  \beta\right)  \equiv0$ and hence $\zeta
_{j}\left(  \beta\right)  w_{j}\left(  \beta\right)  \equiv A\left(
\beta\right)  w_{j}\left(  \beta\right)  \equiv\Omega w_{j}\left(
\beta\right)  $ which implies that $\zeta_{j}\left(  \beta\right)  \equiv
\rho_{j}$. \ Furthermore, if $1\leq j\leq N_{B}$ then we know that
$\operatorname{Im}\zeta_{j}\left(  \beta\right)  =-\mathring{\zeta}_{j}%
\beta+O\left(  \beta^{-1}\right)  $ as $\beta\rightarrow\infty$ with
$\mathring{\zeta}_{j}>0$ and so $\operatorname{Im}\zeta_{j}\left(
\beta\right)  <0$ for $\beta\gg1$. \ This completes the proof of the first two statements.
\end{proof}

\begin{proof}
[Proof of Proposition \ref{Prop3}]We begin by proving the first statement of
this proposition. \ Recall that $B_{2}=\left.  P_{B}BP_{B}\right\vert _{H_{B}%
}:H_{B}\rightarrow H_{B}$ where $P_{B}$ is the orthogonal projection onto
$H_{B}$. \ By Theorem \ref{Thm1} it follows that $\mathring{\zeta}_{j}$,
$\mathring{w}_{j}$, $1\leq j\leq N_{B}$ are a complete set of eigenvalues and
eigenvectors for $B_{2}$ with these eigenvectors forming an orthonormal basis
for $H_{B}$. \ As these eigenvalues are all positive this implies $B_{2}>0$
and, in particular, it is invertible and its inverse satisfies $B_{2}^{-1}>0$.
\ This completes the proof of the first statement.

Next, we prove the second statement of this proposition. \ Well, by Theorem
\ref{Thm1} we know that the vectors $\mathring{w}_{j}$, $N_{B}+1\leq j\leq N$
form an orthonormal basis for the no-loss subspace $H_{B}^{\bot}$. \ By
definition $P_{B}^{\bot}$ is the orthogonal projection onto $H_{B}^{\bot}$
implying $P_{B}^{\bot}\mathring{w}_{j}=\mathring{w}_{j}$, for $j=N_{B}%
+1,\ldots,N$. \ And hence, since $\Omega_{1}=\left.  P_{B}^{\bot}\Omega
P_{B}^{\bot}\right\vert _{H_{B}^{\bot}}:H_{B}^{\bot}\rightarrow H_{B}^{\bot}$
which is obviously a self-adjoint operator because $P_{B}^{\bot}$ and $\Omega$
are, we have $\Omega_{1}\mathring{w}_{j}=P_{B}^{\bot}\Omega\mathring{w}_{j}$,
$j=N_{B}+1,\ldots,N$. \ Therefore this, the fact $P_{B}^{\bot}B=0$, and
Theorem \ref{Thm1} imply
\[
\Omega_{1}\mathring{w}_{j}=P_{B}^{\bot}\Omega\mathring{w}_{j}=\lim
_{\beta\rightarrow\infty}P_{B}^{\bot}\left(  \Omega-\mathrm{i}\beta B\right)
w_{j}\left(  \beta\right)  =\lim_{\beta\rightarrow\infty}P_{B}^{\bot}\zeta
_{j}\left(  \beta\right)  w_{j}\left(  \beta\right)  =\rho_{j}\mathring{w}_{j}%
\]
for $j=N_{B}+1,\ldots,N$. \ This completes the proof of the second statement.

Finally, we prove the third and final statement of this proposition. \ Recall
the operator $\Theta=\left.  P_{B}\Omega P_{B}^{\bot}\right\vert _{H_{B}%
^{\bot}}:H_{B}^{\bot}\rightarrow H_{B}$. \ We now show that%
\begin{equation}
\Theta\mathring{w}_{j}=\mathrm{i}B_{2}P_{B}w_{j}^{\left(  -1\right)  },\text{
\ \ }N_{B}+1\leq j\leq N. \label{pfhlr5}%
\end{equation}
Well, by (\ref{hlre1}), (\ref{hlre5}), and (\ref{hlre6}) of Theorem \ref{Thm1}
it follows that if $N_{B}+1\leq j\leq N$ then
\begin{gather*}
\rho_{j}\mathring{w}_{j}+O\left(  \beta^{-1}\right)  =\zeta_{j}\left(
\beta\right)  w_{j}\left(  \beta\right)  =A\left(  \beta\right)  w_{j}\left(
\beta\right) \\
=\left(  \Omega-\mathrm{i}\beta B\right)  \mathring{w}_{j}+\beta^{-1}\left(
\Omega-\mathrm{i}\beta B\right)  w_{j}^{\left(  -1\right)  }+O\left(
\beta^{-1}\right)  =\Omega\mathring{w}_{j}-\mathrm{i}Bw_{j}^{\left(
-1\right)  }+O\left(  \beta^{-1}\right)
\end{gather*}
as $\beta\rightarrow\infty$. \ Equating the zeroth order terms we conclude
that%
\[
\rho_{j}\mathring{w}_{j}=\Omega\mathring{w}_{j}-\mathrm{i}Bw_{j}^{\left(
-1\right)  }\text{.}%
\]
Applying $P_{B}$ to both sides of this equation we find that%
\[
0=P_{B}\Omega\mathring{w}_{j}-\mathrm{i}P_{B}Bw_{j}^{\left(  -1\right)
}=\left(  P_{B}\Omega P_{B}^{\bot}\right)  \mathring{w}_{j}-\mathrm{i}\left(
P_{B}BP_{B}\right)  P_{B}w_{j}^{\left(  -1\right)  }=\Theta\mathring{w}%
_{j}-\mathrm{i}B_{2}P_{B}w_{j}^{\left(  -1\right)  }%
\]
which proves the identity (\ref{pfhlr5}).

Therefore by (\ref{hlre6}) of Theorem \ref{Thm1}, (\ref{pfhlr5}), the facts
$P_{B}$ is the orthogonal projection onto $H_{B}$, $P_{B}BP_{B}=B$,
$B_{2}=\left.  P_{B}BP_{B}\right\vert _{H_{B}}$, and since $B_{2}^{-1}$ is
self-adjoint we conclude that%
\begin{gather*}
\text{\ }-\mathrm{i}d_{j}=-\mathrm{i}\left(  w_{j}^{(-1)},Bw_{j}%
^{(-1)}\right)  =-\left(  P_{B}w_{j}^{(-1)},\mathrm{i}B_{2}P_{B}w_{j}%
^{(-1)}\right)  =-\left(  P_{B}w_{j}^{(-1)},\Theta\mathring{w}_{j}\right) \\
=-\left(  B_{2}^{-1}B_{2}P_{B}w_{j}^{(-1)},\Theta\mathring{w}_{j}\right)
=-\mathrm{i}\left(  \mathrm{i}B_{2}P_{B}w_{j}^{(-1)},B_{2}^{-1}\Theta
\mathring{w}_{j}\right)  =-\mathrm{i}\left(  \Theta\mathring{w}_{j},B_{2}%
^{-1}\Theta\mathring{w}_{j}\right) \\
=-\mathrm{i}\left(  \mathring{w}_{j},\Theta^{\ast}B_{2}^{-1}\Theta\mathring
{w}_{j}\right)  .
\end{gather*}
This proves the final statement and hence the proof of the proposition is complete.
\end{proof}

\begin{proof}
[Proof of Proposition \ref{Prop2}: 3. \& 4]We now complete the proof of
Proposition \ref{Prop2} by proving the last two statements. \ Let
$j\in\left\{  1,\ldots,N\right\}  $. \ We begin by proving the third
statement. \ Suppose that $d_{j}=0$. \ Then by Proposition \ref{Prop3} we know
that $0=-\mathrm{i}\left(  \mathring{w}_{j},\Theta^{\ast}B_{2}^{-1}%
\Theta\mathring{w}_{j}\right)  $ and $B_{2}^{-1}>0$. \ This implies
$\Theta\mathring{w}_{j}=0$. \ By Theorem \ref{Thm1} we have $\mathring{w}%
_{j}\in H_{B}^{\bot}$ so that $P_{B}^{\bot}\mathring{w}_{j}=\mathring{w}_{j}$
and hence%
\[
0=\Theta\mathring{w}_{j}=P_{B}\Omega P_{B}^{\bot}\mathring{w}_{j}=P_{B}%
\Omega\mathring{w}_{j}\text{.}%
\]
From this, the fact $P_{B}+P_{B}^{\bot}=I$, and by the first statement in
Proposition \ref{Prop3} it follows that
\[
\Omega\mathring{w}_{j}=P_{B}^{\bot}\Omega\mathring{w}_{j}=P_{B}^{\bot}\Omega
P_{B}^{\bot}\mathring{w}_{j}=\Omega_{1}\mathring{w}_{j}=\rho_{j}\mathring
{w}_{j}.
\]
Thus we have shown if $d_{j}=0$ then $\mathring{w}_{j}\in\ker\left(  \rho
_{j}I-\Omega\right)  $. \ We now prove the converse. \ Suppose $\mathring
{w}_{j}\in\ker\left(  \rho_{j}I-\Omega\right)  $. \ This hypothesis and the
fact $\mathring{w}_{j}\in H_{B}^{\bot}$ imply%
\[
\Theta\mathring{w}_{j}=P_{B}\Omega P_{B}^{\bot}\mathring{w}_{j}=P_{B}%
\Omega\mathring{w}_{j}=\rho_{j}P_{B}\mathring{w}_{j}=0\text{.}%
\]
Therefore\ by Proposition \ref{Prop3} we conclude $d_{j}=\left(  \mathring
{w}_{j},\Theta^{\ast}B_{2}^{-1}\Theta\mathring{w}_{j}\right)  =0$. \ This
proves the third statement of Proposition \ref{Prop2}.

Finally, we will complete the proof of Proposition \ref{Prop2} by proving the
fourth statement. \ Suppose $\mathring{w}_{j}\not \in \ker\left(
\Omega\right)  $ but $\rho_{j}=0$ and $d_{j}=0$. \ Then it would follow from
the third statement of this proposition that $\mathring{w}_{j}\in\ker\left(
\Omega\right)  $, a contradiction. \ Thus we have shown if $\mathring{w}%
_{j}\not \in \ker\left(  \Omega\right)  $ then $\rho_{j}\not =0$ or
$d_{j}\not =0$. \ We now prove the converse. \ Suppose $\rho_{j}\not =0$ or
$d_{j}\not =0$. \ If $\mathring{w}_{j}\in\ker\left(  \Omega\right)  $ then
\[
\rho_{j}\mathring{w}_{j}=\Omega_{1}\mathring{w}_{j}=P_{B}^{\bot}%
\Omega\mathring{w}_{j}=0.
\]
This implies $\rho_{j}=0$ and hence $\mathring{w}_{j}\in\ker\left(  \rho
_{j}I-\Omega\right)  $. \ By the previous statement in this proposition this
implies $d_{j}=0$. \ This yields a contradiction of our hypothesis.
\ Therefore we have shown if $\rho_{j}\not =0$ or $d_{j}\not =0$ then
$\mathring{w}_{j}\not \in \ker\left(  \Omega\right)  $. \ This proves the
fourth statement and hence completes the proof of this proposition.
\end{proof}

\begin{proof}
[Proof of Corollary \ref{Cor2}]The proof of this corollary is straightforward
and follows from Theorem \ref{Thm1} and Proposition \ref{Prop2}.
\end{proof}

\begin{proof}
[Proof of Proposition \ref{Prop4}]For high-loss and low-loss eigenvectors
$w_{j}\left(  \beta\right)  $, $1\leq j\leq N$ we have $\left(  w_{j}\left(
\beta\right)  ,w_{j}\left(  \beta\right)  \right)  =1$ $+O\left(  \beta
^{-1}\right)  $ as $\beta\rightarrow\infty$. From this it follows that the
energy has the asymptotic expansions
\[
U\left[  w_{j}\left(  \beta\right)  \right]  =\frac{1}{2}\left(  w_{j}\left(
\beta\right)  ,w_{j}\left(  \beta\right)  \right)  =\frac{1}{2}+O\left(
\beta^{-1}\right)  ,\text{ \ \ }1\leq j\leq N
\]
as $\beta\rightarrow\infty$. \ Now by Proposition \ref{Prop1} we know the
imaginary parts of high-loss and low-loss eigenvalues $\zeta_{j}\left(
\beta\right)  $, $1\leq j\leq N$ have the asymptotic expansions%
\begin{gather*}
\operatorname{Im}\zeta_{j}\left(  \beta\right)  =-\mathring{\zeta}_{j}%
\beta+O\left(  \beta^{-1}\right)  ,\text{ \ \ }\mathring{\zeta}_{j}>0,\text{
\ \ }1\leq j\leq N_{B};\\
\operatorname{Im}\zeta_{j}\left(  \beta\right)  =-d_{j}\beta^{-1}+O\left(
\beta^{-3}\right)  ,\text{ \ \ }d_{j}\geq0,\text{ \ \ }N_{B}+1\leq j\leq N
\end{gather*}
as $\beta\rightarrow\infty$. \ Thus by these asymptotic expansions we conclude
the power of energy dissipation of these eigenvectors by the formula
(\ref{syop6}) have the asymptotic expansions%
\[
W_{\text{dis}}\left[  w_{j}\left(  \beta\right)  \right]  =-2\operatorname{Im}%
\zeta_{j}\left(  \beta\right)  U\left[  w_{j}\left(  \beta\right)  \right]
=\left\{
\begin{array}
[c]{cc}%
\mathring{\zeta}_{j}\beta+O\left(  1\right)  & \text{if }1\leq j\leq N_{B},\\
d_{j}\beta^{-1}+O\left(  \beta^{-2}\right)  & \text{if }N_{B}+1\leq j\leq N
\end{array}
\right.
\]
as $\beta\rightarrow\infty$. \ The proof of the proposition now follows from
these asymptotics.
\end{proof}

\begin{proof}
[Proof of Proposition \ref{Prop5}]By the formula (\ref{syop6}) for the quality
factor of an eigenvector of the system operator and by Proposition \ref{Prop1}
it follows for $\beta\gg1$ we have $Q\left[  w_{j}\left(  \beta\right)
\right]  $ is finite if and only if $\operatorname{Im}\zeta_{j}\left(
\beta\right)  \not \equiv 0$ in which case%
\[
Q\left[  w_{j}\left(  \beta\right)  \right]  =\frac{-1}{2}\frac{\left\vert
\operatorname{Re}\zeta_{j}\left(  \beta\right)  \right\vert }%
{\operatorname{Im}\zeta_{j}\left(  \beta\right)  }\text{.}%
\]
By Proposition \ref{Prop1} we know that $\operatorname{Re}\zeta_{j}\left(
\beta\right)  $ and $\operatorname{Im}\zeta_{j}\left(  \beta\right)  $ are
either analytic or have poles at $\beta=\infty$ whose Laurent series
expansions contain only even and odd powers of $\beta^{-1}$, respectively.
\ In particular, this implies either $\operatorname{Re}\zeta_{j}\left(
\beta\right)  \equiv0$ or the sign of $\operatorname{Re}\zeta_{j}\left(
\beta\right)  $ does not change for $\beta\gg1$ and $\left\vert
\operatorname{Re}\zeta_{j}\left(  \beta\right)  \right\vert $
$=\operatorname{Re}\zeta_{j}\left(  \beta\right)  $, for $\beta\gg1$ if this
sign is positive or if this sign is negative $\left\vert \operatorname{Re}%
\zeta_{j}\left(  \beta\right)  \right\vert $ $=-\operatorname{Re}\zeta
_{j}\left(  \beta\right)  $, for $\beta\gg1$. \ It follows from this and
Proposition \ref{Prop1} that we have the asymptotic expansions%
\[
\left\vert \operatorname{Re}\zeta_{j}\left(  \beta\right)  \right\vert
=\left\vert \rho_{j}\right\vert +O\left(  \beta^{-2}\right)  ,\text{
\ }\operatorname{Im}\zeta_{j}\left(  \beta\right)  =-\mathring{\zeta}_{j}%
\beta+O\left(  \beta^{-1}\right)  ,\text{$\quad$}1\leq j\leq N_{B};
\]%
\[
\left\vert \operatorname{Re}\zeta_{j}(\beta)\right\vert =\left\vert \rho
_{j}\right\vert +O\left(  \beta^{-2}\right)  ,\text{ }\operatorname{Im}%
\zeta_{j}\left(  \beta\right)  =-d_{j}\beta^{-1}+O\left(  \beta^{-3}\right)
,\text{ }N_{B}+1\leq j\leq N
\]
as $\beta\rightarrow\infty$. \ It also follows that if $\operatorname{Im}%
\zeta_{j}\left(  \beta\right)  \not \equiv 0$ for $\beta\gg1$ then $Q\left[
w_{j}\left(  \beta\right)  \right]  $ is the product of two meromorphic
functions one of which is even and one of which is odd in a neighborhood of
$\beta=\infty$ which implies it has a Laurent series expansion at
$\beta=\infty$ containing only odd powers of $\beta^{-1}$ and, in particular,
$\lim_{\beta\rightarrow\infty}Q\left[  w_{j}\left(  \beta\right)  \right]  =0$
or $\infty$. \ These facts and the fact $\mathring{\zeta}_{j}>0$ for $1\leq
j\leq N_{B}$ implies we have the asymptotic expansions
\[
Q\left[  w_{j}\left(  \beta\right)  \right]  =\frac{-1}{2}\frac{\left\vert
\operatorname{Re}\zeta_{j}\left(  \beta\right)  \right\vert }%
{\operatorname{Im}\zeta_{j}\left(  \beta\right)  }=\left\{
\begin{array}
[c]{c}%
\frac{\frac{1}{2}\left\vert \rho_{j}\right\vert }{\mathring{\zeta}_{j}}%
\beta^{-1}+O\left(  \beta^{-3}\right)  \text{ \ \ if }1\leq j\leq N_{B},\\
\frac{\frac{1}{2}\left\vert \rho_{j}\right\vert }{d_{j}}\beta+O\left(
\beta^{-1}\right)  \text{ \ \ if }d_{j}\not =0
\end{array}
\right.
\]
as $\beta\rightarrow\infty$. \ The proof of this propositions now follows from this.
\end{proof}

\subsubsection{The low-loss regime}

\begin{proof}
[Proof of Theorem \ref{Thm2}]Let $\omega_{j}$, $1\leq j\leq N$ be an indexing
of all the eigenvalues of $\Omega$ (counting multiplicities). \ The proof of
this theorem is similar in essence to the proof of Theorem \ref{Thm1}. \ We
begin by extending the domain of the system operator by $A\left(
\beta\right)  =\Omega-\mathrm{i}\beta B$, $\beta\in%
%TCIMACRO{\U{2102} }%
%BeginExpansion
\mathbb{C}
%EndExpansion
$. \ Using the substitution $\varepsilon=$ $-\mathrm{i}\beta$ then the
operator $A\left(  \mathrm{i}\varepsilon\right)  =\Omega+\varepsilon B$ is an
analytic operator which is self-adjoint for real $\varepsilon$. \ Thus by a
theorem of Rellich \cite[p. 21, Theorem 1]{Bau85} we know for $\varepsilon\in%
%TCIMACRO{\U{2102} }%
%BeginExpansion
\mathbb{C}
%EndExpansion
$ with $\left\vert \varepsilon\right\vert \ll1$, the operator $\Omega
+\varepsilon B$ is diagonalizable and there exists a complete set of analytic
eigenvalues $\lambda_{j}\left(  \varepsilon\right)  $ and eigenvectors
$x_{j}\left(  \varepsilon\right)  $ satisfying%
\begin{gather}
\left(  \Omega+\varepsilon B\right)  x_{j}\left(  \varepsilon\right)
=\lambda_{j}\left(  \varepsilon\right)  x_{j}\left(  \varepsilon\right)
,\text{ \ \ }\overline{\lambda_{j}\left(  \overline{\varepsilon}\right)
}=\lambda_{j}\left(  \varepsilon\right)  ,\text{ \ \ }1\leq j\leq
N;\label{pfllr1}\\
\left(  x_{j}\left(  \varepsilon\right)  ,x_{k}\left(  \varepsilon\right)
\right)  =\delta_{jk},\text{ \ \ for }\varepsilon\in%
%TCIMACRO{\U{211d} }%
%BeginExpansion
\mathbb{R}
%EndExpansion
,\text{ \ \ }1\leq j,k\leq N,\nonumber
\end{gather}
where $\delta_{jk}$ denotes the Kronecker delta symbol. \ In particular, the
vectors $u_{j}:=x_{j}\left(  0\right)  $, $j=1,\ldots,N$ form an orthonormal
basis of eigenvectors for the operator $\Omega$. \ Thus, after a possible
reindexing of these analytic eigenpairs, we may assume without loss of
generality that%
\[
\Omega u_{j}=\omega_{j}u_{j},\text{ \ \ }\lambda_{j}\left(  0\right)
=\omega_{j},\text{ \ \ }1\leq j\leq N.
\]
Denote the derivatives of these eigenvalues at $\varepsilon=0$ by $\sigma
_{j}:=\lambda_{j}^{\prime}\left(  0\right)  $, $1\leq j\leq N$. \ Then it
follows that they satisfy
\[
\sigma_{j}=\left(  u_{j},Bu_{j}\right)  \in%
%TCIMACRO{\U{211d} }%
%BeginExpansion
\mathbb{R}
%EndExpansion
,\ \ 1\leq j\leq N.
\]
Indeed, for $\varepsilon$ real with $\left\vert \varepsilon\right\vert \ll1$
by (\ref{pfllr1}) we have%
\[
0=\left(  x_{j}\left(  \varepsilon\right)  ,\left(  \Omega+\varepsilon
B-\lambda_{j}\left(  \varepsilon\right)  I\right)  x_{j}\left(  \varepsilon
\right)  \right)  ,\text{ \ \ }1\leq j\leq N
\]
and the result follows immediately by taking the derivative on both sides and
evaluating at $\varepsilon=0$.

Now from these facts and recalling the substitution $\varepsilon=$
$-\mathrm{i}\beta$ that was made, we conclude that if $\beta\in%
%TCIMACRO{\U{2102} }%
%BeginExpansion
\mathbb{C}
%EndExpansion
$ with $\left\vert \beta\right\vert \ll1$ then the system operator $A\left(
\beta\right)  =\Omega+\left(  -\mathrm{i}\beta\right)  B$ is diagonalizable
with a complete set of eigenvalues $\zeta_{j}\left(  \beta\right)  $ and
eigenvectors $\upsilon_{j}\left(  \beta\right)  $ satisfying%
\begin{gather}
\zeta_{j}\left(  \beta\right)  =\lambda_{j}\left(  -i\beta\right)  ,\text{
\ \ }\upsilon_{j}\left(  \beta\right)  =x_{j}\left(  -i\beta\right)
,\label{pfllr2}\\
A\left(  \beta\right)  \upsilon_{j}\left(  \beta\right)  =\zeta_{j}\left(
\beta\right)  \upsilon_{j}\left(  \beta\right)  ,\text{ \ \ }1\leq j\leq
N\text{.}\nonumber
\end{gather}
We will now show that these eigenpairs have the properties described in
Theorem \ref{Thm2}. \ First, it follows from our results above that the
eigenpairs $\lambda_{j}\left(  \varepsilon\right)  $, $x_{j}\left(
\varepsilon\right)  $, $1\leq j\leq N$ are analytic for\ $\varepsilon\in%
%TCIMACRO{\U{2102} }%
%BeginExpansion
\mathbb{C}
%EndExpansion
$ with $\left\vert \varepsilon\right\vert \ll1$ and have the asymptotic
expansions%
\begin{gather*}
\lambda_{j}\left(  \varepsilon\right)  =\omega_{j}+\sigma_{j}\varepsilon
+O\left(  \varepsilon^{2}\right)  ,\text{ \ \ }\omega_{j}\in%
%TCIMACRO{\U{211d} }%
%BeginExpansion
\mathbb{R}
%EndExpansion
,\text{ \ \ }\sigma_{j}\geq0,\\
x_{j}\left(  \varepsilon\right)  =u_{j}+O\left(  \varepsilon\right)  ,\text{
\ \ }1\leq j\leq N
\end{gather*}
as $\varepsilon\rightarrow0$. \ This implies by (\ref{pfllr2}) these
eigenvalues and eigenvectors are analytic at $\beta=0$ and they have the
asymptotic expansions
\begin{gather*}
\zeta_{j}\left(  \beta\right)  =\lambda_{j}\left(  -\mathrm{i}\beta\right)
=\omega_{j}-\mathrm{i}\sigma_{j}\beta+O\left(  \beta^{2}\right)  ,\text{
\ \ }\omega_{j}\in%
%TCIMACRO{\U{211d} }%
%BeginExpansion
\mathbb{R}
%EndExpansion
,\text{ \ \ }\sigma_{j}\geq0,\\
\upsilon_{j}\left(  \beta\right)  =x_{j}\left(  -\mathrm{i}\beta\right)
=u_{j}+O\left(  \beta\right)  ,\text{ \ \ }1\leq j\leq N
\end{gather*}
as $\beta\rightarrow0$, where the vectors $u_{j}$, $1\leq j\leq N_{B}$ form an
orthonormal basis of eigenvectors of $\Omega$ and
\[
\Omega u_{j}=\omega_{j}u_{j},\text{ \ \ }\sigma_{j}=\left(  u_{j},\Omega
u_{j}\right)  ,\text{ \ \ }1\leq j\leq N.
\]
This completes the proof.
\end{proof}

\begin{proof}
[Proof of Corollary \ref{Cor3}]This corollary follows immediately from the
asymptotic expansions (\ref{llre2}) of Theorem \ref{Thm2} and the formulas
from (\ref{syop6}).
\end{proof}

\begin{proof}
[Proof of Corollary \ref{Cor4}]The proof of this corollary is similar in
essence to the proof of Proposition \ref{Prop5}. \ As in the proof of Theorem
\ref{Thm2} we can extend the domain of the system operator $A\left(
\beta\right)  =\Omega-\mathrm{i}\beta B$, $\beta\in%
%TCIMACRO{\U{2102} }%
%BeginExpansion
\mathbb{C}
%EndExpansion
$ and the functions $\zeta_{j}\left(  \beta\right)  $, $1\leq j\leq N$ can be
extended uniquely to analytic functions in a neighborhood of $\beta=0$ whose
values are the eigenvalues of $A\left(  \beta\right)  $. \ Moreover, using the
same notation to denote their extensions, it follows from (\ref{pfllr1}) and
(\ref{pfllr2}) that these functions satisfy
\[
\overline{\zeta_{j}\left(  \overline{\beta}\right)  }=\zeta_{j}\left(
-\beta\right)  ,\text{ \ \ }1\leq j\leq N
\]
for all $\beta\in%
%TCIMACRO{\U{2102} }%
%BeginExpansion
\mathbb{C}
%EndExpansion
$ with $\left\vert \beta\right\vert \ll1$. \ This implies $\frac{1}{2}\left(
\zeta_{j}\left(  \beta\right)  +\zeta_{j}\left(  -\beta\right)  \right)  $ and
$\frac{1}{2i}\left(  \zeta_{j}\left(  \beta\right)  -\zeta_{j}\left(
-\beta\right)  \right)  $ are even and odd functions, respectively, are
analytic at $\beta=0$, and equal $\operatorname{Re}\zeta_{j}\left(
\beta\right)  $ and $\operatorname{Im}\zeta_{j}\left(  \beta\right)  $,
respectively, for real $\beta$. \ In particular, this implies their Taylor
series in $\beta$ have only even and odd powers, respectively. \ And this
implies as in the proof of Proposition \ref{Prop5} that $\left\vert
\operatorname{Re}\zeta_{j}\left(  \beta\right)  \right\vert $ is analytic at
$\beta=0$ and its Taylor series in $\beta$ has only even powers. \ Now by the
formula in (\ref{syop6}) we know that if $\operatorname{Im}\zeta_{j}\left(
\beta\right)  \equiv0$ for $\beta\in%
%TCIMACRO{\U{211d} }%
%BeginExpansion
\mathbb{R}
%EndExpansion
$ with $\left\vert \beta\right\vert \ll1$ then the quality factor is given by%
\[
Q\left[  \upsilon_{j}\left(  \beta\right)  \right]  =-\frac{1}{2}%
\frac{\left\vert \operatorname{Re}\zeta_{j}\left(  \beta\right)  \right\vert
}{\operatorname{Im}\zeta_{j}\left(  \beta\right)  },
\]
and hence can be extended to a function meromorphic at $\beta=0$ whose Laurent
series in $\beta$ contains only odd powers. \ Using these facts and the
asymptotic expansions in (\ref{llre2}) of Theorem \ref{Thm2} we conclude%
\begin{align*}
Q\left[  \upsilon_{j}\left(  \beta\right)  \right]   &  =-\frac{1}{2}%
\frac{\left\vert \operatorname{Re}\zeta_{j}\left(  \beta\right)  \right\vert
}{\operatorname{Im}\zeta_{j}\left(  \beta\right)  }=-\frac{1}{2}%
\frac{\left\vert \omega_{j}\right\vert +O\left(  \beta^{2}\right)  }%
{-\sigma_{j}\beta+O\left(  \beta^{3}\right)  }\\
&  =\frac{1}{2}\frac{\left\vert \omega_{j}\right\vert }{\sigma_{j}}\beta
^{-1}+O\left(  \beta\right)
\end{align*}
as $\beta\rightarrow\infty$, provided $\sigma_{j}\not =0$, in which case
$\lim_{\beta\rightarrow0}Q\left[  \upsilon_{j}\left(  \beta\right)  \right]
=0$ or $\infty$ depending on whether $\omega_{j}=0$ or $\omega_{j}\not =0$,
respectively. \ This completes the proof.
\end{proof}

\subsection{Perturbation analysis of a system subjected to harmonic forces}

This purpose of this section is to prove the results given in Section
\ref{spsyo}.

\begin{proof}
[Proof of Proposition \ref{Prop6}]Let $\omega\in%
%TCIMACRO{\U{211d} }%
%BeginExpansion
\mathbb{R}
%EndExpansion
$. \ If $\omega$ is an asymptotic nonresonance frequency of the system
(\ref{mvt7}) then Proposition \ref{Prop7}, which we have proved in this
section below, tells us the admittance operator $\mathfrak{A}\left(
\omega\right)  $ is analytic at $\beta=\infty$. \ So in order to complete the
proof of this proposition we need only prove the admittance operator
$\mathfrak{A}\left(  \omega\right)  $ cannot be analytic at $\beta=\infty$ if
$\omega$ is an asymptotic resonance frequency of the system (\ref{mvt7}).
\ Thus suppose $\omega=\rho_{k}$ for some $k\in\left\{  N_{B}+1,\ldots
,N\right\}  $ but $\mathfrak{A}\left(  \omega\right)  $ was analytic at
$\beta=\infty$. \ Then it is continuous and so by (\ref{hlre1}), (\ref{hlre5})
this implies
\[
0=\lim_{\beta\rightarrow\infty}\left[  -\mathrm{i}\left(  \omega I-\zeta
_{k}\left(  \beta\right)  \right)  \mathfrak{A}\left(  \omega\right)
w_{k}\left(  \beta\right)  \right]  =\lim_{\beta\rightarrow\infty}w_{k}\left(
\beta\right)  =\mathring{w}_{k}\not =0\text{,}%
\]
a contraction. \ This contradiction proves the proposition.
\end{proof}

For the rest of the proofs in this section the symbol $\omega$ whenever it
appears will mean an asymptotic nonresonance frequency of the system
(\ref{mvt7}), i.e., $\omega\in%
%TCIMACRO{\U{211d} }%
%BeginExpansion
\mathbb{R}
%EndExpansion
$ and $\omega\not =\rho_{j}$, $N_{B}+1\leq j\leq N$. \ Also, without loss of
generality we may extend the domain of the system operator $A\left(
\beta\right)  =\Omega-\mathrm{i}\beta B$, $\beta\in%
%TCIMACRO{\U{2102} }%
%BeginExpansion
\mathbb{C}
%EndExpansion
$. \ Similarly we extend the domains of the operators $\Xi_{2}\left(
\omega,\beta\right)  =\omega I_{2}-\left(  \Omega_{2}-\mathrm{i}\beta
B_{2}\right)  $, $\beta\in%
%TCIMACRO{\U{2102} }%
%BeginExpansion
\mathbb{C}
%EndExpansion
$ and $S_{2}\left(  \omega,\beta\right)  =\Xi_{1}\left(  \omega\right)
-\Theta^{\ast}\Xi_{2}\left(  \omega,\beta\right)  ^{-1}\Theta$, $\beta\in%
%TCIMACRO{\U{2102} }%
%BeginExpansion
\mathbb{C}
%EndExpansion
$ provided $\Xi_{2}\left(  \omega,\beta\right)  $ is invertible. \ For the
rest of the proofs in this section we use these extensions. \ Also whenever it
is convenient we will suppress the dependency of these operators on the
symbols $\omega$, $\beta$.

\begin{proof}
[Proof of Proposition \ref{Prop7}]We begin by proving for $\left\vert
\beta\right\vert \gg1$, the operators $\Xi_{1}\left(  \omega\right)  $,
$\Xi_{2}\left(  \omega,\beta\right)  $, $S_{2}\left(  \omega,\beta\right)  $
are invertible with $\Xi_{2}\left(  \omega,\beta\right)  ^{-1}$, $S_{2}\left(
\omega,\beta\right)  $, and $S_{2}\left(  \omega,\beta\right)  ^{-1}$ analytic
at $\beta=\infty$. \ By Proposition \ref{Prop3} it follows that the spectrum
of the operator $\Omega_{1}$ is the set $\left\{  \rho_{j}:N_{B}+1\leq j\leq
N\right\}  $ and since $\omega\not \in \left\{  \rho_{j}:N_{B}+1\leq j\leq
N\right\}  $ then the operator $\Xi_{1}\left(  \omega\right)  =\omega
I_{1}-\Omega_{1}$ is invertible. \ Now recall the well-known fact from
perturbation theory that if $T$ is an linear operator on $H$ such that in the
operator norm $\left\Vert T\right\Vert <1$ then $I-T$ is invertible and the
series $\sum_{n=0}^{\infty}T^{n}$ converges absolutely and uniformly to
$\left(  I-T\right)  ^{-1}$. \ This implies if $T\left(  \beta\right)  $ is an
operator-valued function which is analytic at $\beta=\infty$ then it has an
asymptotic expansion $T\left(  \beta\right)  =T_{0}+T_{1}\beta^{-1}+O\left(
\beta^{-2}\right)  $ as $\beta\rightarrow\infty$ and if $T_{0}$ is invertible
then $T\left(  \beta\right)  =T_{0}\left(  I-T_{0}^{-1}\left[  T\left(
\beta\right)  -T_{0}\right]  \right)  $ is invertible for $\left\vert
\beta\right\vert \gg1$ and its inverse $T\left(  \beta\right)  ^{-1}$ is
analytic at $\beta=\infty$ with the asymptotic expansion $T\left(
\beta\right)  ^{-1}=T_{0}^{-1}-T_{0}^{-1}T_{1}T_{0}^{-1}\beta^{-1}+O\left(
\beta^{-2}\right)  $ as $\beta\rightarrow\infty$. We will use some of these
facts now to prove $\Xi_{2}\left(  \omega,\beta\right)  $, $S_{2}\left(
\omega,\beta\right)  $ are invertible for $\left\vert \beta\right\vert \gg1$
and $\Xi_{2}\left(  \omega,\beta\right)  ^{-1}$, $S_{2}\left(  \omega
,\beta\right)  $, and $S_{2}\left(  \omega,\beta\right)  ^{-1}$are analytic at
$\beta=\infty$. \ Well, the function$\ T\left(  \beta\right)  =\mathrm{i}%
\beta^{-1}\Xi_{2}\left(  \omega,\beta\right)  $ is analytic at $\beta=\infty$
and its limit as $\beta\rightarrow\infty$ is an invertible operator since
\[
\lim_{\beta\rightarrow\infty}\mathrm{i}\beta^{-1}\Xi_{2}\left(  \zeta
,\beta\right)  =\lim_{\beta\rightarrow\infty}\mathrm{i}\beta^{-1}\left[  \zeta
I_{2}-\left(  \Omega_{2}-\mathrm{i}\beta B_{2}\right)  \right]  =B_{2}%
>0\text{.}%
\]
This implies $\Xi_{2}\left(  \omega,\beta\right)  =-\mathrm{i}\beta T\left(
\beta\right)  $ is invertible for $\left\vert \beta\right\vert \gg1$ with the
inverse $\Xi_{2}\left(  \omega,\beta\right)  ^{-1}=\mathrm{i}\beta
^{-1}T\left(  \beta\right)  ^{-1}$ analytic at $\beta=\infty$ having the
asymptotics
\begin{equation}
\Xi_{2}\left(  \omega,\beta\right)  ^{-1}=-\mathrm{i}B_{2}^{-1}\beta
^{-1}+B_{2}^{-1}\left(  \zeta I_{2}-\Omega_{2}\right)  B_{2}^{-1}\beta
^{-2}+O\left(  \beta^{-3}\right)  \label{pflhf1}%
\end{equation}
as $\beta\rightarrow\infty$. \ From which it follows that the operator%
\[
S_{2}\left(  \omega,\beta\right)  =\Xi_{1}\left(  \omega\right)  -\Theta
^{\ast}\Xi_{2}\left(  \omega,\beta\right)  ^{-1}\Theta=\Xi_{1}\left(
\omega\right)  \left(  I_{1}-\Xi_{1}\left(  \omega\right)  ^{-1}\Theta^{\ast
}\Xi_{2}\left(  \omega,\beta\right)  ^{-1}\Theta\right)
\]
is well-defined and invertible for $\left\vert \beta\right\vert \gg1$ as well
as it and its inverse, $S_{2}\left(  \omega,\beta\right)  ^{-1}$, are analytic
at $\beta=\infty$.

Now we prove $\omega I-A\left(  \beta\right)  $ is invertible for $\left\vert
\beta\right\vert \gg1$ and $\mathfrak{A}\left(  \omega\right)  =\mathrm{i}%
\left(  \omega I-A\left(  \beta\right)  \right)  ^{-1}$ is analytic at
$\beta=\infty$. \ First, by the $2\times2$ block operator matrix
representation of $\omega I-A\left(  \beta\right)  $ from (\ref{anrf1}) and
since $\Xi_{2}\left(  \omega,\beta\right)  $ is invertible for $\left\vert
\beta\right\vert \gg1$ then, as discussion in Appendix \ref{apxsc} on the
Atiken block diagonalization formula (\ref{mab1})--(\ref{mab3}), the operator
admits for $\left\vert \beta\right\vert \gg1$ the Frobenius-Schur
factorization%
\[
\omega I-A\left(  \beta\right)  =\left[
\begin{array}
[c]{cc}%
\Xi_{2} & -\Theta\\
-\Theta^{\ast} & \Xi_{1}%
\end{array}
\right]  =\left[
\begin{array}
[c]{cc}%
I_{2} & 0\\
-\Theta^{\ast}\Xi_{2}^{-1} & I_{1}%
\end{array}
\right]  \left[
\begin{array}
[c]{cc}%
\Xi_{2} & 0\\
0 & S_{2}%
\end{array}
\right]  \left[
\begin{array}
[c]{cc}%
I_{2} & -\Xi_{2}^{-1}\Theta\\
0 & I_{1}%
\end{array}
\right]  .
\]
Furthermore, this implies for $\beta\gg1$ that since $S_{2}\left(
\omega,\beta\right)  $ is invertible then $\omega I-A\left(  \beta\right)  $
is invertible and
\begin{gather}
\mathfrak{A}\left(  \omega\right)  =\mathrm{i}\left[
\begin{array}
[c]{cc}%
I_{2} & \Xi_{2}^{-1}\Theta\\
0 & I_{1}%
\end{array}
\right]  \left[
\begin{array}
[c]{cc}%
\Xi_{2}^{-1} & 0\\
0 & S_{2}^{-1}%
\end{array}
\right]  \left[
\begin{array}
[c]{cc}%
I_{2} & 0\\
\Theta^{\ast}\Xi_{2}^{-1} & I_{1}%
\end{array}
\right] \label{pflhf2}\\
=\mathrm{i}\left[
\begin{array}
[c]{cc}%
\Xi_{2}^{-1}+\Xi_{2}^{-1}\Theta S_{2}^{-1}\Theta^{\ast}\Xi_{2}^{-1} & \Xi
_{2}^{-1}\Theta S_{2}^{-1}\\
S_{2}^{-1}\Theta^{\ast}\Xi_{2}^{-1} & S_{2}^{-1}%
\end{array}
\right]  ,\nonumber
\end{gather}
which proves formula (\ref{adop1}). \ From this formula and the fact both
$\Xi_{2}\left(  \omega,\beta\right)  ^{-1}$ and $S_{2}\left(  \omega
,\beta\right)  ^{-1}$ are analytic at $\beta=\infty$ we conclude that
$\mathfrak{A}\left(  \omega\right)  $ is analytic at $\beta=\infty$.
\ Moreover, this formula, (\ref{mab3})--(\ref{mab6}) in Appendix \ref{apxsc},
that fact$\left(  \Xi_{1}\left(  \omega\right)  ^{-1}\right)  ^{\ast}=\Xi
_{1}\left(  \omega\right)  ^{-1}$, and (\ref{pflhf1}) imply we have the
asymptotic expansion%
\begin{gather}
\mathfrak{A}\left(  \omega\right)  =\mathrm{i}\left[
\begin{array}
[c]{cc}%
\Xi_{2}^{-1} & \Xi_{2}^{-1}\Theta\Xi_{1}^{-1}\\
\Xi_{1}^{-1}\Theta^{\ast}\Xi_{2}^{-1} & \Xi_{1}^{-1}+\Xi_{1}^{-1}\Theta^{\ast
}\Xi_{2}^{-1}\Theta\Xi_{1}^{-1}%
\end{array}
\right]  +O\left(  \beta^{-2}\right) \label{pflhf3}\\
=\left[
\begin{array}
[c]{cc}%
0 & 0\\
0 & \mathrm{i}\Xi_{1}^{-1}%
\end{array}
\right]  +W^{\left(  -1\right)  }\beta^{-1}+O\left(  \beta^{-2}\right)
,\text{ \ \ where}\nonumber\\
W^{\left(  -1\right)  }=\left[
\begin{array}
[c]{cc}%
B_{2}^{-1} & B_{2}^{-1}\Theta\Xi_{1}^{-1}\\
\left(  \Xi_{1}^{-1}\right)  ^{\ast}\Theta^{\ast}B_{2}^{-1} & \left(  \Xi
_{1}^{-1}\right)  ^{\ast}\Theta^{\ast}B_{2}^{-1}\Theta\Xi_{1}^{-1}%
\end{array}
\right] \nonumber\\
=\left[
\begin{array}
[c]{cc}%
I_{2} & 0\\
\left(  \Xi_{1}^{-1}\right)  ^{\ast}\Theta^{\ast} & I_{1}%
\end{array}
\right]  \left[
\begin{array}
[c]{cc}%
B_{2}^{-1} & 0\\
0 & 0
\end{array}
\right]  \left[
\begin{array}
[c]{cc}%
I_{2} & \Theta\Xi_{1}^{-1}\\
0 & I_{1}%
\end{array}
\right] \nonumber
\end{gather}
as $\beta\rightarrow\infty$. \ As $B_{2}^{-1}>0$ it follows immediately from
this block operator representation that $W^{\left(  -1\right)  }\geq0$. \ This
completes the proof of the proposition.
\end{proof}

\begin{proof}
[Proof of Corollary \ref{Cor5}]As the the operator $\mathfrak{A}\left(
\omega\right)  =\mathrm{i}\left(  \omega I-A\left(  \beta\right)  \right)
^{-1}$ is analytic at $\beta=\infty$ and since $A\left(  \beta\right)  ^{\ast
}=A\left(  -\beta\right)  $ for $\beta$ real then $\mathfrak{A}\left(
\omega\right)  ^{\ast}=-\mathrm{i}\left(  \omega I-A\left(  -\beta\right)
\right)  ^{-1}$ for $\beta$ real which is clearly analytic at $\beta=\infty$.
\ This implies each of the operators $\mathfrak{A}\left(  \omega\right)
^{\ast}\mathfrak{A}\left(  \omega\right)  $, $P_{B}\mathfrak{A}\left(
\omega\right)  ^{\ast}\mathfrak{A}\left(  \omega\right)  P_{B}$, and
$\mathfrak{A}\left(  \omega\right)  ^{\ast}B\mathfrak{A}\left(  \omega\right)
$ is analytic at $\beta=\infty$. \ The fact $\mathfrak{A}\left(
\omega\right)  ^{\ast}\beta B\mathfrak{A}\left(  \omega\right)  $ is analytic
at $\beta=\infty$ follows immediately from our proof below that $\lim
_{\beta\rightarrow\infty}\mathfrak{A}\left(  \omega\right)  ^{\ast
}B\mathfrak{A}\left(  \omega\right)  =0$. \ Now the asymptotic expansion of
$\mathfrak{A}\left(  \omega\right)  ^{\ast}\mathfrak{A}\left(  \omega\right)
$ as $\beta\rightarrow\infty$ in (\ref{adop4}) follows immediately from the
asymptotic expansion of $\mathfrak{A}\left(  \omega\right)  $ in (\ref{adop2})
and the definition of the real part of an operator $T$ as $\operatorname{Re}%
T=\frac{1}{2}\left(  T+T^{\ast}\right)  $. \ The proofs of the asymptotic
expansions of $P_{B}\mathfrak{A}\left(  \omega\right)  ^{\ast}\mathfrak{A}%
\left(  \omega\right)  P_{B}$ and $\mathfrak{A}\left(  \omega\right)  ^{\ast
}\beta B\mathfrak{A}\left(  \omega\right)  $ in (\ref{adop5}) and
(\ref{adop6}) are similiar, using the asymptotic expansion (\ref{adop2}) for
$\mathfrak{A}\left(  \omega\right)  $ and the fact that $P_{B}$ and
$B=P_{B}BP_{B}$, with respect to the direct sum $H=H_{B}\oplus H_{B}^{\bot}$,
are the block operators%
\begin{equation}
P_{B}=\left[
\begin{array}
[c]{cc}%
I_{2} & 0\\
0 & 0
\end{array}
\right]  ,\text{ \ \ }B=\left[
\begin{array}
[c]{cc}%
B_{2} & 0\\
0 & 0
\end{array}
\right]  \text{.} \label{pflhf4}%
\end{equation}
It follows from (\ref{adop3}) and (\ref{adop2}) that%
\begin{gather*}
P_{B}\left(  W^{\left(  -1\right)  }\right)  ^{\ast}W^{\left(  -1\right)
}P_{B}= \left[
\begin{array}
[c]{cc}%
B_{2}^{-2}+B_{2}^{-1}\Theta\left(  \Xi_{1}^{-1}\right)  ^{\ast}\Xi_{1}%
^{-1}\Theta^{\ast}B_{2}^{-1} & 0\\
0 & 0
\end{array}
\right]  ,\\
\left(  W^{\left(  -1\right)  }\right)  ^{\ast}BW^{\left(  -1\right)
}=W^{\left(  -1\right)  },\\
P_{B}\mathfrak{A}\left(  \omega\right)  =P_{B}W^{\left(  -1\right)  }%
\beta^{-1}+O\left(  \beta^{-2}\right)  ,\text{ \ \ }\mathfrak{A}\left(
\omega\right)  P_{B}=W^{\left(  -1\right)  }P_{B}\beta^{-1}+O\left(
\beta^{-2}\right)  ,\\
P_{B}\mathfrak{A}\left(  \omega\right)  ^{\ast}=P_{B}\left(  W^{\left(
-1\right)  }\right)  ^{\ast}\beta^{-1}+O\left(  \beta^{-2}\right)  ,\text{
\ \ }\mathfrak{A}\left(  \omega\right)  ^{\ast}P_{B}=\left(  W^{\left(
-1\right)  }\right)  ^{\ast}P_{B}\beta^{-1}+O\left(  \beta^{-2}\right)
\end{gather*}
as $\beta\rightarrow\infty$. \ Thus from these facts it follows that%
\begin{align*}
P_{B}\mathfrak{A}\left(  \omega\right)  ^{\ast}\mathfrak{A}\left(
\omega\right)  P_{B}  &  =P_{B}\left(  W^{\left(  -1\right)  }\right)  ^{\ast
}W^{\left(  -1\right)  }P_{B}\beta^{-2}+O\left(  \beta^{-3}\right) \\
\mathfrak{A}\left(  \omega\right)  ^{\ast}B\mathfrak{A}\left(  \omega\right)
&  =\mathfrak{A}\left(  \omega\right)  ^{\ast}P_{B}BP_{B}\mathfrak{A}\left(
\omega\right)  =\left(  W^{\left(  -1\right)  }\right)  ^{\ast}P_{B}%
BP_{B}W^{\left(  -1\right)  }\beta^{-2}+O\left(  \beta^{-3}\right) \\
&  =W^{\left(  -1\right)  }\beta^{-2}+O\left(  \beta^{-3}\right)
\end{align*}
as $\beta\rightarrow\infty$. \ Therefore (\ref{adop5}) is proven and from the
fact $\mathfrak{A}\left(  \omega\right)  ^{\ast}B\mathfrak{A}\left(
\omega\right)  $ is analytic at $\beta=\infty$ with the asymptotic for it just
derived, it follow that $\mathfrak{A}\left(  \omega\right)  ^{\ast}\beta
B\mathfrak{A}\left(  \omega\right)  $ is analytic at $\beta=\infty$ and has
the asymptotic expansion (\ref{adop6}). \ This completes the proof.
\end{proof}

Before we now proceed to prove Theorems \ref{Thm3} and \ref{Thm4} we will find
it convenient to first prove Corollary \ref{Cor6}.

\begin{proof}
[Proof of Corollary \ref{Cor6}]From the block operator factorization of the
operator $W^{(-1)}$ in (\ref{adop3}) with respect to the direct sum
$H=H_{B}\oplus H_{B}^{\bot}$ it follows immediately that $W^{(-1)}f=0$ if and
only if $P_{B}f+\Theta\Xi_{1}\left(  \omega\right)  ^{-1}P_{B}^{\bot}f=0$,
i.e., $f\in\ker\left(  P_{B}+\Theta\Xi_{1}\left(  \omega\right)  ^{-1}%
P_{B}^{\bot}\right)  $. \ But from this direct sum and the fact $P_{B}$,
$P_{B}^{\bot}$ are the orthogonal projections onto $H_{B}$, $H_{B}^{\bot}$,
respectively, we can compute this kernel to conclude
\begin{gather*}
\ker W^{(-1)}=\ker\left(  P_{B}+\Theta\Xi_{1}\left(  \omega\right)  ^{-1}%
P_{B}^{\bot}\right) \\
=\left\{  f_{1}+f_{2}\in H:f_{1}\in H_{B}^{\bot}\text{ and }f_{2}=-\Theta
\Xi_{1}\left(  \omega\right)  ^{-1}f_{1}\right\}  \text{.}%
\end{gather*}
This representation of the kernel, the direct sum, and the fact that $\dim
H_{B}^{\bot}=N-N_{B}$ implies the kernel has dimension $N-N_{B}$. \ This
completes the proof.
\end{proof}

\begin{proof}
[Proof of Theorems \ref{Thm3} \& \ref{Thm4}]Let $f\in H$, $f\not =0$. \ The
stored energy $U$ and power of dissipated energy $W_{\text{dis}}$ associated
with the harmonic external force $f\left(  t\right)  =fe^{-\mathrm{i}\omega
t}$ by (\ref{palhf1}) and (\ref{palhf2}) are given by the formulas%
\[
U=\frac{1}{2}\left(  f,\mathfrak{A}\left(  \omega\right)  ^{\ast}%
\mathfrak{A}\left(  \omega\right)  f\right)  ,\text{ \ \ }W_{\text{dis}%
}=\left(  f,\mathfrak{A}\left(  \omega\right)  ^{\ast}\beta B\mathfrak{A}%
\left(  \omega\right)  f\right)
\]
for $\beta\gg1$. \ In particular, by these formulas and Corollary \ref{Cor5}
it follows that $U$ and $W_{\text{dis}}$ are analytic at $\beta=\infty$ and by
(\ref{adop4}), (\ref{adop6}) have the asymptotic expansions%
\begin{gather*}
U=\frac{1}{2}\left(  \Xi_{1}\left(  \omega\right)  ^{-1}P_{B}^{\bot}f,\Xi
_{1}\left(  \omega\right)  ^{-1}P_{B}^{\bot}f\right)  +O\left(  \beta
^{-1}\right) \\
W_{\text{\textrm{dis}}}=\left(  f,W^{\left(  -1\right)  }f\right)  \beta
^{-1}+O\left(  \beta^{-2}\right)
\end{gather*}
as $\beta\rightarrow\infty$. \ In particular, by the fact $W^{\left(
-1\right)  }\geq0$, the leading order term for $W_{\text{\textrm{dis}}}$ is a
nonnegative number and if $P_{B}^{\bot}f\not =0$ then the leading order term
for $U$ is a positive number.

Now since $W_{\text{\textrm{dis}}}$ is analytic at $\beta=\infty$ then either
$W_{\text{\textrm{dis}}}\equiv0$ for $\beta\gg0$ or $W_{\text{\textrm{dis}}%
}\not =0$ for $\beta\gg1$. \ Hence by the definition in (\ref{palhf1}) and
(\ref{palhf2}) of the quality factor $Q=Q_{f,\omega}$ it will be finite for
$\beta\gg1$ if and only if $W_{\text{\textrm{dis}}}\not \equiv 0$ for
$\beta\gg0$, in which case it is given by the formula $Q=\left\vert
\omega\right\vert U/W_{\text{\textrm{dis}}}$ implying it is a meromorphic
function at $\beta=\infty$. \ For example, by the asymptotic expansions just
derived it follows that if $P_{B}^{\bot}f\not =0$ then it must have a pole and
if $P_{B}^{\bot}f=0$ then it must be analytic. \ In particular, if we have
$\left(  f,W^{\left(  -1\right)  }f\right)  \not =0$, which is equivalent to
$f\in\ker W^{\left(  -1\right)  }$, then it has the asymptotic expansion%
\[
Q=\left\vert \omega\right\vert \frac{\frac{1}{2}\left(  \Xi_{1}\left(
\omega\right)  ^{-1}P_{B}^{\bot}f,\Xi_{1}\left(  \omega\right)  ^{-1}%
P_{B}^{\bot}f\right)  }{\left(  f,W^{\left(  -1\right)  }f\right)  }%
\beta+O\left(  1\right)
\]
as $\beta\rightarrow\infty$ whose leading order term is nonnegative and if
$P_{B}^{\bot}f\not =0$ then it is positive.

Now we complete the proof of Theorem \ref{Thm3}. \ Suppose $P_{B}^{\bot}f=0$.
\ Then we have $P_{B}f=f$. \ Then it follows from\ this and the block operator
representation for $W^{\left(  -1\right)  }$ and $P_{B}$ in (\ref{adop3}) and
(\ref{pflhf4}), respectively, that $P_{B}W^{\left(  -1\right)  }P_{B}%
f=B_{2}^{-1}f$. \ In particular, since $B_{2}^{-1}>0$ then $\left(
f,W^{\left(  -1\right)  }f\right)  =\left(  f,B_{2}^{-1}f\right)  >0$ and
hence from the statements in this proof above the quality factor $Q$ is finite
for $\beta\gg1$ and is analytic at $\beta=\infty$. \ From our discussion above
and (\ref{adop5}) we have the asymptotic expansions%
\begin{gather*}
W_{\text{\textrm{dis}}}=\left(  f,B_{2}^{-1}f\right)  \beta^{-1}+O\left(
\beta^{-2}\right)  ,\\
U=\frac{1}{2}\left(  f,\mathfrak{A}\left(  \omega\right)  ^{\ast}%
\mathfrak{A}\left(  \omega\right)  f\right)  =\frac{1}{2}\left(
f,P_{B}\mathfrak{A}\left(  \omega\right)  ^{\ast}\mathfrak{A}\left(
\omega\right)  P_{B}f\right) \\
=\frac{1}{2}\left(  f,\left[  B_{2}^{-2}+B_{2}^{-1}\Theta\left(  \Xi_{1}%
^{-1}\right)  ^{\ast}\Xi_{1}^{-1}\Theta^{\ast}B_{2}^{-1}\right]  f\right)
\beta^{-2}+O\left(  \beta^{-3}\right) \\
Q=\left\vert \omega\right\vert \frac{\frac{1}{2}\left(  f,\left[  B_{2}%
^{-2}+B_{2}^{-1}\Theta\left(  \Xi_{1}^{-1}\right)  ^{\ast}\Xi_{1}^{-1}%
\Theta^{\ast}B_{2}^{-1}\right]  f\right)  }{\left(  f,B_{2}^{-1}f\right)
}\beta^{-1}+O\left(  \beta^{-2}\right)
\end{gather*}
as $\beta\rightarrow\infty$. \ Thus to complete the proof of Theorem
\ref{Thm3} we need only prove the inequalities described in that theorem.
\ First, since $B_{2}^{-1}>0$ then it has a positive square root
$B_{2}^{-\frac{1}{2}}>0$. \ Second, it follows from Theorem \ref{Thm1} and
(\ref{hlre10}) that for any $u\in H_{B}$ we have%
\begin{gather*}
\inf_{u\in H_{B},u\not =0}\frac{\left(  u,B_{2}^{-1}u\right)  }{\left(
u,u\right)  }=\inf_{u\in H_{B},u\not =0}\frac{\sum_{1\leq j\leq N_{B}%
}\mathring{\zeta}_{j}^{-1}\left\vert \left(  \mathring{w}_{j},u\right)
\right\vert ^{2}}{\sum_{1\leq j\leq N_{B}}\left\vert \left(  \mathring{w}%
_{j},u\right)  \right\vert ^{2}}\\
\geq\min_{1\leq j\leq N_{B}}\mathring{\zeta}_{j}^{-1}=\left(  \max_{1\leq
j\leq N_{B}}\mathring{\zeta}_{j}\right)  ^{-1}.
\end{gather*}
Thus imply, with $u=f$, $B_{2}^{-\frac{1}{2}}f$, the inequalities%
\begin{align*}
\frac{1}{2}\left(  f,B_{2}^{-1}f\right)   &  \geq\left(  \max_{1\leq j\leq
N_{B}}\mathring{\zeta}_{j}\right)  ^{-1}\left(  f,f\right)  \text{,}\\
\frac{1}{2}\left(  f,B_{2}^{-2}f\right)   &  =\left(  \max_{1\leq j\leq N_{B}%
}\mathring{\zeta}_{j}\right)  ^{-1}\left(  f,B_{2}^{-1}f\right)
\end{align*}
These facts imply, since $B_{2}^{-1}\Theta\left(  \Xi_{1}^{-1}\right)  ^{\ast
}\Xi_{1}^{-1}\Theta^{\ast}B_{2}^{-1}\geq0$, that
\begin{gather*}
\frac{1}{2}\left(  f,\left[  B_{2}^{-2}+B_{2}^{-1}\Theta\left(  \Xi_{1}%
^{-1}\right)  ^{\ast}\Xi_{1}^{-1}\Theta^{\ast}B_{2}^{-1}\right]  f\right)
\geq\frac{1}{2}\left(  \max_{1\leq j\leq N_{B}}\mathring{\zeta}_{j}\right)
^{-1}\left(  f,B_{2}^{-1}f\right) \\
\geq\frac{1}{2}\left(  \max_{1\leq j\leq N_{B}}\mathring{\zeta}_{j}\right)
^{-2}\left(  f,f\right)  >0\text{.}%
\end{gather*}
This completes the proof of the theorems.
\end{proof}

\begin{proof}
[Proof of Corollary \ref{Cor7}]The proof of this corollary follows immediately
from Theorem \ref{Thm3} and Theorem \ref{Thm4}.
\end{proof}

\section{Appendix: Schur complement and the Aitken formula\label{apxsc}}

Let $M$
\begin{equation}
M=\left[
\begin{array}
[c]{cc}%
P & Q\\
R & S
\end{array}
\right]  \label{mab1}%
\end{equation}
be a square matrix represented in block form where $P$ and $S$ are square
matrices with the former invertible, that is, $\left\Vert P^{-1}\right\Vert
<\infty$. Then the following \emph{Aitken block-diagonalization formula} holds
\cite[Sec. 0.9, 1.1]{Zhang}, \cite[p. 67 (4)]{Aitken},%
\begin{equation}
M=\left[
\begin{array}
[c]{cc}%
P & Q\\
R & S
\end{array}
\right]  =\left[
\begin{array}
[c]{cc}%
I & 0\\
RP^{-1} & I
\end{array}
\right]  \left[
\begin{array}
[c]{cc}%
P & 0\\
0 & S_{P}%
\end{array}
\right]  \left[
\begin{array}
[c]{cc}%
I & P^{-1}Q\\
0 & I
\end{array}
\right]  , \label{mab2}%
\end{equation}
i.e., the Frobenius-Schur factorization of the block matrix $M$ \cite[p.
xiv]{Tretter}, where the matrix%
\begin{equation}
S_{P}=M/P=S-RP^{-1}Q \label{mab3}%
\end{equation}
is known as the \emph{Schur complement} of $P$ in $M$. The Aitken formula
(\ref{mab2}) readily implies%
\begin{gather}
M^{-1}=\left[
\begin{array}
[c]{cc}%
I & -P^{-1}Q\\
0 & I
\end{array}
\right]  \left[
\begin{array}
[c]{cc}%
P^{-1} & 0\\
0 & S_{P}^{-1}%
\end{array}
\right]  \left[
\begin{array}
[c]{cc}%
I & 0\\
-RP^{-1} & I
\end{array}
\right]  =\label{mab4}\\
=\left[
\begin{array}
[c]{cc}%
P^{-1}+P^{-1}QS_{P}^{-1}RP^{-1} & -P^{-1}QS_{P}^{-1}\\
-S_{P}^{-1}RP^{-1} & S_{P}^{-1}%
\end{array}
\right]  .\nonumber
\end{gather}
In particular, for $\left\Vert P^{-1}\right\Vert \ll1$ and under\ the
assumption of invertibility of the matrix $S$, formulas (\ref{mab3}) and
(\ref{mab4}) imply%
\begin{gather}
S_{P}^{-1}=\left[  S-RP^{-1}Q\right]  ^{-1}=S^{-1}\left[  I-RP^{-1}%
QS^{-1}\right]  ^{-1}\label{mab5}\\
=S^{-1}+S^{-1}RP^{-1}QS^{-1}+O\left(  \left\Vert P^{-2}\right\Vert \right)
,\nonumber
\end{gather}%
\[
M^{-1}=\left[
\begin{array}
[c]{cc}%
0 & 0\\
0 & S^{-1}%
\end{array}
\right]  +O\left(  \left\Vert P^{-1}\right\Vert \right)  ,
\]%
\begin{gather}
M^{-1}=\left[  M^{-1}\right]  _{1}+O\left(  \left\Vert P^{-2}\right\Vert
\right)  ,\text{ where}\label{mab6}\\
\left[  M^{-1}\right]  _{1}=\left[
\begin{array}
[c]{cc}%
P^{-1} & -P^{-1}QS^{-1}\\
-S^{-1}RP^{-1} & S^{-1}+S^{-1}RP^{-1}QS^{-1}%
\end{array}
\right]  =\nonumber\\
=\left[
\begin{array}
[c]{cc}%
P^{-1} & 0\\
0 & S^{-1}+S^{-1}RP^{-1}QS^{-1}%
\end{array}
\right]  +\left[
\begin{array}
[c]{cc}%
P^{-1} & 0\\
0 & S^{-1}%
\end{array}
\right]  \left[
\begin{array}
[c]{cc}%
0 & -Q\\
-R & 0
\end{array}
\right]  \left[
\begin{array}
[c]{cc}%
P^{-1} & 0\\
0 & S^{-1}%
\end{array}
\right]  .\nonumber
\end{gather}
\qquad

\section{Appendix: Quality factor for eigenmodes\label{apxqf}}

In this appendix we derive a simple and descriptive formula for the energy
$U\left[  w\right]  $, power of energy dissipation $W_{\text{\textrm{dis}}%
}\left[  w\right]  $, and quality factor $Q\left[  w\right]  $, for any
eigenmode $w$ of the system operator $A\left(  \beta\right)  $ with eigenvalue
$\zeta$ which are the quantities%
\begin{equation}
U\left[  w\right]  =\frac{1}{2}\left(  w,w\right)
\end{equation}%
\begin{equation}
W_{\text{\textrm{dis}}}\left[  w\right]  =\left(  w,\beta Bw\right)
\end{equation}%
\begin{gather}
Q\left[  w\right]  =2\pi\frac{\text{energy stored in the system}}{\text{energy
lost per cycle}}\\
=\left\vert \operatorname{Re}\zeta\right\vert \frac{U\left[  w\right]
}{W_{\text{\textrm{dis}}}\left[  w\right]  }=\left\vert \operatorname{Re}%
\zeta\right\vert \frac{\frac{1}{2}\left(  w,w\right)  }{\beta\left(
w,Bw\right)  }\text{.}\nonumber
\end{gather}
where $\operatorname{Re}\zeta$ denotes the real part of the eigenvalue $\zeta$
and $Q\left[  w\right]  $ is finite if $W_{\text{\textrm{dis}}}\left[
w\right]  \not =0$.

\begin{lemma}
\label{apxlm}If $w$ is an eigenvector of the system operator $A\left(
\beta\right)  =\Omega-\mathrm{i}\beta B$ with eigenvalue $\zeta$ then%
\begin{equation}
\operatorname{Re}\zeta=\frac{\left(  w,\Omega w\right)  }{\left(  w,w\right)
},\quad\operatorname{Im}\zeta=-\frac{\left(  w,\beta Bw\right)  }{\left(
w,w\right)  }.
\end{equation}
In particular, if $\beta\geq0$ then $\operatorname{Im}\zeta\leq0$.
\end{lemma}

\begin{proof}
By assumption $B\geq0$ and by hypothesis we have $A\left(  \beta\right)
w=\zeta w$, $w\not =0$ so that $\left(  w,w\right)  \not =0$ and%
\[
\frac{\operatorname{Re}\left(  w,A\left(  \beta\right)  w\right)  }{\left(
w,w\right)  }=\operatorname{Re}\zeta,\text{ \ }\frac{\text{\ }%
\operatorname{Im}\left(  w,A\left(  \beta\right)  w\right)  }{\left(
w,w\right)  }=\operatorname{Im}\zeta.
\]
On the other hand, since%
\[
\frac{1}{2}\left[  A\left(  \beta\right)  +A\left(  \beta\right)  ^{\ast
}\right]  =\Omega,\text{ \ \ }\frac{1}{2\mathrm{i}}\left[  A\left(
\beta\right)  -A\left(  \beta\right)  ^{\ast}\right]  =-\beta B\text{,}%
\]
this implies
\begin{gather*}
\operatorname{Re}\zeta=\frac{\operatorname{Re}\left(  w,A\left(  \beta\right)
w\right)  }{\left(  w,w\right)  }=\frac{\left(  w,\Omega w\right)  }{\left(
w,w\right)  },\\
\operatorname{Im}\zeta=\frac{\operatorname{Im}\left(  w,A\left(  \beta\right)
w\right)  }{\left(  w,w\right)  }=-\frac{\left(  w,\beta Bw\right)  }{\left(
w,w\right)  }\leq0\text{.}%
\end{gather*}
This completes the proof.
\end{proof}

\begin{proposition}
[quality factor]\label{apxpr}If $w$ is an eigenvector of the system operator
$A\left(  \beta\right)  =\Omega-\mathrm{i}\beta B$ with eigenvalue $\zeta$
then the energy $U\left[  w\right]  $, power of energy dissipation
$W_{\text{\textrm{dis}}}\left[  w\right]  $, and quality factor $Q\left[
w\right]  $ satisfy%
\begin{equation}
\operatorname{Re}\zeta U\left[  w\right]  =\frac{1}{2}\left(  w,\Omega
w\right)  ,\quad W_{\text{\textrm{dis}}}\left[  w\right]  =\left(  w,\beta
Bw\right)  =-2\operatorname{Im}\zeta U\left[  w\right]  ,
\end{equation}%
\begin{equation}
Q\left[  w\right]  =\frac{1}{2}\frac{\left\vert \left(  w,\Omega w\right)
\right\vert }{\left(  w,\beta Bw\right)  }=-\frac{1}{2}\frac{\left\vert
\operatorname{Re}\zeta\right\vert }{\operatorname{Im}\zeta},
\end{equation}
and $Q\left[  w\right]  $ is finite if and only if $\operatorname{Im}%
\zeta\not =0$. \ In particular, the quality factor $Q\left[  w\right]
=-\frac{1}{2}\frac{\left\vert \operatorname{Re}\zeta\right\vert }%
{\operatorname{Im}\zeta}$ is independent of the choice of eigenvector $w$
since it depends only on its corresponding eigenvalue $\zeta$.
\end{proposition}

\begin{proof}
By the lemma it follows that $W_{\text{\textrm{dis}}}\left[  w\right]
=\left(  w,\beta Bw\right)  =-\left(  w,w\right)  \operatorname{Im}%
\zeta=-2\operatorname{Im}\zeta U\left[  w\right]  $ and $\operatorname{Re}%
\zeta U\left[  w\right]  =\frac{1}{2}\left(  w,w\right)  \operatorname{Re}%
\zeta=\frac{1}{2}\left(  w,\Omega w\right)  $. \ This implies that $Q\left[
w\right]  $ is finite if and only if $\operatorname{Im}\zeta\not =0$, in which
case
\[
Q\left[  w\right]  =\left\vert \operatorname{Re}\zeta\right\vert
\frac{U\left[  w\right]  }{W_{\text{\textrm{dis}}}\left[  w\right]  }=\frac
{1}{2}\frac{\left\vert \left(  w,\Omega w\right)  \right\vert }{\left(
w,\beta Bw\right)  }=\frac{1}{2}\frac{\left\vert \operatorname{Re}\zeta\left(
w,w\right)  \right\vert }{-\operatorname{Im}\zeta\left(  w,w\right)  }%
=-\frac{1}{2}\frac{\left\vert \operatorname{Re}\zeta\right\vert }%
{\operatorname{Im}\zeta}\text{.}%
\]
This completes the proof.
\end{proof}

\textbf{Acknowledgment:} The research of A. Figotin was supported through Dr.
A. Nachman of the U.S. Air Force Office of Scientific Research (AFOSR), under
grant number FA9550-11-1-0163. \ Both authors are indebted to the referees for
their valuable comments on our original manuscript.


\begin{thebibliography}{9999999}                                                                                          %


\bibitem[Aitken]{Aitken}Aitken A., \textsl{Determinants and Matrices}, 3d ed.,
Oliver and Boyd, 1944.

\bibitem[Bau85]{Bau85}Baumgartel, H., \textsl{Analytic Perturbation Theory for
Matrices and Operators}, Birkh\"{a}user Verlag, Basel 1985.

\bibitem[BenGre]{BenGre}Ben-Israel, A. and Greville, T., \textsl{Generalized
Inverses: Theory and Applications}, 2nd ed., Springer-Verlag, 2003.

\bibitem[FigSch1]{FigSch1}A. Figotin and J. H. Schenker, \textsl{Spectral
theory of time dispersive and dissipative systems}, \textit{J. Stat. Phys.}
\textbf{118}, 199--262, (2005).

\bibitem[FigSch2]{FigSch2}A. Figotin and J. H. Schenker, \textsl{Hamiltonian
stucture for dispersive and dissipative dynamical systems}, J. Stat. Phys.,
\textbf{128}, 969-1056, (2007).

\bibitem[FigShi]{FigShi}Figotin A. and Shipman S., \textsl{Open Systems Viewed
Through Their Conservative Extensions}, J. Stat. Phys., \textbf{125}, 363-413, (2006).

\bibitem[FigVit8]{FigVit8}Figotin A. and Vitebskiy I., \textsl{Absorption
suppression in photonic crystals}, Phys. Rev. B, \textbf{77}, 104421 (2008).

\bibitem[Gant]{Gantmacher}Gantmacher F., \textsl{Lectures in Analytical
Mechanics}, Mir, 1975.

\bibitem[Gold]{Goldstein}Goldstein H., Poole C., Safko J., \textsl{Classical
Mechanics}, 3rd ed., Addison-Wesley, 2000.

\bibitem[KubTod2]{KubTod2}R. Kubo, M. Toda and N. Hashitsume,
\textsl{Statistical Physics II, Nonequilibrium Stastical Mechanics}, 2nd
edition, Springer-Verlag, Berlin, 1991.

\bibitem[LanTis]{LanTis}Lancaster P. and Tismenetsky M., \textsl{The Theory of
Matrices}, 2nd ed., Academic Press, 1985.

\bibitem[Liv]{Liv}Livsic M., \textsl{Operators, oscillations, waves, open
systems}, AMS, 1973

\bibitem[Pain]{Pain}Pain J., \textsl{The Physics of Vibrations and Waves}, 6th
ed., Wiley, 2004.

\bibitem[Pars]{Pars}Pars L., \textsl{Treatise on Analytical Mechanics},
Heinemann, 1965.

\bibitem[Tret08]{Tretter}Tretter C., \textsl{Spectral theory of block operator
matrices and applications}, Imperial College Press, 2008.

\bibitem[Welt]{Welters}Welters, A., \textsl{On explicit recursive formulas in
the spectral perturbation analysis of a Jordan block}, SIAM J. Matrix Anal.
Appl., 32, 1, 1--22, (2011).

\bibitem[Zhang]{Zhang}Zhang F., \textsl{Schur Complement and Its
Applications}, Springer, 2005.
\end{thebibliography}
\end{document}